\pdfoutput=1
\PassOptionsToPackage{table,dvipsnames}{xcolor}
\documentclass[a4paper,twocolumn,nopdfoutputerror]{quantumarticle}

\usepackage[utf8]{inputenc}
\usepackage[T1]{fontenc}
\usepackage[british]{babel}

\usepackage[babel]{microtype}

\usepackage{amsmath,amssymb,amsthm,amsfonts,bm,mathrsfs,mathtools,bbm}

\usepackage{centernot}

\usepackage{graphicx}
\usepackage{subcaption}
\usepackage{float}

\usepackage{booktabs}
\usepackage{tabularx}
\usepackage{multirow}
\usepackage{array}
\usepackage{bigstrut}
\usepackage{makecell}

\usepackage{tikz}
\usetikzlibrary{positioning,arrows.meta,fit}

\usepackage{algorithm}
\usepackage{algorithmic}

\usepackage{enumitem}
\usepackage{comment}
\usepackage{xspace}
\usepackage{ragged2e}

\DeclareCaptionJustification{justified}{\justifying}

\usepackage{hyperref}
\hypersetup{colorlinks=true,linkcolor=Blue,citecolor=Blue,urlcolor=Blue}

\newcommand{\be}{\begin{equation}}
\newcommand{\ee}{\end{equation}}

\newcommand{\bea}{\begin{eqnarray}}
\newcommand{\eea}{\end{eqnarray}}

\newcommand{\Tr}{\operatorname{Tr}}

\newcommand{\rank}{\operatorname{rank}}

\newcommand{\id}{\operatorname{id}}

\newcommand{\F}{\mathbb{F}_2}

\newcommand{\ket}[1]{\lvert #1\rangle}
\newcommand{\bra}[1]{\langle #1\rvert}

\newcommand{\proj}[1]{\ket{#1}\bra{#1}}

\newtheoremstyle{breakprop}
    {6pt}        
    {6pt}        
    {\itshape}   
    {}           
    {\bfseries}  
    {.}          
    {\newline}   
    {}           

\theoremstyle{breakprop}

\newtheorem{theorem}{Theorem}
\newtheorem{lemma}{Lemma}
\newtheorem{proposition}{Proposition}
\newtheorem{corollary}{Corollary}
\newtheorem{observation}{Observation}
\newtheorem{definition}{Definition}

\makeatletter

\newtheorem*{rep@theorem}{\rep@title}

\newcommand{\newreptheorem}[2]{%
    \newenvironment{rep#1}[1]{%
        \def\rep@title{#2 \ref{##1}}%
        \begin{rep@theorem}%
    }{%
        \end{rep@theorem}%
    }%
}

\makeatother

\newreptheorem{thm}{Theorem}

\newcommand{\QuantumIncludeFigure}[2][]{%
  \IfFileExists{#2}{\includegraphics[#1]{#2}}{%
    \fbox{\parbox[c][1.35cm][c]{0.84\linewidth}{\centering\footnotesize
    MISSING ORIGINAL FIGURE\\\texttt{\detokenize{#2}}}}%
  }%
}

\begin{document}

\title{Encrypted Quantum Cloning Beyond Maximal Entanglement: Exact Classification of Multipartite Resources}

\author{Pritam Roy}
\email{roy.pritamphy@gmail.com}
\affiliation{
S. N. Bose National Centre for Basic Sciences,
Kolkata 700106, India
}

\author{Shashank Gupta}
\email{shashankg687@gmail.com}
\affiliation{
Centre for Interdisciplinary Areas in Quantum Computing,
Indian Institute of Technology Indore,
Indore 453552, India
}
\affiliation{
Department of Physics,
Indian Institute of Technology Indore,
Indore 453552, India
}

\begin{abstract}
Encrypted quantum cloning distributes an unknown $k$-qubit state among $m$ individually concealed signals, each permitting exact recovery when combined with a common quantum key. While canonical constructions rely on maximally entangled resources, they do not establish which multipartite states are necessary or sufficient for this task. Here we derive necessary and sufficient conditions characterizing all pure resources compatible with a prescribed sector-wise two-Pauli encoder for $m\geq 2$, with necessity holding for arbitrary completely positive trace-preserving recovery maps. The classification reveals a sharp parity dependence: even $m$ requires maximal signal--key entanglement, whereas odd $m$ admits nonmaximally entangled resources when their residual signal correlations belong to an encoder-selected commuting algebra. Beyond this architecture, we show that exact recovery from every authorized subsystem under arbitrary unitary encoders acting on the input and signals while leaving the key untouched necessarily implies perfect individual concealment, requires at least $(m-1)k$ ebits of signal--key entanglement, and imposes spectral constraints on the resource. For graph states, the classification reduces to an exact cut-kernel criterion, enabling binary certification and constructive Clifford recovery. In particular, a rank-deficient three-output resource in the YZ/ZY branch attains the general entanglement lower bound with two ebits rather than three. We also exclude the entire Dicke family, including $W$ states, within the two-Pauli architecture. Together, these results distinguish fundamental entanglement requirements from encoder-dependent correlation constraints, providing an exact characterization of resource admissibility and a constructive recovery framework for graph states.
\end{abstract}
\maketitle

\section{Introduction}
\label{sec:introduction}

The no-cloning theorem~\cite{wootters1982single,dieks1982communication}
forbids perfect duplication of an arbitrary unknown quantum state,
while approximate cloning~\cite{scarani2005quantum} and no
broadcasting~\cite{barnum1996noncommuting} establish related
limitations on quantum information distribution. Encrypted quantum
cloning, introduced by Yamaguchi and Kempf~\cite{yamaguchi2026encrypted},
nevertheless permits a different form of quantum redundancy: an
unknown state is distributed among individually concealed signals,
each allowing exact recovery when combined with a common consumable
quantum key. These signals provide alternative recovery pathways
rather than simultaneously accessible copies, remaining consistent
with the no-cloning principle. The protocol has also been demonstrated
on superconducting quantum hardware~\cite{yamaguchi2026experimental}.
Its canonical construction uses maximally entangled signal--key
resources, raising a fundamental question: \emph{is maximal
entanglement necessary for exact encrypted recovery, or can the
same task be realized with less entanglement?} We show that, for
encoders built from two Pauli rotations, the answer depends sharply
on the parity of the number $m$ of signals: maximal entanglement
is necessary for even $m$, whereas odd $m$ admits nonmaximally
entangled resources whose residual signal correlations belong
to an algebra selected by the encoder.

This question extends beyond finding alternative resources for
an existing protocol. A successful construction establishes the
sufficiency of a particular shared state, but does not characterize
all states capable of realizing the same task or identify which
resource requirements are fundamentally unavoidable. Since preparing
and distributing high-quality entanglement are important
considerations in quantum networks and hardware
implementations~\cite{kimble2008quantum,wehner2018quantum,
yamaguchi2026experimental}, determining whether maximal
entanglement is genuinely \emph{necessary}, and for which encoders,
is relevant to assessing the resources required for encrypted
information distribution. Recent studies have extended encrypted
cloning to higher-dimensional systems~\cite{ceara2026qudit},
established connections with absolutely maximally entangled
states and quantum secret sharing~\cite{lim2026ame},
analyzed information leakage~\cite{gianini2026leak,
gianini2026full,bai2026}, developed more general authorization
structures~\cite{gianini2026access}, and explored state-blind
relational fault diagnosis~\cite{gianini2026diagnosticresource}.
While these advances establish new constructions and operational
properties, to our knowledge they do not provide a complete
characterization of which arbitrary multipartite pure states
support exact encrypted recovery for a prescribed encoding and
recovery architecture. Addressing this resource-classification
problem requires separating fundamental information-theoretic
restrictions from conditions imposed by the chosen encoding
operations.

This perspective connects encrypted cloning to distributed quantum
information processing, quantum secret sharing, and quantum error
correction~\cite{kimble2008quantum,wehner2018quantum,
cuomo2020towards,hillery1999quantum,CleveGottesmanLo1999,
Gottesman2000,ImaiEtAl2005,Gottesman1997,
SchumacherNielsen1996}. Graph states have been particularly
useful for characterizing authorized subsets and decoding
operations in secret-sharing schemes~\cite{MarkhamSanders2008,
Sarvepalli2012,MarinMarkhamPerdrix2013}. Whereas such approaches
commonly examine recoverability within a given sharing
construction, we prescribe the encrypted-cloning encoder and
authorized recovery subsystems while allowing the resource state
to vary. Specifically, an unknown $k$-qubit input is distributed
among $m\geq2$ signals $S_1,\ldots,S_m$, with exact recovery
required from every $S_i\mathcal N$, where $\mathcal N$ denotes
the complete common quantum key. The encoded source register
$A$ is not counted among the prescribed outputs, unlike the
additional recovery possibility discussed in the Supplemental
Material of Ref.~\cite{yamaguchi2026encrypted}. Concealment
is required only for individual signals; joint access to
several signals may reveal additional information about
the input.

We first establish restrictions that hold for arbitrary unitary
encoders acting on the input and signal registers while leaving
$\mathcal N$ untouched. Exact recovery from every
$S_i\mathcal N$ necessarily implies perfect individual
concealment, requires at least $(m-1)k$ ebits of signal--key
entanglement, and imposes a spectral-degeneracy constraint
on the signal marginal. These conditions are necessary but
not sufficient for general encoders. For the prescribed
sector-wise two-Pauli architecture, we derive a complete
necessary-and-sufficient characterization of compatible
multipartite pure resources, with necessity holding for
arbitrary completely positive trace-preserving (CPTP)
recovery maps. For even $m$, the signal marginal must
be maximally mixed, whereas for odd $m$, residual correlations
remain admissible precisely when they belong to an
encoder-selected commuting algebra of complete-sector
Pauli operators. Consequently, resources with identical
signal--key entanglement can exhibit different recovery
capabilities depending on their correlation structure.
Thus, beyond conventional entanglement
characterization~\cite{horodecki2009quantum,
guhne2009entanglement}, exact recovery requires identifying
which correlations are compatible with the prescribed encoder.

For graph states, this distinction becomes an exact
$\mathbb F_2$ linear-algebraic criterion. Their bipartite
entanglement and stabilizer correlations admit descriptions
through adjacency matrices~\cite{hein2006entanglement,
hein2004,vandennest2004graphical,FattalEtAl2004}.
We show that resource compatibility is determined by
the kernel of the signal--key cut matrix: its rank
quantifies the signal--key entanglement, while its
complete kernel determines whether the residual
correlations permit recovery. Full cut rank guarantees
exact recovery for every $m\geq2$ and every prescribed
two-Pauli encoder, but certain odd-$m$ branches also
admit rank-deficient resources. These distinctions
are relative to a fixed graph-state Pauli frame,
since simultaneous local Clifford transformations
relate the corresponding encoder--resource descriptions.
Within such a frame, the odd-$m$ YZ/ZY branch can attain
$r=(m-1)k$, whereas XZ/ZX requires
$r\geq(m-1)k+1$ and XY/YX requires full cut rank.
These rank bounds are necessary consequences of the
classification, while the complete kernel criterion
determines sufficiency in the exceptional branches.
In particular, for $(m,k)=(3,1)$, a rank-deficient
YZ/ZY graph-state resource with only two ebits of
signal--key entanglement, rather than the three ebits
of a maximally entangled resource, saturates the
architecture-independent lower bound while preserving
exact recovery and individual concealment.

The graph-state classification is also constructive.
We introduce \emph{Graph-State Encrypted-Cloning
Certification} (\textsc{GSECC}), which uses binary
linear algebra to identify compatible graph-state
realizations and yields explicit Clifford recovery
operations for every certified resource. Certification
for a prescribed signal--key cut is polynomial-time,
although graph-level certification may require a
combinatorial search over cuts. Full-rank resources
additionally admit a key-side Clifford reduction to
the canonical maximally entangled form. The
classification also establishes no-go results:
the entire Dicke family, including $W$ states, is
incompatible with every sector-wise two-Pauli
encoder for $m\geq2$. Independently, the general
entanglement bound excludes $W$-state resources for
arbitrary unitary encoders whenever $(m-1)k>1$.
These results distinguish insufficient entanglement
from incompatibility between resource correlations
and the encoding architecture.

Together, our results establish general necessary resource
conditions beyond the canonical encoder, an exact
pure-resource classification revealing parity-dependent
nonmaximally entangled solutions, constructive graph-state
certification and Clifford recovery, and structural no-go
theorems for multipartite resource families. Rather than
extending particular encrypted-cloning constructions,
the framework distinguishes the entanglement required
for exact recovery from the additional correlation
constraints imposed by its implementation. The resulting
resource--encoder dependence provides a concrete
starting point for investigating which multipartite
states can realize prescribed quantum recovery tasks.
A complete characterization for arbitrary unitary
encoders remains open.

\begin{figure*}[t]
\centering
\QuantumIncludeFigure[width=0.94\textwidth]{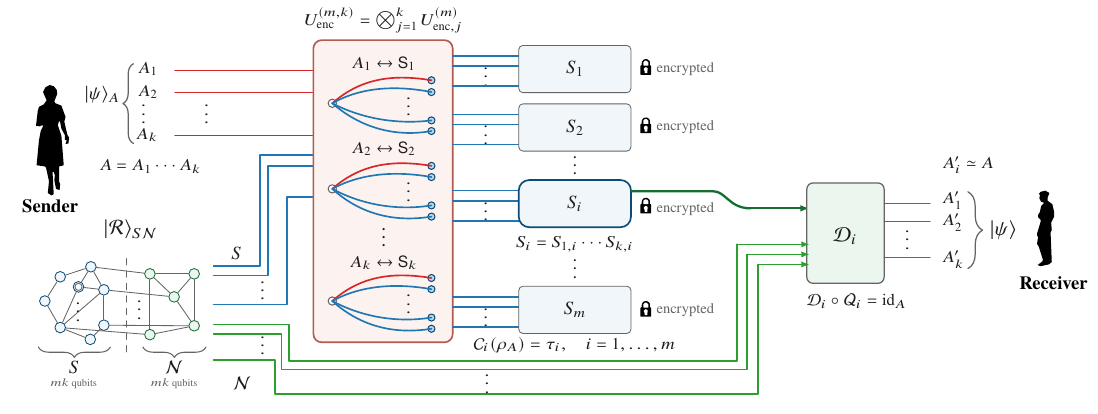}
\caption{\footnotesize
\textbf{Operational architecture of multiparty encrypted quantum cloning.}
An arbitrary $k$-qubit state $\ket{\psi}_A$, with
$A=A_1\cdots A_k$, is encoded together with a pure multipartite resource
$\ket{\mathcal R}_{S\mathcal N}$ containing $\nu=mk$ signal qubits and
$\nu$ noise, or key, qubits. The signal register has two compatible
decompositions: the logical sector
$\mathsf S_j=(S_{j,1},\ldots,S_{j,m})$ is associated with input qubit
$A_j$, while
$S_i=S_{1,i}\cdots S_{k,i}$ forms the $i$th encrypted output.
The sector-wise encoder leaves the common key $\mathcal N$ untouched.
Each $S_i$ is individually concealed, whereas $S_i\mathcal N$ permits
exact recovery of the complete logical state. The registers
$S_1,\ldots,S_m$ therefore constitute alternative encrypted recovery
pathways sharing a common quantum key rather than simultaneously
accessible copies.
}
\label{fig:main_encrypted_cloning_architecture}
\end{figure*}

The remainder of the paper is organized as follows.
Sec.~\ref{sec:general_framework} presents the general resource
framework and derives architecture-independent constraints on
exact recovery. Sec.~\ref{sec:graph_resources} establishes
the exact graph-state classification, including its parity-
and encoder-dependent structure. Sec.~\ref{sec:certification}
introduces \textsc{GSECC} and develops constructive Clifford
recovery procedures. Sec.~\ref{sec:resource_landscape}
examines representative multipartite resource families and
no-go results. Finally, Sec.~\ref{sec:discussion} discusses
the broader implications and open questions.

\section{General framework and resource conditions}
\label{sec:general_framework}

\subsection{Multiparty encrypted-cloning architecture}
\label{subsec:main_architecture}

We consider a multiparty extension of encrypted quantum cloning
~\cite{yamaguchi2026encrypted}, illustrated in
Fig.~\ref{fig:main_encrypted_cloning_architecture}.
Let $m,k\geq1$ and set $\nu=mk$. The input register
$A=A_1\cdots A_k$ contains $k$ qubits, while the pure multipartite
resource $\ket{\mathcal R}_{S\mathcal N}$ contains $\nu$ signal
qubits and $\nu$ noise, or key, qubits. The signal register admits
two compatible decompositions,
\[
\begin{aligned}
S_i
&=
S_{1,i}\cdots S_{k,i},
&& i=1,\ldots,m,
\\
\mathsf S_j
&=
(S_{j,1},\ldots,S_{j,m}),
&& j=1,\ldots,k,
\\
S
&=
S_1\cdots S_m
=
\mathsf S_1\cdots\mathsf S_k .
\end{aligned}
\]
Here $S_i$ is the $i$th $k$-qubit signal output, whereas
$\mathsf S_j$ is the $m$-qubit signal sector associated with the
input qubit $A_j$.

For an ordered pair of distinct Pauli operators
$P,Q\in\{X,Y,Z\}$, we consider the sector-wise encoder
\[
\begin{aligned}
U_{P,Q}^{(m,k)}
&=
\bigotimes_{j=1}^{k}
U_{P,Q;j}^{(m)}\quad,
\\
U_{P,Q;j}^{(m)}
&=
e^{-i\pi P_{A_j}P_{\mathsf S_j}^{\otimes m}/4}
e^{-i\pi Q_{A_j}Q_{\mathsf S_j}^{\otimes m}/4}.
\end{aligned}
\]
The encoder acts on $AS$ and leaves the complete key register
$\mathcal N$ untouched. The canonical encrypted-cloning construction
corresponds to $(P,Q)=(X,Z)$
~\cite{yamaguchi2026encrypted}. We retain the full ordered two-Pauli
family because, relative to a fixed resource and Pauli frame, exact
recovery can depend on the encoding axes.

Writing $U\equiv U_{P,Q}^{(m,k)}$ and suppressing the identity action
on $\mathcal N$, define the encoded state by
\[
\mathcal V_{\mathcal R}(\rho_A)
=
U
\left(
\rho_A\otimes\proj{\mathcal R}
\right)
U^\dagger .
\]
For each signal output $S_i$, let
\[
\begin{aligned}
\mathcal C_i(\rho_A)
&=
\Tr_{A S_{\bar i}\mathcal N}
\mathcal V_{\mathcal R}(\rho_A),
\\
\mathcal Q_i(\rho_A)
&=
\Tr_{A S_{\bar i}}
\mathcal V_{\mathcal R}(\rho_A),
\qquad
S_{\bar i}
=
\bigotimes_{\ell\neq i}S_\ell .
\end{aligned}
\]
Thus $\mathcal C_i$ is the channel accessible from $S_i$ alone,
whereas $\mathcal Q_i$ describes the system available when $S_i$ is
combined with the full key register $\mathcal N$.

Perfect \emph{individual concealment} requires $\mathcal C_i$ to be
constant,
\[
\mathcal C_i(X)
=
\Tr(X)\tau_i,
\qquad
i=1,\ldots,m,
\]
for some fixed state $\tau_i$ on $S_i$. Equivalently, the state of
$S_i$ is independent of the input. Perfect \emph{authorized recovery} requires the existence of a
completely positive trace-preserving (CPTP) recovery map
\[
\mathcal D_i:
S_i\mathcal N\longrightarrow A_i',
\qquad
A_i'\simeq A,
\]
satisfying $\mathcal D_i\circ\mathcal Q_i
=
\id_A .$

Recovery is required for the complete $k$-qubit input, including its
arbitrary correlations among the input qubits and with an external
reference.

Throughout the fixed-encoder classification, we assume a pure resource
on $S\mathcal N$, the sector-wise two-Pauli encoder defined above
acting trivially on $\mathcal N$, and $m$ prescribed disjoint
$k$-qubit signal outputs $S_1,\ldots,S_m$. Every authorized decoder
receives one signal $S_i$ together with the complete common key
$\mathcal N$, and the recovery map is otherwise unrestricted.
Concealment is imposed on each $S_i$ individually; no joint-concealment
condition is imposed on collections of several signal outputs.

To formulate concealment and recovery independently of a particular
input state, introduce a $k$-qubit reference $\widetilde A$ maximally
entangled with $A$,
\[
\ket{\Phi^+}_{\widetilde A A}
=
2^{-k/2}
\sum_{x\in\{0,1\}^k}
\ket{x}_{\widetilde A}\ket{x}_A ,
\]
and define the encoded Choi state~\cite{Choi1975},
\[
\ket{\Omega}_{\widetilde A A S\mathcal N}
=
\left(
I_{\widetilde A}\otimes U\otimes I_{\mathcal N}
\right)
\left(
\ket{\Phi^+}_{\widetilde A A}
\otimes
\ket{\mathcal R}_{S\mathcal N}
\right).
\]
Because the encoder leaves $\mathcal N$ untouched,
\[
\Omega_{\widetilde A\mathcal N}
=
\frac{I_{\widetilde A}}{2^k}
\otimes
\rho_{\mathcal N},
\qquad
\rho_{\mathcal N}
=
\Tr_S\proj{\mathcal R}.
\]

Individual concealment of $S_i$ is equivalent to
\[
\Omega_{\widetilde A S_i}
=
\frac{I_{\widetilde A}}{2^k}
\otimes
\Omega_{S_i}.
\]
Exact recovery from $S_i\mathcal N$ is equivalent to decoupling the
reference from the complementary subsystem,
\begin{equation}
\mathcal D_i\circ\mathcal Q_i=\id_A
\quad\Longleftrightarrow\quad
I(\widetilde A:AS_{\bar i})_\Omega=0 .
\label{eq:main_recovery_decoupling}
\end{equation}
This is the exact information--disturbance, or recovery--decoupling,
criterion~\cite{SchumacherNielsen1996}. A direct derivation from purity
and uniqueness of purification is given in Appendix~\ref{subsec:supp_general_preliminaries}.

Purity of $\Omega$ further gives
\[
I(\widetilde A\rangle S_i\mathcal N)_\Omega
=
k-
I(\widetilde A:AS_{\bar i})_\Omega ,
\]
so exact recovery is equivalent to
\[
I(\widetilde A\rangle S_i\mathcal N)_\Omega
=
k .
\]
Since the input dimension is $2^k$, the coherent information cannot
exceed $k$. Exact reversibility attains this value for the maximally
entangled input and also permits perfect transmission of the full
$k$-qubit input. Consequently,
\begin{equation}
I(\widetilde A\rangle S_i\mathcal N)_\Omega
=
Q^{(1)}(\mathcal Q_i)
=
C_Q(\mathcal Q_i)
=
k ,
\label{eq:main_capacity_saturation}
\end{equation}
where $Q^{(1)}$ denotes the optimized single-use coherent information
and $C_Q$ the quantum capacity
~\cite{Lloyd1997,Devetak2005,DevetakShor2005}.
\subsection{General pure-state resources}
\label{subsec:main_general_resource}

We now characterize all pure resources compatible with the fixed
sector-wise two-Pauli encoder, without assuming graph-state or
stabilizer structure. For distinct Pauli operators $P$ and $Q$, let
\[
R=iPQ,
\]
which, up to sign, is the Pauli operator complementary to $P$ and $Q$.
For each sector define
\[
T_j
=
\bigotimes_{\ell=1}^{m}R_{S_{j,\ell}},
\qquad
T(c)
=
\prod_{j=1}^{k}T_j^{c_j},
\qquad
c\in\F^k .
\]
The operators $T_j$ act on disjoint sectors and therefore commute.

The relevant structure follows from the Pauli action of the encoder.
For even $m$, no nonidentity Pauli operator supported entirely on $S$
is conjugated to an operator with trivial signal support. For odd $m$,
the unique such direction within each sector is $R^{\otimes m}$.
Hence all signal correlations compatible with exact all-output recovery
are generated by the complete-sector operators $T_j$. The full
binary-symplectic derivation is given in Appendix~\ref{subsec:supp_general_preliminaries}.

\begin{theorem}[Exact fixed-encoder pure-resource criterion]
\label{thm:exact_general_resource_main}
Let $m\ge2$, $k\ge1$, $\nu=mk$, and fix
$U_{P,Q}^{(m,k)}$. For an arbitrary pure resource
$\ket{\mathcal R}_{S\mathcal N}$ with
$\rho_S=\Tr_{\mathcal N}\proj{\mathcal R}$, exact recovery of the
complete $k$-qubit input from every $S_i\mathcal N$, using otherwise
unrestricted CPTP recovery maps, is possible if and only if
\begin{equation}
\rho_S=
\begin{cases}
\dfrac{I_S}{2^\nu},
&
m\ \mathrm{even},
\\[3mm]
\dfrac{1}{2^\nu}
\displaystyle\sum_{c\in\F^k}
\alpha_c\,T(c),
&
m\ \mathrm{odd},
\end{cases}
\label{eq:main_exact_pure_resource}
\end{equation}
where $\alpha_0=1$, $\alpha_c\in\mathbb R$, and the odd-$m$
expression is positive semidefinite. Whenever
Eq.~\eqref{eq:main_exact_pure_resource} holds, every individual signal
is automatically perfectly concealed,
\[
\mathcal C_i(\rho_A)
=
\frac{I_{S_i}}{2^k},
\qquad
i=1,\ldots,m,
\]
for every input state $\rho_A$.
\end{theorem}

Theorem~\ref{thm:exact_general_resource_main} is decoder independent
but encoder dependent. If Eq.~\eqref{eq:main_exact_pure_resource}
fails, exact recovery is impossible for at least one
$S_i\mathcal N$, irrespective of the CPTP recovery map. The theorem
is therefore a complete operator-level characterization of the
admissible signal marginal, rather than merely an entanglement
condition.

For even $m$, all nontrivial Pauli components of $\rho_S$ are excluded,
so $\rho_S$ must be maximally mixed. Since the resource is pure, this
is equivalent to maximal entanglement across the signal--key cut
$S:\mathcal N$. For odd $m$, nonmaximal entanglement across the same
cut is possible, but every departure from maximal mixing must lie in
the commuting algebra generated by $T_1,\ldots,T_k$. Thus resources
with the same signal--key entanglement can nevertheless differ in
their ability to support exact recovery because their surviving
correlations are organized differently relative to the encoder. The
complete necessity-and-sufficiency proof, including the Pauli
no-cancellation argument, is given in Appendix~\ref{subsec:supp_proof_general_resource}.

The concealment established by
Theorem~\ref{thm:exact_general_resource_main} applies to each signal
individually and does not imply concealment of the signals under joint
access. For odd \(m\), one can identify a family of input moments that
is always accessible from the complete signal register. Define
\begin{equation*}
R_A(c)
=
\prod_{j=1}^{k} R_{A_j}^{\,c_j},
\qquad
\sigma_m
=
(-1)^{(m-1)/2}.
\end{equation*}
The encoder satisfies
\begin{equation*}
U_{P,Q}^{(m,k)\dagger}
\bigl(I_A\otimes T(c)\bigr)
U_{P,Q}^{(m,k)}
=
\sigma_m^{|c|}
R_A(c)\otimes I_S,
\end{equation*}
where \(|c|\) denotes the Hamming weight of \(c\). For the channel to
the complete signal register,
\begin{equation*}
\mathcal C_S(\rho_A)
=
\Tr_{A\mathcal N}
\mathcal V_{\mathcal R}(\rho_A),
\end{equation*}
this identity implies
\begin{equation*}
\Tr\!\left[
T(c)\,
\mathcal C_S(\rho_A)
\right]
=
\sigma_m^{|c|}
\Tr\!\left[
R_A(c)\rho_A
\right].
\end{equation*}

Joint access to \(S_1\cdots S_m\) therefore always reveals at least the
commuting \(R\)-basis moments of the input, equivalently, its diagonal
statistics in the product \(R\) basis. These moments constitute a
guaranteed component of the joint-signal information, but need not
exhaust it. Nonidentity terms in an admissible signal marginal
\(\rho_S\) can make additional, generally noncommuting, input moments
accessible from the complete signal register, and for some resources
these moments are sufficient to determine the full input state
\(\rho_A\).
The derivation is given in Appendix~\ref{subsec:supp_joint_ciphertext}.

This is consistent with our operational requirement, which imposes
perfect concealment on each individual \(S_i\), but not on arbitrary
collections of signal outputs. Related leakage from unauthorized
subsystems in fixed encrypted-cloning constructions has been analyzed
in Refs.~\cite{gianini2026leak,gianini2026full,bai2026}. In the canonical
single-qubit Bell-pair protocol, the odd-\(m\) leakage is carried by
\(Y^{\otimes m}\)~\cite{gianini2026leak}; more general access-structure
formulations interpret such partially informative non-recovering
subsystems within ramp-QSS schemes.

A direct consequence of
Theorem~\ref{thm:exact_general_resource_main} is the maximally
entangled resource class.

\begin{corollary}[Maximally entangled resources]
\label{cor:maximal_resource_main}
For every $m\ge2$, $k\ge1$, and fixed sector-wise two-Pauli encoder,
any pure resource satisfying
\[
\rho_S
=
\frac{I_S}{2^{mk}}
\]
supports exact recovery of the complete $k$-qubit input from every
$S_i\mathcal N$ and perfect individual concealment of every $S_i$.
For even $m$, this condition is also necessary.
\end{corollary}

The proof is given in Appendix~\ref{subsec:supp_general_corollaries}.

The condition $\rho_S=I_S/2^{mk}$ is equivalent to maximal
entanglement across $S:\mathcal N$. Since
$\dim S=\dim\mathcal N=2^{mk}$, uniqueness of purification implies
that every such resource is related to a canonical maximally
entangled state by a unitary acting only on $\mathcal N$,
\[
\ket{\mathcal R}_{S\mathcal N}
=
\left(
I_S\otimes V_{\mathcal N}
\right)
\ket{\Phi_{2^{mk}}}_{S\mathcal N}.
\]
The Bell-pair resource used in the canonical protocol
~\cite{yamaguchi2026encrypted} is therefore one representative of an
entire key-side-unitary equivalence class of valid maximally entangled
resources.

The same maximally entangled resource can also support different
factorizations of the signal register.

\begin{corollary}[Relabelling of maximally entangled resources]
\label{cor:hierarchical_main}
Let $\ket{\mathcal R}_{S\mathcal N}$ be maximally entangled across
$S:\mathcal N$, with
$\dim S=\dim\mathcal N=2^\nu$. For any factorization
$\nu=m'k'$ with $m'\ge2$, a relabelling of the $\nu$ signal qubits
into $k'$ sectors of size $m'$, together with the corresponding
sector-wise encoder and recovery maps, yields a valid $(m',k')$
encrypted-cloning realization.
\end{corollary}

Thus, within the maximally entangled class, the same resource can
support several output--input factorizations of a fixed total signal
size $\nu$. Explicit relabellings and recovery constructions are given
in Appendix~\ref{subsec:supp_general_corollaries}.
\subsection{Single-output boundary and architecture-independent constraints}
\label{subsec:main_single_clone}

Theorem~\ref{thm:exact_general_resource_main} concerns the
multi-output regime $m\ge2$. The single-output case separates a
general information-theoretic obstruction from the stronger restriction
imposed by the sector-wise encoder. Set $d=2^k$ and regard
$A$, $S$, and $\mathcal N$ as $d$-dimensional systems. For
$m=1$ and $k=1$, corresponding to the single-qubit limit of the
original encrypted-cloning setting~\cite{yamaguchi2026encrypted},
perfect concealment and exact recovery cannot coexist. For larger
input dimension, however, the unrestricted task becomes possible.

\begin{proposition}[Global single-output characterization]
\label{prop:single_clone_global_main}
For $m=1$, let the encoder be an arbitrary unitary on $AS$ and let
$\ket{\mathcal R}_{S\mathcal N}$ be an arbitrary pure resource.
Perfect concealment of $S$ together with exact recovery of the complete
$k$-qubit input from $S\mathcal N$ is possible if and only if the
encoded Choi state
$\ket{\Omega}_{\widetilde AAS\mathcal N}$ is an
$\operatorname{AME}(4,d)$ state. Every such realization satisfies
\[
\rho_S
=
\rho_{\mathcal N}
=
\frac{I_d}{d}.
\]
Consequently, the task is impossible for $k=1$ and achievable for
every $k\ge2$.
\end{proposition}

For $k=1$, the proposition would require an
$\operatorname{AME}(4,2)$ state, which does not exist
~\cite{HiguchiSudbery2000}. For $k\ge2$, $d=2^k$ is a prime-power
dimension admitting suitable $\operatorname{AME}(4,d)$ constructions
~\cite{HelwigEtAl2012}. An explicit finite-field realization and
recovery map are given in Appendix~\ref{subsec:supp_proof_single_clone_global}. This AME characterization is consistent with
the established connection between pure-state information sharing and
AME structures; here its role is specifically to identify the
$m=1$ boundary of the encrypted-cloning task. It is distinct from the
$\operatorname{AME}(5,d)$ construction discussed by Lim and Lo
~\cite{lim2026ame}, which concerns a particular two-output
encrypted-cloning protocol and its relation to quantum secret sharing.

The fixed sector-wise two-Pauli encoder is more restrictive.

\begin{proposition}[Single-output obstruction of the sector-wise encoder]
\label{prop:single_clone_leakage_main}
For $m=1$, every encoder $U_{P,Q}^{(1,k)}$ fails perfect individual
concealment for every pure resource
$\ket{\mathcal R}_{S\mathcal N}$ and every $k\ge1$.
\end{proposition}

Indeed, with $R=iPQ$, each sector obeys
\[
U_{P,Q;j}^{(1)\dagger}
\left(
I_{A_j}\otimes R_{S_{j,1}}
\right)
U_{P,Q;j}^{(1)}
=
R_{A_j}\otimes I_{S_{j,1}} .
\]
Thus an observable of the signal reproduces an input expectation value
independently of the resource state, so the signal channel cannot be
constant. The obstruction is therefore architectural: it persists for
the sector-wise encoder for every $k$, although the unrestricted
single-output task is possible for $k\ge2$. The complete Clifford
derivation is given in Appendix~\ref{subsec:supp_proof_single_clone_sectorwise}.

We next remove the sector-wise Pauli
assumptions. The only structural restriction retained is that the
encoder is unitary on $AS$ and acts trivially on the common key
$\mathcal N$.

\begin{proposition}[Architecture-independent all-output constraint]
\label{prop:all_clone_entropy_main}
Let $m\ge2$ and let $\ket{\mathcal R}_{S\mathcal N}$ be pure.
Allow an arbitrary unitary encoder on $AS$ acting trivially on
$\mathcal N$. If the complete $k$-qubit input is exactly recoverable
from $S_i\mathcal N$ for every $i=1,\ldots,m$, then
\begin{equation}
S(\rho_{\mathcal N})
=
(m-1)k+S(A)_\Omega
\ge
(m-1)k .
\label{eq:main_all_clone_entropy}
\end{equation}
Moreover, exact all-output recovery already implies perfect
individual concealment,
\[
\mathcal C_i(\rho_A)
=
\frac{I_{S_i}}{2^k},
\qquad
i=1,\ldots,m ,
\]
for every input state $\rho_A$.
\end{proposition}

This result is independent of the encoding architecture. Recovery from
the alternative authorized subsystems enforces the decoupling condition
in Eq.~\eqref{eq:main_recovery_decoupling}; for $m\ge2$, this also
decouples the reference from each individual $S_i$ and forces $S_i$
to be maximally mixed. The resulting entropy constraint is related to
share-size bounds in quantum secret sharing
~\cite{Gottesman2000,ImaiEtAl2005}, but here it constrains the common
key because the same $k$-qubit state must be recoverable from every
$S_i\mathcal N$. The proof is given in Appendix~\ref{subsec:supp_proof_all_output_entropy}.

Since the resource is pure,
$S(\rho_{\mathcal N})=S(\rho_S)$ is the entanglement entropy across
the signal--key cut $S:\mathcal N$. Thus
Eq.~\eqref{eq:main_all_clone_entropy} gives
\[
0
\le
mk-S(\rho_S)
=
k-S(A)_\Omega
\le
k .
\]
Any pure resource supporting exact all-output recovery therefore
contains at least $(m-1)k$ ebits across $S:\mathcal N$.

The same argument gives a stronger constraint on the full spectrum of
the signal marginal.

\begin{corollary}[Architecture-independent spectral constraint]
\label{cor:main_spectral_constraint}
Under the assumptions of
Proposition~\ref{prop:all_clone_entropy_main}, there exists a
$k$-qubit density operator $\sigma$ such that
\[
\rho_S
\simeq
\sigma
\otimes
\frac{I_{2^{(m-1)k}}}{2^{(m-1)k}},
\]
where $\simeq$ denotes unitary equivalence. Equivalently, every
distinct eigenvalue of $\rho_S$ has multiplicity divisible by
$2^{(m-1)k}$.
\end{corollary}

The proof is given in Appendix~\ref{subsec:supp_proof_all_output_entropy}.

The spectral condition goes beyond the entropy bound by constraining how
the eigenvalues of the resource marginal are degenerate, rather than
only their total entropy. It is a necessary consequence of exact
all-output recovery, but we do not claim that it is sufficient for an
arbitrary encoder.

The fixed two-Pauli theorem sharpens these architecture-independent
constraints. For even $m$, the allowed entanglement deficit vanishes
and $\rho_S$ must be maximally mixed. For odd $m$, a deficit is
possible only when the surviving signal correlations belong to the
encoder-selected commuting algebra in
Eq.~\eqref{eq:main_exact_pure_resource}. For graph states, these two
features become complementary binary data: the cut rank determines the
signal--key entanglement, while the complete cut kernel determines
whether a rank deficiency is compatible with exact recovery.

\section{Exact characterization of graph-state resources}
\label{sec:graph_resources}

Graph states provide a natural stabilizer framework for quantum secret
sharing and access-structure problems
~\cite{MarkhamSanders2008,Sarvepalli2012,MarinMarkhamPerdrix2013},
where one typically asks which subsets can reconstruct an encoded
secret. Here the authorized subsystems are fixed in advance by the
encrypted-cloning architecture. Our question is instead which
graph-state resources support exact recovery from every prescribed
$S_i\mathcal N$ under a fixed encoder.

We now specialize the general pure-resource criterion to graph states,
while retaining unrestricted CPTP recovery maps. Throughout this
section, $m\ge2$, the encoder $U_{P,Q}^{(m,k)}$ is fixed, and
$S:\mathcal N$ denotes an oriented balanced bipartition of the
$2mk$ resource qubits. The signal register is further decomposed into
the $k$ prescribed sectors
$\mathsf S_1,\ldots,\mathsf S_k$, each containing $m$ qubits. All
binary vectors and matrices are defined over $\mathbb F_2$.

\subsection{Cut rank and encoder-compatible kernel structure}
\label{subsec:main_graph_kernel}

Let $\ket{G}_{S\mathcal N}$ be a graph state on $2mk$ qubits. With
the vertices ordered according to the signal--noise bipartition, write
\[
\Gamma=
\begin{pmatrix}
A_S & B_{S\mathcal N}\\
B_{S\mathcal N}^{\mathsf T} & A_{\mathcal N}
\end{pmatrix},
\]
where $B_{S\mathcal N}$ is the $mk\times mk$ cut matrix. The
graph-state stabilizer generators are
\[
K_v=X_v\prod_{u\in N(v)}Z_u
\]
~\cite{Gottesman1997,hein2004}. For
$x\in\mathbb F_2^{mk}$ supported on $S$, define
$K(x)=\prod_{v\in S}K_v^{x_v}$. Up to Pauli phase,
\[
K(x)\big|_S
\doteq
X^xZ^{A_Sx},
\qquad
K(x)\big|_{\mathcal N}
\doteq
Z^{B_{S\mathcal N}^{\mathsf T}x}.
\]
where $\ker(B_{S\mathcal N}^{\mathsf T})
=\{x\in\mathbb F_2^\nu:B_{S\mathcal N}^{\mathsf T}x=0\}$
denotes the binary kernel of the cut matrix. Writing
$K_S(x)=K(x)\big|_S$ for these signal-only stabilizers, the reduced
signal state is
\[
\rho_S
=
\frac{1}{2^{mk}}
\sum_{x\in\ker(B_{S\mathcal N}^{\mathsf T})}
K_S(x).
\]
The complete cut kernel therefore specifies exactly the signal-side
Pauli correlations responsible for any departure of $\rho_S$ from
maximal mixing.

The same cut matrix determines the entanglement across
$S:\mathcal N$. If
$r=\rank_{\mathbb F_2}B_{S\mathcal N}$, then
\[
\operatorname{SchmidtRank}_{S:\mathcal N}(\ket G)=2^r,
\qquad
S(\rho_S)=r,
\]
which is the standard graph-state cut-rank relation
~\cite{hein2004,hein2006entanglement,FattalEtAl2004}.
The same binary cut-rank function underlies the graph-theoretic notion
of rank width~\cite{oum2005rankwidth}. The graph-state reduced-state
derivation is given in Appendix~\ref{subsec:supp_graph_cut_algebra}. Since
\[
\dim\ker(B_{S\mathcal N}^{\mathsf T})=mk-r,
\]
full cut rank is equivalent to
$\rho_S=I_S/2^{mk}$ and therefore yields a valid resource by
Corollary~\ref{cor:maximal_resource_main}. Conversely,
Proposition~\ref{prop:all_clone_entropy_main} gives the
architecture-independent necessary condition $\rank_{\mathbb F_2}B_{S\mathcal N}
\ge
(m-1)k.$

For a rank-deficient cut, however, the rank alone does not determine
validity. The reduced-state expansion above and
Theorem~\ref{thm:exact_general_resource_main} show that one must also
identify which Pauli correlations are represented by the complete cut
kernel.

For
$x\in\ker(B_{S\mathcal N}^{\mathsf T})$, define
\[
\widetilde K(x)
=
U_{P,Q}^{(m,k)}
\left(
I_A\otimes K_S(x)
\right)
U_{P,Q}^{(m,k)\dagger}.
\]

\begin{definition}[Exceptional kernel vector]
\label{def:main_exceptional_kernel_vector}
A vector
$x\in\ker(B_{S\mathcal N}^{\mathsf T})$ is
\emph{exceptional} for $U_{P,Q}^{(m,k)}$ if
\[
\operatorname{supp}\widetilde K(x)\subseteq A.
\]
The set of all exceptional vectors, including $x=0$, is denoted by
$\mathcal E_{P,Q}^{(m,k)}(A_S)$.
\end{definition}

A kernel stabilizer already acts trivially on $\mathcal N$.
Exceptionality means that the encoder also removes its complete signal
support, leaving an operator only on the input register $A$. These are
precisely the graph-state Pauli components belonging to the exceptional
algebra allowed by
Eq.~\eqref{eq:main_exact_pure_resource}. Since distinct Pauli
operators are linearly independent, the general pure-resource
criterion reduces to
$\ker(B_{S\mathcal N}^{\mathsf T})
\subseteq
\mathcal E_{P,Q}^{(m,k)}(A_S).$ The graph-state classification therefore reduces to determining the
exceptional subspace explicitly.

Let $R=iPQ$ denote, up to sign, the Pauli operator complementary to $P$ and $Q$.
For even $m$, no nonidentity signal-only Pauli can lose all signal
support under the encoder. For odd $m$, the unique exceptional
direction within each active sector is $R^{\otimes m}$. Introduce the
sector-indicator map
\begin{equation}
F:\mathbb F_2^k\longrightarrow\mathbb F_2^{mk},
\qquad
Fc=
(c_1\mathbf 1_m,\ldots,c_k\mathbf 1_m).
\label{eq:main_sector_indicator}
\end{equation}
Thus every exceptional vector for odd $m$ must have the form $x=Fc$.
In the standard graph-state Pauli frame,
\begin{equation}
\mathcal E_{P,Q}^{(m,k)}(A_S)
=
\left\{
\begin{array}{@{}ll@{}}
\{0\},
&
m\ \mathrm{even},
\\[1.5mm]

\left\{
Fc:
\substack{
(A_S+I_{mk})Fc=0,\\
\mathbf 1_k^{\mathsf T}c=0
}
\right\},
&
\substack{
m\ \mathrm{odd},\\
PQ\in\{XZ,ZX\},
}
\\[2mm]

\{Fc:A_SFc=0\},
&
\substack{
m\ \mathrm{odd},\\
PQ\in\{YZ,ZY\},
}
\\[1.5mm]

\{0\},
&
\substack{
m\ \mathrm{odd},\\
PQ\in\{XY,YX\}.
}
\end{array}
\right.
\label{eq:main_exceptional_subspace}
\end{equation}
The complete binary-symplectic derivation is given in Appendix~\ref{subsec:supp_exceptional_kernel}.

The three odd-$m$ branches admit a direct stabilizer interpretation.
For $XZ/ZX$, the exceptional axis is $R=\pm Y$. Since
$K_S(Fc)\doteq X^{Fc}Z^{A_SFc},$ an active sector can be exceptional only if
$(A_S+I_{mk})Fc=0.$ For a simple graph,
$(Fc)^{\mathsf T}A_S(Fc)=0$. Combining this identity with the
condition above and using odd $m$ gives
$\mathbf 1_k^{\mathsf T}c=0.$

Thus every nonzero $Y$-type exceptional stabilizer occupies an even
number of sectors. This explains the parity restriction appearing in
Eq.~\eqref{eq:main_exceptional_subspace}.

For $YZ/ZY$, the exceptional axis is $R=\pm X$, and the condition
reduces to $A_SFc=0,$ with no additional sector-parity constraint. For $XY/YX$, the
exceptional axis is $R=\pm Z$. Every nonzero graph stabilizer
$K_S(x)\doteq X^xZ^{A_Sx}$ nevertheless contains a nontrivial
$X$ component, so no nonzero $Z$-type exceptional vector exists.
The exceptional subspace is therefore trivial in this branch.

The encoder classes in
Eq.~\eqref{eq:main_exceptional_subspace} should be understood relative
to a fixed Pauli frame of the resource. Any two ordered pairs of
distinct Pauli axes are related by single-qubit Clifford conjugation.
A simultaneous Clifford rotation of the input and signal systems
therefore maps one encoder frame and its compatible resource
representatives into another without changing the entanglement across
$S:\mathcal N$. The distinctions above consequently characterize the
compatibility of a fixed graph-state representative with a fixed
encoder frame; they do not imply an absolute inequivalence of the
encoder families under unrestricted basis changes

\subsection{Exact graph-state resource criterion}
\label{subsec:main_exact_graph}

The preceding reduction turns
Theorem~\ref{thm:exact_general_resource_main} into an exact binary
criterion for graph states. The cut kernel specifies the signal-side
Pauli correlations present in $\rho_S$, while the fixed encoder
determines which of those correlations are compatible with exact
recovery.

\begin{theorem}[Exact graph-state resource criterion]
\label{thm:graph_exact_main}
Let $m\ge2$, and fix an oriented balanced cut $S:\mathcal N$, a
sector decomposition of $S$, and a sector-wise two-Pauli encoder
$U_{P,Q}^{(m,k)}$. A graph-state resource
$\ket{G}_{S\mathcal N}$ permits exact recovery of the complete
$k$-qubit input from every $S_i\mathcal N$, using otherwise
unrestricted CPTP recovery maps, if and only if
\begin{equation}
\ker\!\left(B_{S\mathcal N}^{\mathsf T}\right)
\subseteq
\mathcal E_{P,Q}^{(m,k)}(A_S).
\label{eq:main_exact_graph_condition}
\end{equation}
Whenever Eq.~\eqref{eq:main_exact_graph_condition} holds, every
individual signal $S_i$ is perfectly concealed. If the inclusion
fails, exact recovery is impossible for at least one
$S_i\mathcal N$, independently of the recovery map.
\end{theorem}

Theorem~\ref{thm:graph_exact_main} is an immediate specialization of
Theorem~\ref{thm:exact_general_resource_main}. For a graph state, the
Pauli components of the reduced signal state are precisely the
signal-side stabilizers indexed by
$\ker(B_{S\mathcal N}^{\mathsf T})$. The general pure-resource
criterion allows exactly those components belonging to the
encoder-compatible exceptional subspace. Uniqueness of the Pauli
expansion therefore gives
Eq.~\eqref{eq:main_exact_graph_condition}. Decoder independence and
perfect individual concealment are inherited from the general
theorem. A direct stabilizer derivation is given in Appendix~\ref{subsec:supp_proof_graph_exact}.

Exact reversibility also gives the coherent-information and
quantum-capacity saturation discussed in Sec.~\ref{sec:general_framework},
consistently with the standard quantum-capacity framework
~\cite{Lloyd1997,Devetak2005,DevetakShor2005}. Importantly, exact
recovery does not require maximal signal--noise entanglement. A valid
graph state may satisfy
\[
S(\rho_S)
=
\rank_{\mathbb F_2}B_{S\mathcal N}
<
mk,
\]
provided that every direction in the resulting cut kernel belongs to
the exceptional subspace. Conversely, failure of
Eq.~\eqref{eq:main_exact_graph_condition} implies that the maximally
entangled Choi input does not attain coherent information $k$ for at
least one authorized channel. This statement alone does not imply a
strict upper bound on that channel's regularized quantum capacity.

For odd $m$, let $\mathcal E_Y(A_S)$ and $\mathcal E_X(A_S)$ denote
the exceptional spaces associated, respectively, with the $XZ/ZX$
and $YZ/ZY$ classes in
Eq.~\eqref{eq:main_exceptional_subspace}. We use
\[
\mathcal K
\equiv
\ker(B_{S\mathcal N}^{\mathsf T}),
\qquad
r
\equiv
\rank_{\mathbb F_2}B_{S\mathcal N}.
\]

\begin{corollary}[Parity- and encoder-dependent graph-state classes]
\label{cor:main_graph_parity_classes}
For a fixed cut, sector decomposition, and Pauli frame,
\[
\ket G\ \mathrm{valid}
\quad\Longleftrightarrow\quad
\begin{cases}
\mathcal K=\{0\},
&
m\ \mathrm{even},
\\[1.5mm]
\mathcal K\subseteq\mathcal E_Y(A_S),
&
\substack{
m\ \mathrm{odd},\\
(P,Q)\in\{(X,Z),(Z,X)\},
}
\\[1.5mm]
\mathcal K\subseteq\mathcal E_X(A_S),
&
\substack{
m\ \mathrm{odd},\\
(P,Q)\in\{(Y,Z),(Z,Y)\},
}
\\[1.5mm]
\mathcal K=\{0\},
&
\substack{
m\ \mathrm{odd},\\
(P,Q)\in\{(X,Y),(Y,X)\}.
}
\end{cases}
\]
Consequently, every valid realization satisfies
\[
\begin{array}{rcll}
r&=&mk,
&
m\ \mathrm{even}\ \text{or}\ 
m\ \mathrm{odd},\ XY/YX,
\\[1mm]
r&\ge&mk-k+1,
&
m\ \mathrm{odd},\ XZ/ZX,
\\[1mm]
r&\ge&(m-1)k,
&
m\ \mathrm{odd},\ YZ/ZY.
\end{array}
\]
For the two exceptional odd-$m$ branches, these rank bounds are
necessary but not sufficient; validity is determined by the complete
kernel inclusions above.
\end{corollary}

The rank consequences follow directly from rank--nullity and the
dimensions of the exceptional spaces. In the $XZ/ZX$ branch, every
exceptional vector is sector constant and has even sector parity, so
\[
\dim\mathcal E_Y(A_S)\le k-1.
\]
In the $YZ/ZY$ branch,
\[
\dim\mathcal E_X(A_S)\le k.
\]
Since $\dim\mathcal K=mk-r$, these bounds give the stated rank
conditions. For even $m$ and for odd $m$ in the $XY/YX$ branch, the
exceptional space is trivial and full cut rank is therefore necessary
and sufficient.

Table~\ref{tab:main_graph_classes} summarizes the classification.
The third column gives the exact necessary-and-sufficient kernel
condition, whereas the fourth records only its consequence for the
cut rank. This distinction is essential in the exceptional odd-$m$
branches: two graph-state realizations with the same cut rank can have
different validity because their complete kernels can occupy different
directions relative to the exceptional subspace.
The proof is given in Appendix~\ref{subsec:supp_proof_graph_classes}.

\begin{table*}[t]
\caption{
Parity- and encoder-dependent graph-state resource classes for a fixed
cut, sector decomposition, and Pauli frame. The kernel conditions are
exact. In the two exceptional odd-$m$ branches, the rank bounds are
necessary consequences and are not sufficient by themselves.}
\label{tab:main_graph_classes}
\centering
\begin{tabular}{c c c c}
\toprule
$m$
&
Encoder class
&
Exact validity condition
&
Rank consequence
\\
\midrule
even
&
any $P\neq Q$
&
$\ker B_{S\mathcal N}^{\mathsf T}=\{0\}$
&
$\rank B_{S\mathcal N}=mk$
\\
odd
&
$XZ/ZX$
&
$\ker B_{S\mathcal N}^{\mathsf T}
\subseteq\mathcal E_Y(A_S)$
&
$\rank B_{S\mathcal N}\ge mk-k+1$
\\
odd
&
$YZ/ZY$
&
$\ker B_{S\mathcal N}^{\mathsf T}
\subseteq\mathcal E_X(A_S)$
&
$\rank B_{S\mathcal N}\ge(m-1)k$
\\
odd
&
$XY/YX$
&
$\ker B_{S\mathcal N}^{\mathsf T}=\{0\}$
&
$\rank B_{S\mathcal N}=mk$
\\
\bottomrule
\end{tabular}
\end{table*}

Within the standard graph-state Pauli frame, the $YZ/ZY$ branch is the
only class whose exceptional subspace permits the maximal nullity
$k$. If
\[
\dim\mathcal K=k,
\qquad
\mathcal K\subseteq\mathcal E_X(A_S),
\]
then
\[
r
=
S(\rho_S)
=
(m-1)k.
\]
Such a realization saturates the architecture-independent
signal--noise entanglement lower bound of
Proposition~\ref{prop:all_clone_entropy_main}. By contrast, within the
same fixed Pauli frame, the $XZ/ZX$ branch obeys
\[
r\ge(m-1)k+1
\]
and therefore cannot attain this floor.

These statements compare encoder classes relative to a fixed resource
frame. Ordered pairs of distinct Pauli axes are related by
single-qubit Clifford conjugations, and a simultaneous Clifford
rotation of the input and signal systems maps one encoder frame and
its compatible resource representatives into another without changing
the entanglement across $S:\mathcal N$. The distinctions in
Corollary~\ref{cor:main_graph_parity_classes} and
Table~\ref{tab:main_graph_classes} therefore express compatibility of
a fixed graph-state representative with a fixed encoder frame, rather
than an absolute inequivalence of the encoder families under
unrestricted basis changes.

For a single input qubit, the classification simplifies further.

\begin{corollary}[Single-input-qubit classification]
\label{cor:main_graph_k1}
Let $k=1$ and $m\ge2$. For the $XZ/ZX$ and $XY/YX$ encoder classes,
\[
\ket G\ \mathrm{valid}
\quad\Longleftrightarrow\quad
\rank_{\mathbb F_2}B_{S\mathcal N}=m.
\]
The same full-rank condition holds for $YZ/ZY$ when $m$ is even.
For odd $m$ and $YZ/ZY$,
\[
\begin{aligned}
\ket G\ \mathrm{valid}
\quad\Longleftrightarrow\quad&
\mathcal K=\{0\}
\\
&\text{or}\quad
\left[
\begin{array}{c}
\mathcal K=\operatorname{span}\{\mathbf1_m\},\\
A_S\mathbf1_m=0
\end{array}
\right].
\end{aligned}
\]
\end{corollary}

For $k=1$, the even-sector-parity condition eliminates every nonzero
$XZ/ZX$ exceptional direction, while the $XY/YX$ exceptional space
is trivial. In the odd-$m$ $YZ/ZY$ branch, the sector-indicator space
contains only the nonzero direction $\mathbf1_m$. Hence the only
rank-deficient valid possibility is
\[
\mathcal K
=
\operatorname{span}\{\mathbf1_m\},
\qquad
A_S\mathbf1_m=0,
\]
for which
\[
\rank_{\mathbb F_2}B_{S\mathcal N}=m-1.
\]
The proof is given in Appendix~\ref{subsec:supp_proof_k1_classification}.

The graph-state criterion therefore separates three pieces of
structure. The cut rank determines the amount of signal--noise
entanglement, the complete cut kernel specifies the signal-side
stabilizer correlations responsible for any rank deficiency, and the
fixed encoder frame determines whether those correlations are
compatible with exact recovery. Full rank alone characterizes the
even-$m$ and $XY/YX$ branches, whereas the exceptional odd-$m$
branches require the complete kernel structure.

Eq.~\eqref{eq:main_exact_graph_condition} also provides the fixed-
realization certification test. Given a graph, an oriented cut, and a
sector assignment, one constructs $B_{S\mathcal N}$, computes
$\ker(B_{S\mathcal N}^{\mathsf T})$ over $\mathbb F_2$, and tests its
inclusion in the appropriate exceptional subspace. We next promote
this fixed-realization criterion to graph-level certification and show
how the same binary data yield explicit recovery maps.

\section{Certification and constructive decoding}
\label{sec:certification}

Theorem~\ref{thm:graph_exact_main} decides validity for a fixed
oriented cut and sector decomposition. For a physical graph, however,
these structures need not be specified in advance. The graph-level
question is whether the graph admits at least one realization satisfying
the exact cut-kernel criterion. We first formulate this existential
certification problem and then show how a certified realization yields
an explicit Clifford decoder.

\subsection{Graph-state encrypted-cloning certification}
\label{subsec:gsecc}

Let $G=(V,E)$ be a graph on $2\nu$ vertices, with $\nu=mk$. An
$(m,k)$ realization consists of an oriented balanced partition
$V=S\sqcup\mathcal N$, with $|S|=|\mathcal N|=\nu$, together with a
partition
\[
\Pi_S=\{\mathsf S_1,\ldots,\mathsf S_k\}
\]
of the signal vertices into $k$ sectors of size $m$. We denote the
complete realization by
\[
\mathfrak R=(S:\mathcal N,\Pi_S).
\]

For a fixed encoder $U_{P,Q}^{(m,k)}$, define

\begin{equation}
\begin{aligned}
\operatorname{GSECC}_{P,Q}(G;m,k)
&=\textsc{valid}
\\[-1mm]
&\Longleftrightarrow\
\exists\,\mathfrak R:\ \
\ker B_{S\mathcal N}^{\mathsf T}
\subseteq
\mathcal E_{P,Q}^{(m,k)}(A_S).
\end{aligned}
\label{eq:main_gsecc_definition}
\end{equation}

If no such realization exists, we write
$\operatorname{GSECC}_{P,Q}(G;m,k)=\textsc{invalid}$.
Thus GSECC is the existential graph-level extension of
Theorem~\ref{thm:graph_exact_main}. Importantly, failure of one cut or
one sector assignment does not invalidate the graph; invalidity means
that no admissible realization satisfies the exact kernel criterion.

For a prescribed oriented cut $S:\mathcal N$, extract
$A_S$ and $B\equiv B_{S\mathcal N}$ from the adjacency matrix and let
\[
W=
\begin{pmatrix}
w_1&\cdots&w_\eta
\end{pmatrix}
\in\mathbb F_2^{\nu\times\eta},
\qquad
\operatorname{col}(W)=\ker B^{\mathsf T},
\]
where
\[
\eta=\dim\ker B^{\mathsf T}.
\]
For even $m$, and for odd $m$ in the $XY/YX$ class, the exceptional
space is trivial. The cut is therefore admissible exactly when
\[
\eta=0
\quad\Longleftrightarrow\quad
\rank_{\mathbb F_2}B=\nu,
\]
and the sector decomposition is irrelevant.

For the exceptional odd-$m$ branches, the sector assignment can be
determined without enumerating all partitions of $S$. Associate with
each signal vertex $v$ the row signature
\[
\sigma(v)=W_{v,:}\in\mathbb F_2^\eta.
\]
Every vector in $\ker B^{\mathsf T}$ is constant on a proposed sector
if and only if all vertices in that sector have the same row signature.
Hence a sector decomposition into blocks of size $m$ exists precisely
when every row-signature class has cardinality divisible by $m$.
This condition is independent of the chosen kernel basis, because a
basis change right-multiplies $W$ by an invertible binary matrix and
therefore preserves equality of row signatures.

The encoder-dependent conditions are then tested directly on the
kernel-basis matrix:
\[
(A_S+I_\nu)W=0,
\qquad
(P,Q)\in\{(X,Z),(Z,X)\},
\]
or
\[
A_SW=0,
\qquad
(P,Q)\in\{(Y,Z),(Z,Y)\}.
\]
For $XZ/ZX$, the even-sector-parity condition follows automatically
from the first relation for odd $m$, as shown in
Sec.~\ref{subsec:main_graph_kernel}. If the appropriate matrix
condition holds and every row-signature class has size divisible by
$m$, splitting each class into blocks of size $m$ produces a valid
sector decomposition. These tests are exactly the fixed-cut
implementation of Eq.~\eqref{eq:main_gsecc_definition}; no search over
sector partitions or additional stabilizer enumeration is required.
The complete algorithm and row-signature criterion are given in
Appendix~\ref{subsec:supp_gsecc}.

Computing $B$, a kernel basis $W$, the encoder-dependent matrix
condition, and the row-signature classes requires polynomial-time
binary linear algebra. Gaussian elimination gives an
$O(\nu^3)$ upper bound for a fixed oriented cut. The graph-level
problem still requires a search over the $\binom{2\nu}{\nu}$ oriented balanced cuts, so a direct exhaustive implementation has
the upper bound $O\!\left[
\binom{2\nu}{\nu}\nu^3
\right].$

This is an exhaustive-search bound for the certification procedure,
not a complexity classification of GSECC. Graph automorphisms can
further reduce the number of inequivalent cuts that need to be tested.

A successful GSECC test therefore supplies both the oriented cut and,
when needed, an encoder-compatible sector decomposition. These are
precisely the structural data required for recovery. We next show how
any certified realization can be converted into an explicit Clifford
decoder.

\subsection{Constructive Clifford decoding for every certified resource}
\label{subsec:constructive_decoding}

A certified realization satisfies
Theorem~\ref{thm:graph_exact_main} and therefore admits exact recovery
from every $S_i\mathcal N$. For graph-state resources this recovery can
be constructed entirely within the Clifford formalism.

Define the encoding isometry
\[
V_G\ket{\psi}_A
=
U_{P,Q}^{(m,k)}
\bigl(
\ket{\psi}_A\otimes\ket{G}_{S\mathcal N}
\bigr).
\]
Since $\ket G$ is a stabilizer state and
$U_{P,Q}^{(m,k)}$ is Clifford, the image of $V_G$ is a
$2^k$-dimensional stabilizer code in $AS\mathcal N$
~\cite{Gottesman1997,DehaeneDeMoor2003}.

For recovery from $S_i\mathcal N$, define
\[
\mathcal R_i=S_i\mathcal N,
\qquad
\mathcal E_i=AS_{\bar i},
\]
as the authorized and erased subsystems, respectively.
Theorem~\ref{thm:graph_exact_main} gives
\[
I(\widetilde A:\mathcal E_i)_\Omega=0,
\qquad
i=1,\ldots,m,
\]
so $\mathcal E_i$ is an exactly correctable erasure region for the
stabilizer code defined by $V_G$.

Explicit recovery operations also play a central role in graph-state
secret-sharing constructions~\cite{MarinMarkhamPerdrix2013}. Here the
decoder follows instead from exact correctability of the
encoder-defined erasure region. The stabilizer cleaning lemma implies
that representatives of the encoded input Pauli operators can be
chosen with support entirely inside the authorized subsystem.

\begin{proposition}[Constructive Clifford decoder]
\label{prop:constructive_decoder_main}
Let a graph-state realization satisfy
Eq.~\eqref{eq:main_exact_graph_condition}. Then, for every
$i=1,\ldots,m$, there exist Pauli representatives
$\overline X_j^{(i)}$ and $\overline Z_j^{(i)}$, supported entirely on
$\mathcal R_i=S_i\mathcal N$, such that
\[
\overline X_j^{(i)}V_G
=
V_GX_{A_j},
\qquad
\overline Z_j^{(i)}V_G
=
V_GZ_{A_j},
\qquad
j=1,\ldots,k.
\]
These representatives obey the canonical Pauli symplectic relations.
Consequently, there exists a Clifford unitary $C_i$ and a
factorization
\[
\mathcal H_{\mathcal R_i}
\simeq
\mathcal H_{A_i'}\otimes\mathcal H_{\mathcal J_i},
\qquad
\dim A_i'=2^k,
\]
such that
\[
\begin{aligned}
C_i\overline X_j^{(i)}C_i^\dagger
&=
X_{A_j'}\otimes I_{\mathcal J_i},
\\
C_i\overline Z_j^{(i)}C_i^\dagger
&=
Z_{A_j'}\otimes I_{\mathcal J_i}.
\end{aligned}
\]
The CPTP map $\mathcal D_i(\rho)
=
\Tr_{\mathcal J_i}
\!\left(
C_i\rho C_i^\dagger
\right)$ therefore satisfies
$\mathcal D_i\circ\mathcal Q_i
=
\operatorname{id}_A.$

\end{proposition}

The construction follows from the stabilizer cleaning lemma
~\cite{BravyiTerhal2009}. Exact correctability of $\mathcal E_i$
allows every encoded input Pauli to be multiplied by code stabilizers
until all support on $\mathcal E_i$ is removed. In binary symplectic
form, let $\ell_{j,\mu}$, with $\mu\in\{X,Z\}$, represent an encoded
input Pauli, let the rows of $H_G$ generate the stabilizer of the image
code, and let $\Pi_{\mathcal E_i}$ restrict a Pauli vector to the
erased subsystem. A cleaned representative is obtained by solving
\[
\Pi_{\mathcal E_i}
\left[
\ell_{j,\mu}
+
c_{j,\mu}^{(i)}H_G
\right]
=0
\qquad
\text{over }\mathbb F_2.
\]
Eq.~\eqref{eq:main_exact_graph_condition} guarantees solvability for
every $j$ and $\mu$. Symplectic Gaussian elimination then maps the
resulting $2k$ authorized representatives to standard Pauli generators
and synthesizes $C_i$
~\cite{DehaeneDeMoor2003,aaronson2004improved}. The complete cleaning
argument and Clifford-synthesis procedure are given in Appendices~\ref{subsec:supp_constructive_decoder} and~\ref{subsec:supp_clifford_synthesis}.

Proposition~\ref{prop:constructive_decoder_main} applies to every
valid branch in Table~\ref{tab:main_graph_classes}, including the
rank-deficient exceptional odd-$m$ branches. Thus reduced
signal--noise entanglement does not imply approximate or probabilistic
recovery: whenever the complete cut kernel is encoder compatible, the
authorized channels remain exactly reversible, with the saturation
properties already stated in
Eq.~\eqref{eq:main_capacity_saturation}.

\subsection{Closed-form Bell reduction for the full-rank branch}
\label{subsec:explicit_full_rank_decoder}

The preceding construction yields a Clifford decoder for every
certified graph-state resource. For a full-rank cut, the graph
structure permits a stronger closed-form statement: the resource can
be reduced to the canonical maximally entangled state by a Clifford
acting only on $\mathcal N$.

For a symmetric binary matrix $D$ with zero diagonal, define
\[
CZ_D
=
\prod_{a<b}CZ_{ab}^{D_{ab}},
\]
and for an invertible binary matrix $M$ define
\[
L_M\ket{x}
=
\ket{Mx},
\]
with arithmetic over $\mathbb F_2$. The map $L_M$ is a reversible
linear Clifford transformation
~\cite{DehaeneDeMoor2003,aaronson2004improved}.

\begin{observation}[Explicit noise-side Bell reduction]
\label{obs:explicit_bell_reduction_main}
Let the adjacency matrix of the graph resource relative to a balanced
cut $S:\mathcal N$ be
\[
\Gamma=
\begin{pmatrix}
A_S & B\\
B^{\mathsf T} & A_{\mathcal N}
\end{pmatrix},
\]
and suppose that $B$ is invertible over $\mathbb F_2$. Order the noise
vertices according to the chosen ordering of the signal vertices.
Then the noise-side Clifford
\begin{equation}
U_{\mathcal N}
=
CZ_{A_S}^{(\mathcal N)}
L_{B^{-\mathsf T}}
H_{\mathcal N}^{\otimes\nu}
CZ_{A_{\mathcal N}}^{(\mathcal N)}
\label{eq:main_explicit_bell_reduction}
\end{equation}
satisfies
\[
(I_S\otimes U_{\mathcal N})
\ket{G}_{S\mathcal N}
=
\ket{\Phi_{2^\nu}}_{S\mathcal N},
\]
where
\[
\ket{\Phi_{2^\nu}}
=
2^{-\nu/2}
\sum_{x\in\mathbb F_2^\nu}
\ket{x}_S\ket{x}_{\mathcal N}.
\]
\end{observation}

The reduction follows directly from the graph-state amplitude
representation. First,
$CZ_{A_{\mathcal N}}^{(\mathcal N)}$ removes the quadratic phase from
noise-internal edges. The Hadamard layer converts the cut phase
$(-1)^{s^{\mathsf T}Bn}$ into the computational-basis correlation
$\ket{B^{\mathsf T}s}_{\mathcal N}$, after which
$L_{B^{-\mathsf T}}$ maps the noise register to $\ket{s}_{\mathcal N}$.
Finally, $CZ_{A_S}^{(\mathcal N)}$ reproduces the signal-internal
quadratic phase on the correlated noise basis state, cancelling the
remaining signal phase. A complete derivation is given in
Appendix~\ref{subsec:supp_explicit_bell_reduction}.

Observation~\ref{obs:explicit_bell_reduction_main} shows that every
full-cut-rank graph resource is not only maximally entangled across
$S:\mathcal N$, but is related to the canonical Bell resource by an
explicit Clifford acting solely on the noise subsystem. This realizes
constructively the general noise-side-unitary equivalence of
purifications and agrees with the canonical Bell-pair decomposition of
bipartite stabilizer entanglement~\cite{FattalEtAl2004}.

Because $U_{\mathcal N}$ acts only on $\mathcal N$, it commutes with
the sector-wise encoder. Recovery in the full-rank branch can therefore
be written as
\[
\mathcal D_i
=
\mathcal D_i^{\rm Bell}
\circ
\operatorname{Ad}_{I_{S_i}\otimes U_{\mathcal N}},
\qquad
\operatorname{Ad}_V(X)=VXV^\dagger,
\]
where $\mathcal D_i^{\rm Bell}$ is an exact decoder for the canonical
Bell resource. Thus
Eq.~\eqref{eq:main_explicit_bell_reduction} gives a closed-form
canonicalization of the resource; it is not, by itself, the complete
recovery map. The latter may be chosen Clifford by
Proposition~\ref{prop:constructive_decoder_main}.

The same transformation also identifies when this particular
canonicalization is local on the individual noise qubits.

\begin{observation}[Locality of the Bell reduction]
\label{obs:main_decoder_locality}
The Clifford $U_{\mathcal N}$ in
Eq.~\eqref{eq:main_explicit_bell_reduction} is a product of
single-qubit Clifford operations, up to a permutation of the noise
qubits, if and only if
\[
A_S=A_{\mathcal N}=0,
\qquad
B\ \text{is a permutation matrix}.
\]
Equivalently, relative to the chosen cut, the graph is a perfect
matching between $S$ and $\mathcal N$.
\end{observation}

The proof follows from the binary symplectic action of
Eq.~\eqref{eq:main_explicit_bell_reduction} on the single-qubit Pauli
generators and is given in Appendix~\ref{subsec:supp_bell_reduction_locality}. This statement concerns locality of the
specific noise-side canonicalization above; it does not characterize
the locality of all possible recovery maps on $S_i\mathcal N$.

Finally, the transformation in
Eq.~\eqref{eq:main_explicit_bell_reduction} can be synthesized
efficiently. Gaussian-elimination-based CNOT synthesis implements
$L_{B^{-\mathsf T}}$ using $O(\nu^2)$ elementary gates
~\cite{PatelMarkovHayes2008}, while each controlled-$Z$ layer contains
at most $O(\nu^2)$ gates. Under unrestricted connectivity, the
complete noise-side Bell reduction therefore has an
$O(\nu^2)$ Clifford-gate upper bound. Detailed synthesis procedures
and representative gate counts are given in Appendix~\ref{subsec:supp_clifford_synthesis}.

\section{Resource landscape and representative families}
\label{sec:resource_landscape}

The preceding results provide an exact graph-state resource criterion
together with constructive certification and recovery. We now use this
framework to organize representative multipartite resources according
to three distinct properties: the amount of entanglement across the
signal--noise cut, the correlation structure responsible for that
entanglement, and its compatibility with the prescribed encoder.

\subsection{Full-rank benchmark and scalable graph-state families}
\label{subsec:scalable_graph_families}

The simplest sufficient condition is full cut rank. If a graph on
$2\nu$ vertices admits a balanced cut satisfying $\rank_{\mathbb F_2}B_{S\mathcal N}=\nu,$ then $\rho_S=\frac{I_S}{2^\nu},$ and the resource is valid for every factorization
$\nu=mk$ with $m\ge2$, independently of the sector decomposition and
of the two-Pauli encoder class.

Several standard graph families admit such cuts. A perfect matching
provides the canonical limiting case: choosing one endpoint of each
edge as signal and the other as noise makes
$B_{S\mathcal N}$ a permutation matrix. The associated graph state is
locally Clifford equivalent to $\nu$ Bell pairs, and
Observation~\ref{obs:explicit_bell_reduction_main} reduces it to the
canonical Bell resource using only local noise-side operations, up to
a permutation.

Full cut rank is not restricted to resources that resemble independent
Bell pairs. Suitable balanced cuts are also invertible for even paths,
even cycles, and rectangular cluster graphs
$G_{2a\times b}^{\mathrm{cl}}$ with $ab=\nu$. Linear and higher-
dimensional cluster states therefore provide scalable multipartite
resources despite their nontrivial graph structure
~\cite{briegel2001persistent,raussendorf2001one,
raussendorf2003measurement,hein2004}. Explicit certifying cuts and
their cut matrices are given in Appendix~\ref{subsec:supp_scalable_families}.

In these examples maximal signal--noise entanglement coexists with
nontrivial internal graph structure: the cut matrix is invertible even
when the adjacency blocks $A_S$ and $A_{\mathcal N}$ are nonzero.
Thus full-rank validity is a property of the operational bipartition,
not of a particular Bell-pair geometry.

\begin{figure*}[t]
\centering

\captionsetup[subfigure]{justification=centering,singlelinecheck=false}

\subcaptionbox{
    Path / linear cluster
    \label{fig:gcl}
}[0.22\textwidth]{%
    \centering
    \QuantumIncludeFigure[
        width=\linewidth,
        trim=8 8 8 8,
        clip
    ]{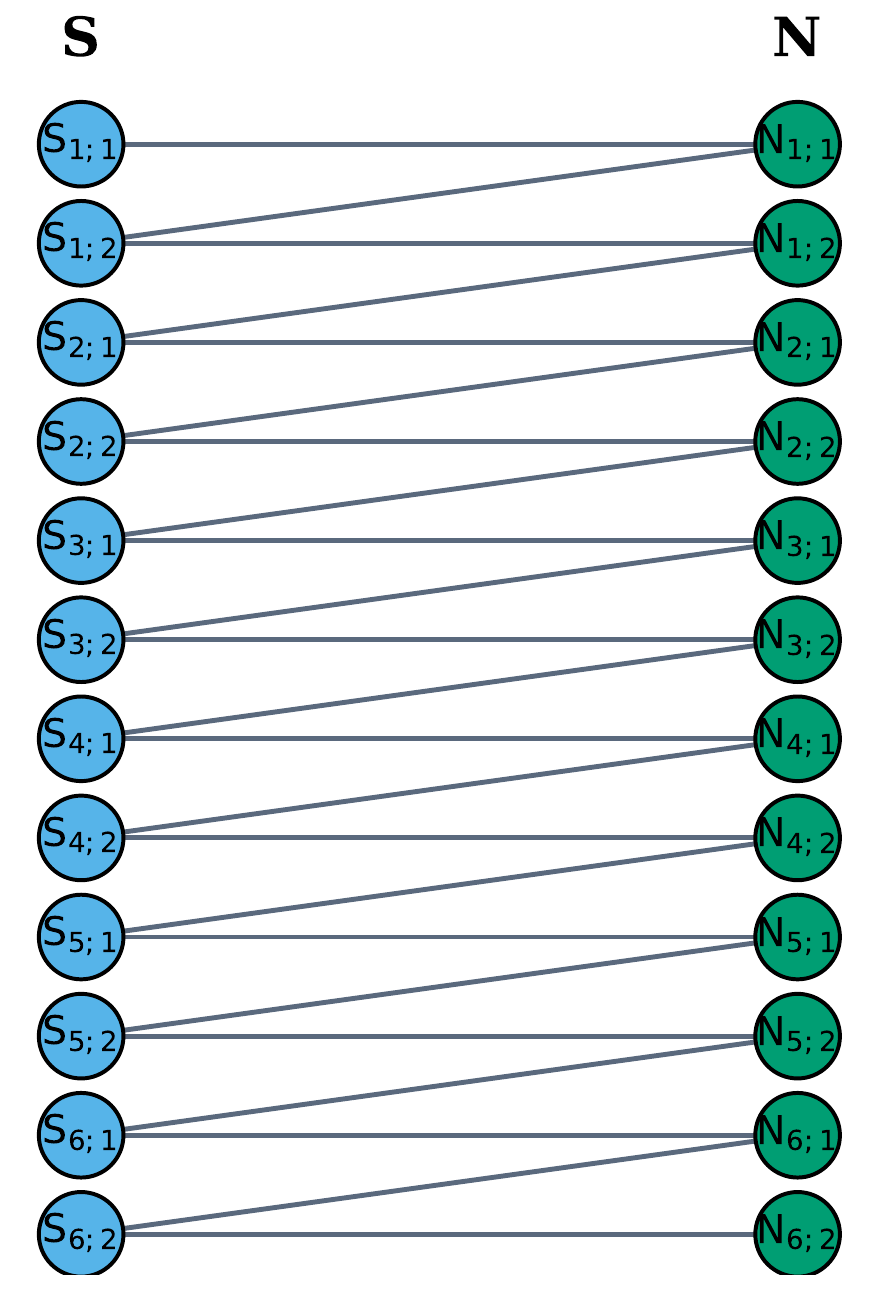}%
}
\hfill
\subcaptionbox{
    Fixed $G(24,0.35)$ sample
    \label{fig:sparse}
}[0.22\textwidth]{%
    \centering
    \QuantumIncludeFigure[
        width=\linewidth,
        trim=8 8 8 8,
        clip
    ]{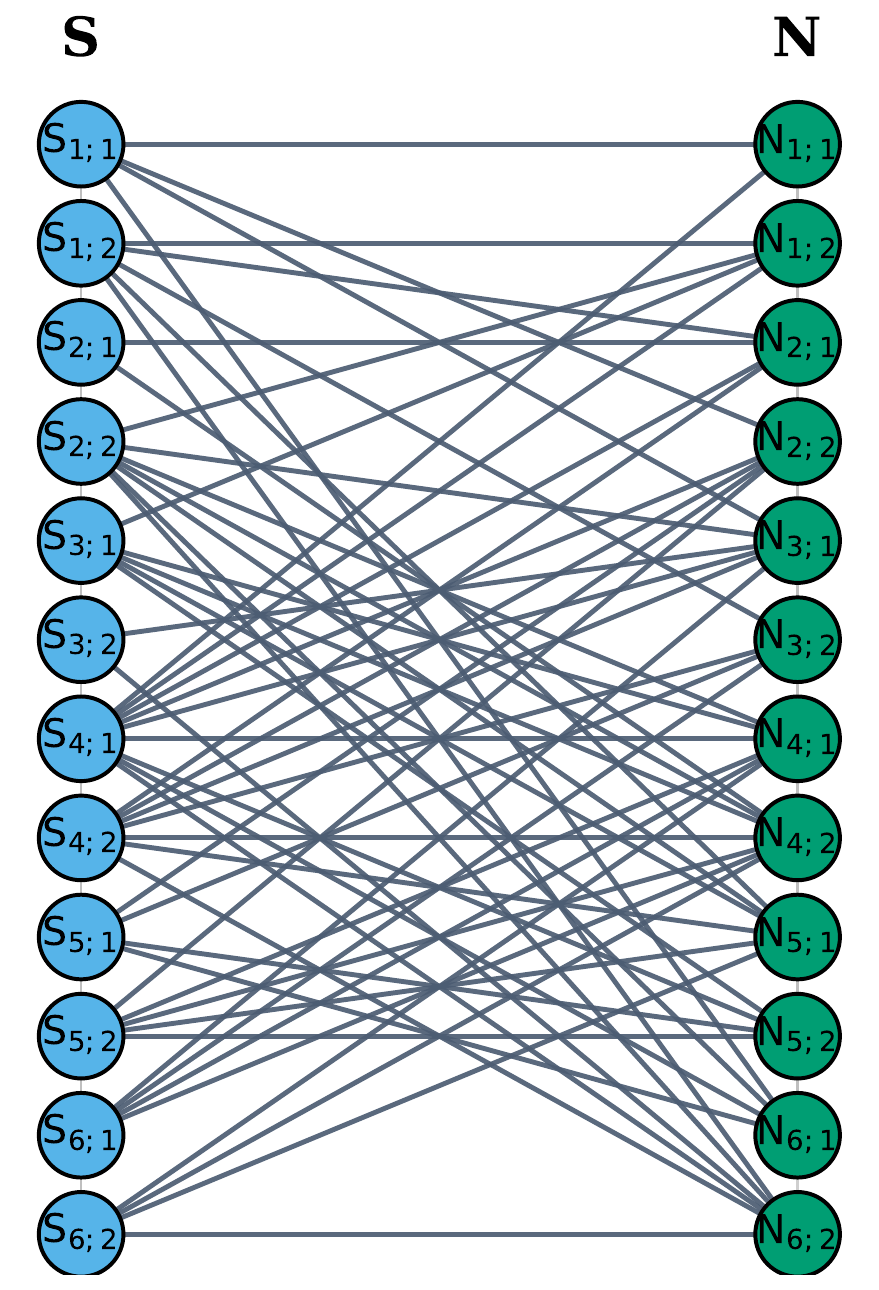}%
}
\hfill
\subcaptionbox{
    Cycle $C_{24}$
    \label{fig:cycle}
}[0.22\textwidth]{%
    \centering
    \QuantumIncludeFigure[
        width=\linewidth,
        trim=8 8 8 8,
        clip
    ]{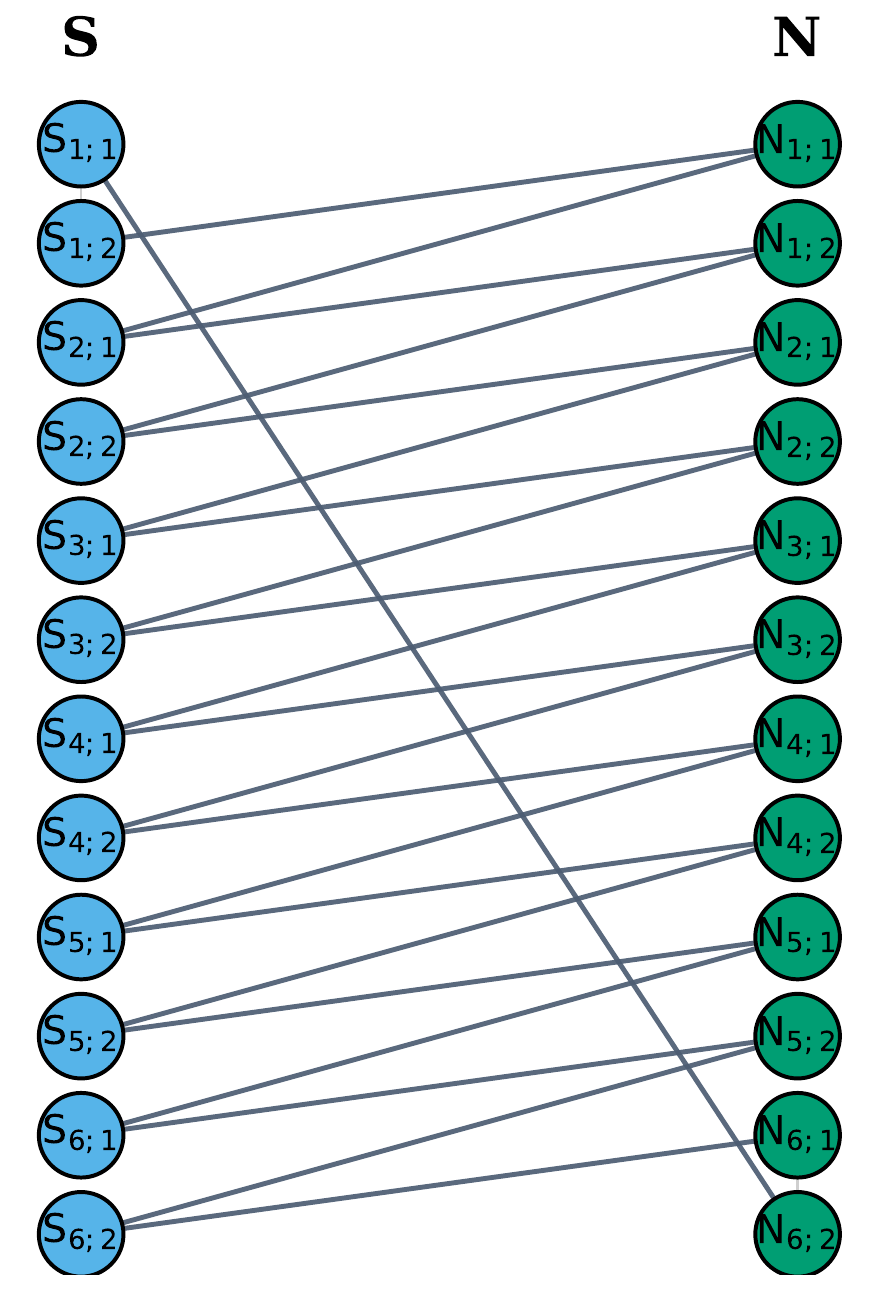}%
}

\vspace{1.2mm}

\subcaptionbox{
    Complete graph $K_{24}$
    \label{fig:complete}
}[0.31\textwidth]{%
    \centering
    \QuantumIncludeFigure[
        width=\linewidth,
        trim=8 8 8 8,
        clip
    ]{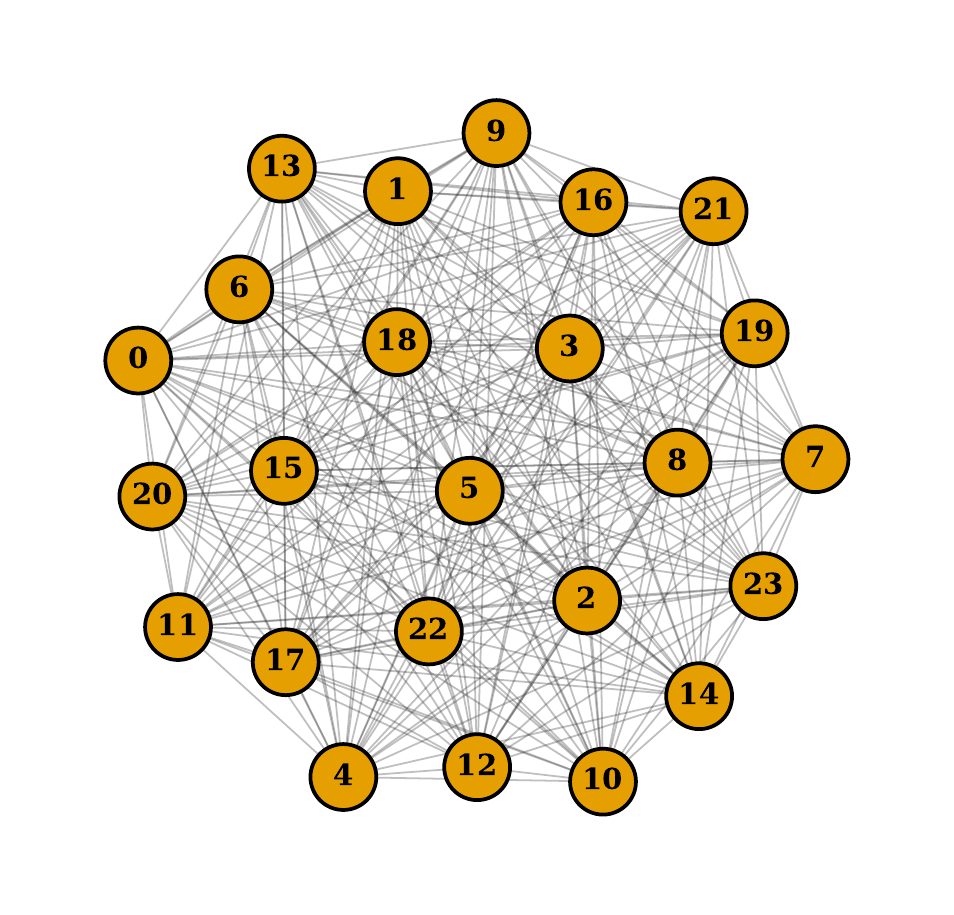}%
}
\hfill
\subcaptionbox{
    Fixed $G(24,0.9)$ sample
    \label{fig:dense}
}[0.31\textwidth]{%
    \centering
    \QuantumIncludeFigure[
        width=\linewidth,
        trim=8 8 8 8,
        clip
    ]{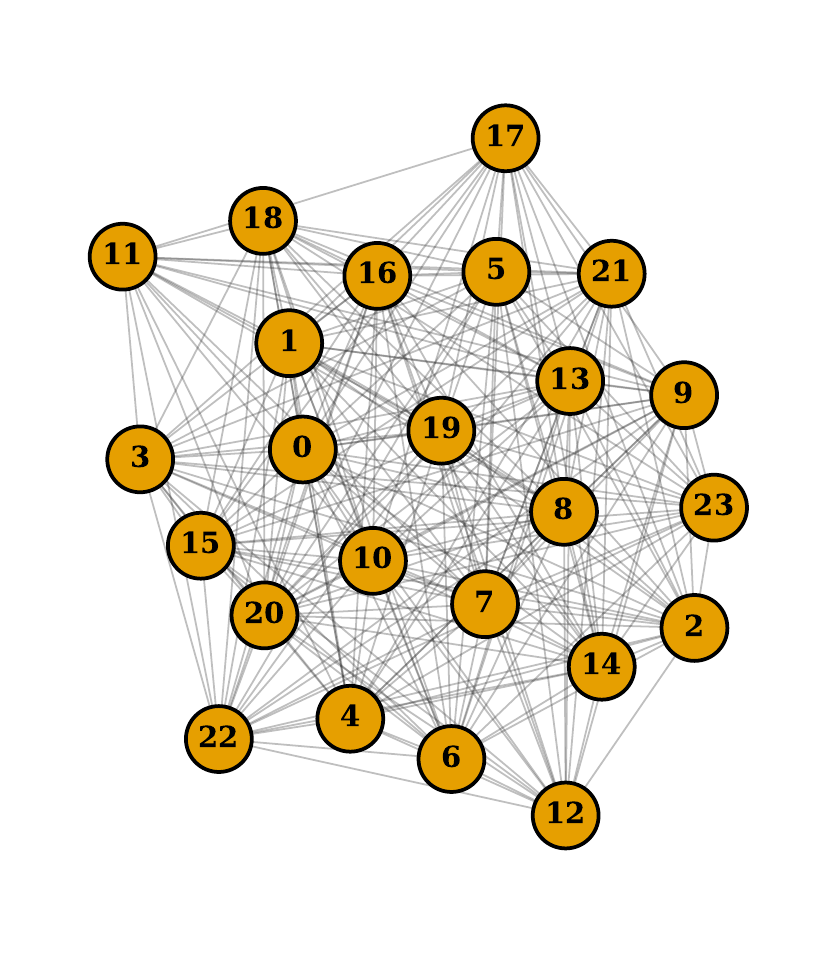}%
}
\hfill
\subcaptionbox{
    Star / GHZ graph
    \label{fig:star}
}[0.31\textwidth]{%
    \centering
    \QuantumIncludeFigure[
        width=\linewidth,
        trim=8 8 8 8,
        clip
    ]{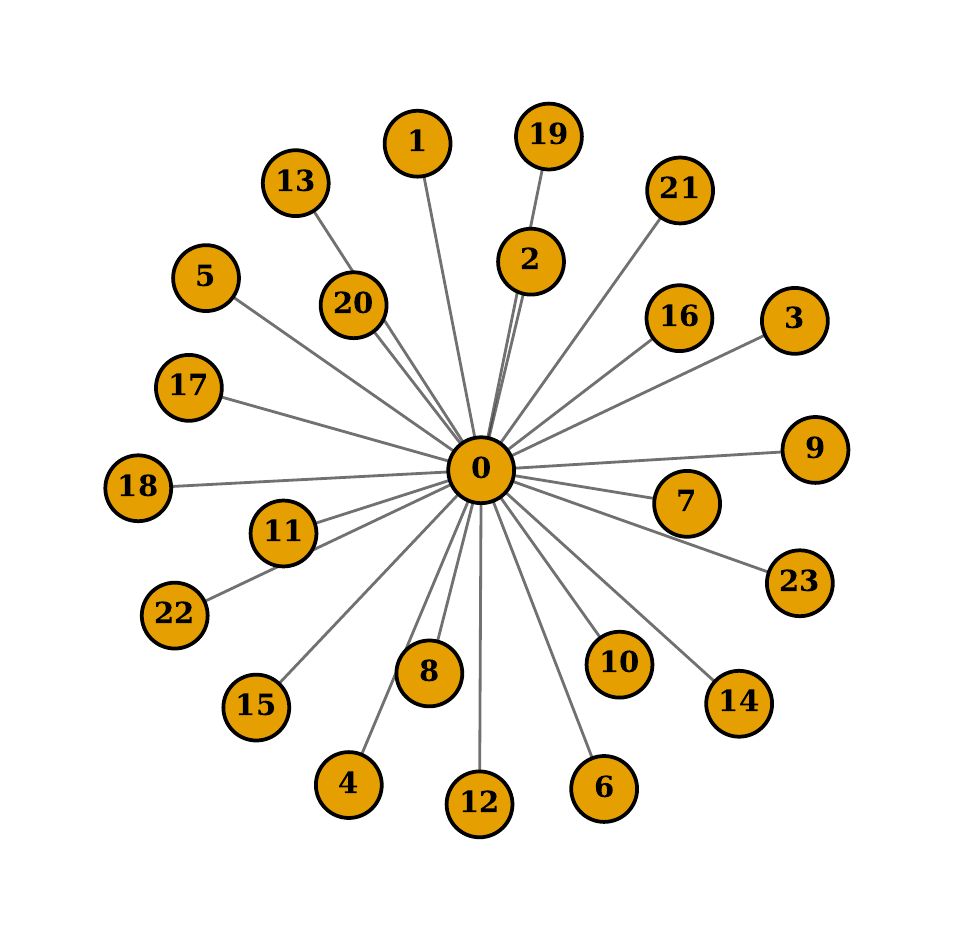}%
}

\caption{\footnotesize
Representative $24$-qubit graph states for the
$(m,k)=(2,6)$ task.
\textbf{Top row:} certified resources:
(a) a path / linear cluster,
(b) a fixed sparse Erd\H{o}s--R\'enyi sample
$G(24,0.35)$~\cite{erdos1959,gilbert1959},
and (c) the cycle $C_{24}$.
Each admits a balanced cut with
$\rank_{\mathbb F_2} B_{S\mathcal N}=12$.
\textbf{Bottom row:} excluded resources:
(d) the complete graph $K_{24}$,
(e) a fixed dense Erd\H{o}s--R\'enyi sample
$G(24,0.9)$~\cite{erdos1959,gilbert1959},
and (f) the star/GHZ graph.
For each excluded graph, the maximum cut rank over all balanced
signal--noise bipartitions is strictly smaller than $12$.
The Erd\H{o}s--R\'enyi examples are fixed sampled realizations and
should not be interpreted as evidence for an edge-density threshold;
their explicit adjacency data are given in Appendix~\ref{subsec:supp_24q_examples}.
}
\label{fig:main_graph_certification}

\end{figure*}
\subsection{Exact rank-deficient resources and encoder dependence}
\label{subsec:resource_encoder_dependence}

The more distinctive consequence of
Theorem~\ref{thm:graph_exact_main} appears when the cut is
rank deficient. In the exceptional odd-$m$ branches, exact recovery
can survive even when
\[
\rank_{\mathbb F_2}B_{S\mathcal N}<mk,
\]
provided that every missing cut direction belongs to the exceptional
subspace selected by the fixed encoder.

The smallest example already exhibits both exact rank-deficient
recovery and encoder dependence. Let $m=3$, $k=1$, and choose
\[
A_S=A_{\mathcal N}=0,
\qquad
B=
\begin{pmatrix}
1&0&0\\
0&1&0\\
1&1&0
\end{pmatrix}.
\]
Then
\[
\rank_{\mathbb F_2}B=2,
\qquad
\ker(B^{\mathsf T})
=
\operatorname{span}\{(1,1,1)^{\mathsf T}\}.
\]

\begin{observation}[Rank-deficient exact recovery and encoder dependence]
\label{obs:main_rank_deficient_encoder_dependence}
For the fixed graph-state representative and cut above, the resource
is valid for the $YZ/ZY$ encoder class and invalid for the $XZ/ZX$
class. In the valid branch, every prescribed subsystem
$S_i\mathcal N$ recovers the complete input qubit exactly, although
the resource contains only two ebits across the signal--noise cut.
\end{observation}

The mechanism follows directly from the unique nonzero kernel vector
$x=(1,1,1)^{\mathsf T}$. Since $A_S=0$,
\[
K_S(x)\doteq X^{\otimes3}.
\]
For the $YZ/ZY$ class, $R=\pm X$, so $X^{\otimes3}$ is precisely the
exceptional complete-sector direction. Hence
\[
\ker(B^{\mathsf T})
\subseteq
\mathcal E_{Y,Z}^{(3,1)}(A_S)
\]
and Theorem~\ref{thm:graph_exact_main} gives exact recovery from every
$S_i\mathcal N$.

For the $XZ/ZX$ class, by contrast, the exceptional direction is
$Y^{\otimes3}$. The same kernel vector is then nonexceptional, so the
exact kernel inclusion fails. The physical graph, the bipartition,
the cut rank, and the complete kernel are unchanged; only the fixed
encoder frame has changed.

This comparison is relative to the fixed graph-state Pauli frame.
As discussed in Sec.~\ref{subsec:main_graph_kernel}, simultaneous
single-qubit Clifford rotations of the input and signal systems map
between Pauli-axis encoder classes and compatible resource
representatives without changing the signal--noise entanglement.
Observation~\ref{obs:main_rank_deficient_encoder_dependence} therefore
demonstrates compatibility of a fixed graph-state representative with
different fixed encoder frames, rather than an absolute inequivalence
of the encoder families.

The example also saturates the architecture-independent lower bound:
\[
S(\rho_S)
=
\rank_{\mathbb F_2}B
=
2
=
(m-1)k.
\]
Thus exact three-output recovery is achieved with the minimum
signal--noise entanglement allowed by
Proposition~\ref{prop:all_clone_entropy_main}. The missing cut ebit
corresponds exactly to an exceptional kernel direction and therefore
does not obstruct complementary decoupling.

This construction should also be distinguished from counting the
transformed source register as an additional encrypted output in the
original encrypted-cloning protocol~\cite{yamaguchi2026encrypted}.
Here the three designated outputs are the prescribed signal registers
$S_1,S_2,S_3$, each recovery receives the complete common key
$\mathcal N$, and validity is established directly from the initial
resource marginal. The example therefore shows that the
architecture-independent entanglement floor is attainable within the
fixed recovery architecture considered here.

The example separates the three structural quantities entering the
classification. The cut rank determines the signal--noise
entanglement, the cut kernel identifies the residual signal-side
stabilizer correlations, and the fixed encoder determines whether
those correlations are compatible with exact recovery. A connected
representative and the explicit stabilizer verification are given in
Appendix~\ref{subsec:supp_rank_deficient_examples}.

\subsection{Representative finite-size certification for full rank }
\label{subsec:finite_graph_landscape}

We next illustrate the certification procedure on representative
$24$-qubit graph states with $\nu=12,
(m,k)=(2,6).$ Because $m$ is even,
Corollary~\ref{cor:main_graph_parity_classes} reduces validity to the
existence of a balanced signal--noise cut with
\[
\rank_{\mathbb F_2}B_{S\mathcal N}=12.
\]
Thus this example isolates the full-rank branch of the classification:
sector assignments and Pauli-frame distinctions play no role once an
invertible cut matrix is found.

Fig~\ref{fig:main_graph_certification} shows representative
certified and excluded graph states. The examples are chosen to
illustrate that full cut rank can occur in graph states with very
different connectivity, while highly connected graphs need not admit
a valid balanced cut. Table~\ref{tab:main_24q_resources} summarizes
the maximal balanced-cut rank and, for certified resources, the
structure of the corresponding noise-side Bell reduction. Explicit
certifying cuts, cut matrices, balanced-cut search data, and circuit
syntheses are given in Appendix~\ref{subsec:supp_24q_examples}
and the accompanying implementation~\cite{Roy_Gupta2026Github}.

For every certified full-rank cut, the same binary data
$(A_S,B,A_{\mathcal N})$ that establish validity also determine the
noise-side Clifford
$U_{\mathcal N}$ of
Eq.~\eqref{eq:main_explicit_bell_reduction}. Composing this
canonicalization with an exact decoder for the Bell resource gives
the authorized recovery map. Certification and constructive recovery
are therefore obtained from the same graph data.

The certified examples all possess maximal Schmidt rank
$2^{12}$ across an appropriate balanced cut and therefore support the
same exact $(2,6)$ recovery task. Their recovery structure is
nevertheless different. For the perfect matching, the noise-side Bell
reduction is a product of single-qubit Clifford operations up to a
permutation. The path, cycle, rectangular cluster, and sampled sparse
graph instead require collective Clifford processing.
Observation~\ref{obs:main_decoder_locality} makes this distinction
exact: the canonicalization of
Eq.~\eqref{eq:main_explicit_bell_reduction} is local, up to a
permutation, if and only if the chosen graph realization is a perfect
matching across $S:\mathcal N$.

The excluded examples fail for different structural reasons. For the
complete graph $K_{24}$, every balanced cut has
\[
B_{S\mathcal N}=J_{12},
\qquad
\rank_{\mathbb F_2}B_{S\mathcal N}=1.
\]
The star graph likewise has balanced-cut rank one. Its graph state lies
in the local-Clifford orbit of the GHZ state, and local complementation
relates the star and complete graphs
~\cite{greenberger1989going,hein2004,vandennest2004graphical}.
Both therefore fail the even-$m$ full-rank criterion.

The two Erd\H{o}s--R\'enyi entries serve only as finite sampled
examples~\cite{erdos1959,gilbert1959}. The chosen
$G(24,0.35)$ realization admits a full-rank balanced cut, whereas the
chosen $G(24,0.9)$ realization has maximum balanced-cut rank $11$.
No monotonic relation between edge density and encrypted-cloning
validity is inferred from these two samples. Rather, they illustrate
that neither edge density nor connectivity alone determines the
independent binary correlations available across the operational cut.

Finally, the full-rank canonicalization used in this subsection is
distinct from recovery for the exceptional rank-deficient resources
discussed in Sec.~\ref{subsec:resource_encoder_dependence}. If
\[
\rank_{\mathbb F_2}B_{S\mathcal N}<\nu,
\]
the Schmidt rank across $S:\mathcal N$ is smaller than $2^\nu$, so no
unitary acting only on $\mathcal N$ can convert the resource into
$\ket{\Phi_{2^\nu}}$. Exact recovery can nevertheless remain possible
when the complete cut kernel is encoder compatible. In that regime,
recovery is provided by the general authorized Clifford construction
of Proposition~\ref{prop:constructive_decoder_main}, rather than by a
noise-side Bell canonicalization.

\begin{table*}[t]
\caption{\footnotesize
Representative $24$-qubit graph-state resources for
$(m,k)=(2,6)$. The reported rank is maximized over balanced
$12|12$ cuts. For every certified resource, the displayed operation
is the noise-side canonicalization of
Eq.~\eqref{eq:main_explicit_bell_reduction}; the complete authorized
decoder is obtained by composing it with the Bell-resource decoder.
``Local'' refers only to locality of this noise-side canonicalization,
up to permutation of the noise qubits. Explicit matrices and circuit
syntheses are given in Appendix~\ref{subsec:supp_24q_examples}.}
\label{tab:main_24q_resources}
\centering
\resizebox{\textwidth}{!}{%
\begin{tabular}{lcccc}
\toprule
Graph resource
&
$\displaystyle
\max_{S:\mathcal N}
\rank_{\mathbb F_2}B_{S\mathcal N}$
&
Status
&
Noise-side reduction
&
Decoder type
\\
\midrule

Perfect matching (locally Clifford equivalent to $\nu$ Bell pairs)
&
$12$
&
Valid
&
$H^{\otimes12}$
&
Local
\\

Path / linear cluster
&
$12$
&
Valid
&
$L_{M_P}H^{\otimes12}$
&
Collective
\\

Cycle $C_{24}$
&
$12$
&
Valid
&
$CZ_{E_L}H^{\otimes12}CZ_{E_R}$
&
Collective
\\

$G_{2\times12}^{\mathrm{cl}}$
&
$12$
&
Valid
&
$CZ_{P_{12}}H^{\otimes12}CZ_{P_{12}}$
&
Collective
\\

Fixed $G(24,0.35)$ sample
~\cite{erdos1959,gilbert1959}
&
$12$
&
Valid
&
$\begin{aligned}
CZ_{A_S^{\rm ER}}
L_{M_{\rm ER}}
H^{\otimes12}
CZ_{A_{\mathcal N}^{\rm ER}}
\end{aligned}$
&
Collective
\\

Fixed $G(24,0.9)$ sample
~\cite{erdos1959,gilbert1959}
&
$11$
&
Invalid
&
---
&
---
\\

Complete graph $K_{24}$
&
$1$
&
Invalid
&
---
&
---
\\

Star / GHZ graph
&
$1$
&
Invalid
&
---
&
---
\\
\bottomrule
\end{tabular}%
}
\end{table*}
\subsection{Non-graph resources and no-go constraints}
\label{subsec:beyond_graph_resources}

Theorem~\ref{thm:exact_general_resource_main} applies to arbitrary
pure resources and is not restricted to graph or stabilizer states.
To illustrate the distinction between multipartite entanglement and
the correlation structure required for encrypted recovery, we consider
Dicke states, a canonical permutation-symmetric family
~\cite{dicke1954,MorenoParisio2018,MunizziSchnitzer2024}.

For an $n$-qubit system, the Dicke state with $r$ excitations is
\[
\ket{D_n^{(r)}}
=
\binom nr^{-1/2}
\sum_{\substack{z\in\{0,1\}^{n}\\ |z|=r}}
\ket z.
\]
For the encrypted-cloning resource we set
\[
n=2\nu=2mk.
\]
Permutation symmetry makes every balanced $\nu|\nu$ bipartition
equivalent. The Schmidt decomposition across such a cut implies
\[
\rank\rho_S\le\nu+1,
\qquad
S(\rho_S)\le\log_2(\nu+1).
\]
Thus even genuinely multipartite-entangled members of the Dicke
family have balanced signal marginals supported on only a small
subspace of the full $2^\nu$-dimensional Hilbert space.

For the prescribed sector-wise two-Pauli architecture, the exclusion
is exact.

\begin{observation}[Dicke-state obstruction]
\label{obs:main_dicke_obstruction}
For every $m\ge2$, $k\ge1$, and
$r=0,\ldots,2mk$, the Dicke state
$\ket{D_{2mk}^{(r)}}$ fails the exact pure-resource criterion of
Theorem~\ref{thm:exact_general_resource_main} for every sector-wise
two-Pauli encoder $U_{P,Q}^{(m,k)}$.
\end{observation}

For even $m$, the obstruction is already visible from the rank of the
signal marginal. Theorem~\ref{thm:exact_general_resource_main}
requires
\[
\rho_S=\frac{I_S}{2^{mk}},
\]
whereas a balanced Dicke marginal satisfies
\[
\rank\rho_S\le mk+1<2^{mk}.
\]

For odd $m$, the obstruction is structural. Theorem~
\ref{thm:exact_general_resource_main} permits nonidentity signal
components only through products of complete-sector operators
$R_{\mathsf S_j}^{\otimes m}$. If one traces out even a single qubit
from an active sector, that contribution vanishes because
$\Tr R=0$. Consequently, every reduced state supported strictly
within a single sector is maximally mixed. In particular, for any
two distinct qubits $S_{j,a}$ and $S_{j,b}$ in the same sector,
\[
\rho_{S_{j,a}S_{j,b}}
=
\frac{I_4}{4}.
\]

Dicke states violate this condition for every excitation number. At
$r=0$ and $r=n$, the two-qubit marginal is a pure product state,
whereas for $0<r<n$ it retains excitation-number correlations and a
nonzero
$\ket{01}\!\bra{10}+\ket{10}\!\bra{01}$ coherence. The explicit
two-qubit reduced-state calculation is given in Appendix~\ref{subsec:supp_dicke_resources}. Hence the entire Dicke family is
excluded for the fixed two-Pauli architecture, independently of the
entropy estimate.

A logically independent obstruction follows from
Proposition~\ref{prop:all_clone_entropy_main}, which assumes only an
arbitrary unitary encoder on $AS$ acting trivially on the common key
$\mathcal N$. Exact recovery from every $S_i\mathcal N$ requires
\[
S(\rho_S)\ge(m-1)k.
\]
Combining this with the Dicke entropy bound gives
\[
\begin{gathered}
(m-1)k>\log_2(mk+1)
\\[1mm]
\Longrightarrow\quad
\text{Dicke resource impossible}.
\end{gathered}
\]
This implication is architecture independent within the class of
unitary encoders considered here. It is only a necessary obstruction:
when the inequality is not violated, the entropy bound alone does not
establish the existence of a successful encoder. By contrast,
Observation~\ref{obs:main_dicke_obstruction} excludes the entire
Dicke family for the prescribed sector-wise two-Pauli architecture.

The single-output boundary shows the same distinction. For $m=1$ and
$k=1$, perfect concealment together with exact recovery would require
an $\operatorname{AME}(4,2)$ encoded Choi state, which does not exist
~\cite{HiguchiSudbery2000}. For $k>1$, suitable
$\operatorname{AME}(4,2^k)$ states do exist, so the unrestricted
single-output task is possible in principle
~\cite{HelwigEtAl2012}. Dicke resources nevertheless fail because
their balanced signal marginal satisfies
\[
\rank\rho_S\le k+1<2^k,
\]
and therefore cannot be maximally mixed.

The $W$ state is the single-excitation Dicke state, $\ket{W_{2\nu}}
=
\ket{D_{2\nu}^{(1)}}.$ Across every balanced cut it has Schmidt rank two and exactly one ebit
of signal--noise entanglement. Observation~
\ref{obs:main_dicke_obstruction} therefore excludes it for every
$m\ge2$ within the fixed two-Pauli architecture. Independently,
Proposition~\ref{prop:all_clone_entropy_main} excludes it for an
arbitrary unitary encoder whenever
\[
(m-1)k>1.
\]
The only point with $m\ge2$ not excluded by this entropy bound alone is
$(m,k)=(2,1),$ which remains excluded by the exact fixed-encoder criterion.

The Dicke family therefore exhibits two distinct failure mechanisms.
The first is structural: its signal correlations are incompatible with
the operator algebra permitted by the prescribed two-Pauli encoder.
The second is architecture independent: over a broad parameter range,
its signal--noise entanglement is insufficient for exact all-output
recovery.

Together with the complete and star/GHZ graph-state examples, these
results show that multipartite entanglement or graph connectivity
alone does not determine usefulness for encrypted cloning. For graph
states, the relevant correlation structure is captured exactly by the
cut kernel; for arbitrary pure resources, it is characterized by
Theorem~\ref{thm:exact_general_resource_main}.


\section{Discussion and Outlook}
\label{sec:discussion}

The canonical encrypted-cloning protocol uses maximally entangled
resources to provide several alternative recovery pathways for an
unknown quantum state. Our results show that, for suitable odd numbers
of outputs, maximal entanglement is not always necessary. For
$(m,k)=(3,1)$, a graph-state resource in the YZ/ZY branch permits
exact recovery with two ebits of signal--key entanglement rather than
the three ebits of the maximally entangled construction. This reduction
does not relax the recovery condition: the input is recovered exactly
from every authorized signal--key subsystem, and each signal alone
reveals no information about it. The construction also saturates the
entanglement lower bound derived without fixing the two-Pauli encoder.
It is therefore an example in which the resources needed for exact
recovery are less demanding, in entanglement entropy, than those used
in the canonical implementation.

The more general result is the characterization of when this is
possible. For the sector-wise two-Pauli architecture, our criterion is
necessary and sufficient for arbitrary pure resources, with necessity
holding for unrestricted completely positive trace-preserving recovery
maps. For even $m$, every nontrivial signal correlation is incompatible
with exact recovery from all prescribed outputs, forcing maximal
signal--key entanglement. For odd $m$, certain correlations may remain,
but only within the commuting algebra selected by the encoder. Thus,
two resources with the same signal--key entanglement can have different
recovery capabilities. The problem is not simply to find less entangled
states that happen to work, but to identify the correlations that may
survive the encoding without obstructing recovery.

This question has a direct connection to quantum error correction. The
initial resource and the encoder together define an isometric encoding,
and recovery from $S_i\mathcal N$ amounts to correctability of the
complementary erasure. The relation between recoverability and the
information available to complementary systems is well established
~\cite{SchumacherNielsen1996,Gottesman1997,
KretschmannKribsSpekkens2008}. Here, however, the same encoded input
must be recoverable from several prescribed regions. Requiring all
these erasures to be correctable imposes restrictions on the initially
shared resource: for arbitrary unitary encoders acting on $AS$ while
leaving the common key $\mathcal N$ untouched, we obtain
$S(\rho_S)\geq(m-1)k$, a spectral-degeneracy constraint on $\rho_S$,
and perfect concealment of every individual signal. These conditions
are necessary, not sufficient, for arbitrary encoders. Their value is
that they identify what follows from the recovery pattern itself,
before the additional restrictions of a particular encoder are imposed.

For graph states, the distinction between entanglement and
recoverability becomes an explicit algebraic test. The rank of the
signal--key cut matrix measures bipartite entanglement
~\cite{hein2004,hein2006entanglement,FattalEtAl2004}, whereas its
complete kernel identifies the residual signal-side stabilizer
correlations. A rank-deficient resource remains admissible precisely
when those correlations lie in the encoder-compatible exceptional
subspace. In a fixed graph-state Pauli frame, the same graph and
bipartition can therefore succeed for one encoder frame and fail for
another; simultaneous local Clifford transformations relate the
corresponding encoder--resource descriptions. Our \textsc{GSECC}
procedure turns the cut-kernel condition into a certification test, and
every certified realization admits an explicit Clifford decoder. The
test uses polynomial-time binary linear algebra for a fixed cut,
although the graph-wide search over cuts is combinatorial. In this
sense, the classification does more than recognize a successful
resource: for graph states, it also provides the recovery operation.

These results complement graph-state quantum secret sharing, where one
often starts with a sharing scheme and determines its authorized sets
and decoding operations~\cite{MarkhamSanders2008,Sarvepalli2012,
MarinMarkhamPerdrix2013}. We instead prescribe the encoder and recovery
regions and ask which resource states realize them. An open problem is
to choose a resource graph and an encoder jointly for a desired family
of authorized regions, subject to constraints on entanglement,
connectivity, or decoding operations. The present theorems answer the
pure-resource feasibility question for the two-Pauli architecture, not
this more general optimization problem. Our exclusions also show why a
feasibility test is needed: the Dicke family~\cite{dicke1954}, including
$W$ states~\cite{dur2000three}, fails the fixed two-Pauli criterion,
while the general entanglement bound independently excludes $W$ states
for arbitrary unitary encoders whenever $(m-1)k>1$. Insufficient
entanglement and incompatible correlations are distinct obstructions.
This accords with the task- and operation-dependent viewpoint of
quantum resource theories~\cite{ChitambarGour2019}, without requiring
us to posit a general resource theory for encrypted cloning.

A possible application is distributed quantum storage, where preparing
and preserving shared entanglement remain important considerations
~\cite{kimble2008quantum,wehner2018quantum,cuomo2020towards}. The
encrypted multicloud setting proposed for the original protocol
~\cite{yamaguchi2026encrypted} illustrates why multiple authorized
recovery pathways may be useful. For odd $m$, some of these pathways
can be realized exactly without maximal signal--key entanglement.
This does not yet establish a hardware advantage: preparing a different
resource may require additional gates, connectivity, or more difficult
recovery operations. Nor does the protocol provide unrestricted
cryptographic secrecy. Our criterion ensures that each individual
signal is concealed; several signals considered jointly can reveal
additional information, and in some cases the complete input.
Applications with potentially colluding parties therefore require
additional assumptions and an analysis of unauthorized sets, as in
quantum secret sharing~\cite{hillery1999quantum,
CleveGottesmanLo1999,Gottesman2000,ImaiEtAl2005,gianini2026access}.

The principal open theoretical question is whether similarly sharp
resource conditions can be found for more general encoding unitaries.
The entropy and spectral constraints remain necessary when the
sector-wise two-Pauli restriction is removed, but they do not establish
sufficiency. Conversely, given a particular nonmaximally entangled
resource, one may ask whether an encoder exists that makes it useful
for a prescribed recovery structure. The single-output boundary
already shows that the answer can change with the allowed operations:
unrestricted recovery is related to $\operatorname{AME}(4,2^k)$ Choi
structures, whereas the prescribed two-Pauli encoder remains
obstructed~\cite{HiguchiSudbery2000,HelwigEtAl2012,lim2026ame}.
Higher-dimensional variants~\cite{ceara2026qudit} and more general
access structures~\cite{gianini2026access} provide natural settings
in which to test which features of the present results depend on the
qubit Pauli algebra.

Another important step is to move beyond pure states and exact
recovery. Approximate quantum error correction relates recovery
performance to information retained by complementary systems
~\cite{BenyOreshkov2010}. In the present setting, one would like to
know how errors in the encoder-compatible correlations affect recovery
fidelity and leakage, and whether the graph-state cut-kernel test has a
quantitative, noise-tolerant analogue. Such results would be needed
before comparing nonmaximally and maximally entangled resource classes
on imperfect hardware~\cite{yamaguchi2026experimental}.


\section*{Code Availability}

The implementation of \textsc{GSECC} is available in the companion
GitHub repository~\cite{Roy_Gupta2026Github}.

\section*{Acknowledgments}

The authors thank Koji Yamaguchi and Achim Kempf for insightful
discussions on encrypted quantum cloning. P.~R. acknowledges financial
support from SNBNCBS, Kolkata. S.~G. acknowledges support from the
Anusandhan National Research Foundation (ANRF) under the PM ECRG Grant
No.~ANRF/ECRG/2025/004000/PMS.

\section*{Author contributions}
S.G. conceived the idea and developed the methodology. P.R. contributed to extending and generalizing the methodology. S.G. and P.R. contributed equally to the drafting of the manuscript. All authors read and approved the final manuscript.


\begin{thebibliography}{10}
	
	\bibitem{wootters1982single}
	William~K. Wootters and Wojciech~H. Zurek.
	\newblock ``A single quantum cannot be cloned''.
	\newblock \href{https://dx.doi.org/10.1038/299802a0}{Nature {\bf 299},
		802--803}~(1982).
	
	\bibitem{dieks1982communication}
	Dennis Dieks.
	\newblock ``Communication by {EPR} devices''.
	\newblock \href{https://dx.doi.org/10.1016/0375-9601(82)90084-6}{Physics
		Letters A {\bf 92}, 271--272}~(1982).
	
	\bibitem{scarani2005quantum}
	Valerio Scarani, Sofyan Iblisdir, Nicolas Gisin, and Antonio Ac\'{\i}n.
	\newblock ``Quantum cloning''.
	\newblock \href{https://dx.doi.org/10.1103/RevModPhys.77.1225}{Rev. Mod. Phys.
		{\bf 77}, 1225--1256}~(2005).
	
	\bibitem{barnum1996noncommuting}
	Howard Barnum, Carlton~M. Caves, Christopher~A. Fuchs, Richard Jozsa, and
	Benjamin Schumacher.
	\newblock ``Noncommuting mixed states cannot be broadcast''.
	\newblock \href{https://dx.doi.org/10.1103/PhysRevLett.76.2818}{Phys. Rev.
		Lett. {\bf 76}, 2818--2821}~(1996).
	
	\bibitem{yamaguchi2026encrypted}
	Koji Yamaguchi and Achim Kempf.
	\newblock ``Encrypted qubits can be cloned''.
	\newblock \href{https://dx.doi.org/10.1103/y4y1-1ll6}{Phys. Rev. Lett. {\bf
			136}, 010801}~(2026).
	
	\bibitem{yamaguchi2026experimental}
	Koji Yamaguchi, Leon Rullk{\"o}tter, Ibrahim Shehzad, Sean~J. Wagner, Christian
	Tutschku, and Achim Kempf.
	\newblock ``Experimental demonstration that qubits can be cloned at will, if
	encrypted with a single-use decryption key''~(2026).
	\newblock  \href{http://arxiv.org/abs/2602.10695}{arXiv:2602.10695}.
	
	\bibitem{kimble2008quantum}
	H.~Jeff Kimble.
	\newblock ``The quantum internet''.
	\newblock \href{https://dx.doi.org/10.1038/nature07127}{Nature {\bf 453},
		1023--1030}~(2008).
	
	\bibitem{wehner2018quantum}
	Stephanie Wehner, David Elkouss, and Ronald Hanson.
	\newblock ``Quantum internet: A vision for the road ahead''.
	\newblock \href{https://dx.doi.org/10.1126/science.aam9288}{Science {\bf 362},
		eaam9288}~(2018).
	
	\bibitem{ceara2026qudit}
	Filip-Ioan Cear{\u{a}}.
	\newblock ``Cloning encrypted quantum states in arbitrary dimensions''.
	\newblock \href{https://dx.doi.org/10.1103/b32z-r5p7}{Phys. Rev. A {\bf 114},
		022450}~(2026).
	
	\bibitem{lim2026ame}
	Zheng~Liang Lim and Hoi-Kwong Lo.
	\newblock ``Encrypted cloning, absolute maximal entanglement and quantum secret
	sharing''~(2026).
	\newblock  \href{http://arxiv.org/abs/2605.26866}{arXiv:2605.26866}.
	
	\bibitem{gianini2026leak}
	Gabriele Gianini, Omar Hasan, Corrado Mio, Stelvio Cimato, and Ernesto Damiani.
	\newblock ``Encrypted clones can leak: Classification of informative subsets in
	quantum encrypted cloning''~(2026).
	\newblock  \href{http://arxiv.org/abs/2604.10155}{arXiv:2604.10155}.
	
	\bibitem{gianini2026full}
	Gabriele Gianini, Stelvio Cimato, Jianyi Lin, Omar Hasan, and Ernesto Damiani.
	\newblock ``Full characterization of informative subsets in quantum encrypted
	cloning''~(2026).
	\newblock  \href{http://arxiv.org/abs/2605.27421}{arXiv:2605.27421}.
	
	\bibitem{bai2026}
	Chen-Ming Bai, Xin-Liang Zhou, and Yu~Luo.
	\newblock ``Classification of informative subsets in quantum encrypted cloning
	on qudits''~(2026).
	\newblock  \href{http://arxiv.org/abs/2605.11642}{arXiv:2605.11642}.
	
	\bibitem{gianini2026access}
	Gabriele Gianini, Stelvio Cimato, Jianyi Lin, Omar Hasan, Corrado Mio, and
	Ernesto Damiani.
	\newblock ``Beyond the canonical protocol: Quantum encrypted cloning from
	secret-sharing access structures''~(2026).
	\newblock  \href{http://arxiv.org/abs/2606.06552}{arXiv:2606.06552}.
	
	\bibitem{gianini2026diagnosticresource}
	Gabriele Gianini, Omar Hasan, Stelvio Cimato, and Ernesto Damiani.
	\newblock ``Encrypted redundancy as a diagnostic resource: Relational diagnosis
	in quantum encrypted cloning''~(2026).
	\newblock  \href{http://arxiv.org/abs/2609.25043}{arXiv:2609.25043}.
	
	\bibitem{cuomo2020towards}
	Daniele Cuomo, Marcello Caleffi, and Angela~Sara Cacciapuoti.
	\newblock ``Towards a distributed quantum computing ecosystem''.
	\newblock \href{https://dx.doi.org/10.1049/iet-qtc.2020.0002}{IET Quantum
		Communication {\bf 1}, 3--8}~(2020).
	
	\bibitem{hillery1999quantum}
	Mark Hillery, Vladim{\'i}r Bu{\v{z}}ek, and Andr{\'e} Berthiaume.
	\newblock ``Quantum secret sharing''.
	\newblock \href{https://dx.doi.org/10.1103/PhysRevA.59.1829}{Phys. Rev. A {\bf
			59}, 1829--1834}~(1999).
	
	\bibitem{CleveGottesmanLo1999}
	Richard Cleve, Daniel Gottesman, and Hoi-Kwong Lo.
	\newblock ``How to share a quantum secret''.
	\newblock \href{https://dx.doi.org/10.1103/PhysRevLett.83.648}{Physical Review
		Letters {\bf 83}, 648--651}~(1999).
	
	\bibitem{Gottesman2000}
	Daniel Gottesman.
	\newblock ``Theory of quantum secret sharing''.
	\newblock \href{https://dx.doi.org/10.1103/PhysRevA.61.042311}{Physical Review
		A {\bf 61}, 042311}~(2000).
	
	\bibitem{ImaiEtAl2005}
	Hideki Imai, J{\"o}rn M{\"u}ller-Quade, Anderson C.~A. Nascimento, Pim Tuyls,
	and Andreas Winter.
	\newblock ``An information theoretical model for quantum secret sharing
	schemes''.
	\newblock \href{https://dx.doi.org/10.26421/QIC5.1-7}{Quantum Information and
		Computation {\bf 5}, 69--80}~(2005).
	
	\bibitem{Gottesman1997}
	Daniel Gottesman.
	\newblock ``Stabilizer codes and quantum error correction''.
	\newblock \href{https://dx.doi.org/10.7907/rzr7-dt72}{PhD thesis}.
	\newblock California Institute of Technology.
	\newblock ~(1997).
	
	\bibitem{SchumacherNielsen1996}
	Benjamin Schumacher and Michael~A. Nielsen.
	\newblock ``Quantum data processing and error correction''.
	\newblock \href{https://dx.doi.org/10.1103/PhysRevA.54.2629}{Physical Review A
		{\bf 54}, 2629--2635}~(1996).
	
	\bibitem{MarkhamSanders2008}
	Damian Markham and Barry~C. Sanders.
	\newblock ``Graph states for quantum secret sharing''.
	\newblock \href{https://dx.doi.org/10.1103/PhysRevA.78.042309}{Phys. Rev. A
		{\bf 78}, 042309}~(2008).
	
	\bibitem{Sarvepalli2012}
	Pradeep Sarvepalli.
	\newblock ``Nonthreshold quantum secret-sharing schemes in the graph-state
	formalism''.
	\newblock \href{https://dx.doi.org/10.1103/PhysRevA.86.042303}{Physical Review
		A {\bf 86}, 042303}~(2012).
	
	\bibitem{MarinMarkhamPerdrix2013}
	Anne Marin, Damian Markham, and Simon Perdrix.
	\newblock ``Access structure in graphs in high dimension and application to
	secret sharing''.
	\newblock In Simone Severini and Fernando Brandao, editors, 8th Conference on
	the Theory of Quantum Computation, Communication and Cryptography (TQC 2013).
	\newblock \href{https://dx.doi.org/10.4230/LIPIcs.TQC.2013.308}{Volume~22 of
		Leibniz International Proceedings in Informatics (LIPIcs), pages 308--324}.
	\newblock Schloss Dagstuhl -- Leibniz-Zentrum f{\"u}r Informatik~(2013).
	
	\bibitem{horodecki2009quantum}
	Ryszard Horodecki, Pawe\l{} Horodecki, Micha\l{} Horodecki, and Karol
	Horodecki.
	\newblock ``Quantum entanglement''.
	\newblock \href{https://dx.doi.org/10.1103/RevModPhys.81.865}{Rev. Mod. Phys.
		{\bf 81}, 865--942}~(2009).
	
	\bibitem{guhne2009entanglement}
	Otfried G{\"u}hne and G{\'e}za T{\'o}th.
	\newblock ``Entanglement detection''.
	\newblock \href{https://dx.doi.org/10.1016/j.physrep.2009.02.004}{Physics
		Reports {\bf 474}, 1--75}~(2009).
	
	\bibitem{hein2006entanglement}
	M.~Hein, W.~D{\"u}r, J.~Eisert, R.~Raussendorf, M.~Van den Nest, and H.-J.
	Briegel.
	\newblock ``Entanglement in graph states and its applications''.
	\newblock In Quantum Computers, Algorithms and Chaos.
	\newblock \href{https://dx.doi.org/10.3254/978-1-61499-018-5-115}{Volume 162 of
		Proceedings of the International School of Physics "Enrico Fermi", pages
		115--218}.
	\newblock IOS Press~(2006).
	
	\bibitem{hein2004}
	M.~Hein, J.~Eisert, and H.~J. Briegel.
	\newblock ``Multiparty entanglement in graph states''.
	\newblock \href{https://dx.doi.org/10.1103/PhysRevA.69.062311}{Phys. Rev. A
		{\bf 69}, 062311}~(2004).
	
	\bibitem{vandennest2004graphical}
	Maarten~Van den Nest, Jeroen Dehaene, and Bart~De Moor.
	\newblock ``Graphical description of the action of local {Clifford}
	transformations on graph states''.
	\newblock \href{https://dx.doi.org/10.1103/PhysRevA.69.022316}{Phys. Rev. A
		{\bf 69}, 022316}~(2004).
	
	\bibitem{FattalEtAl2004}
	David Fattal, Toby~S. Cubitt, Yoshihisa Yamamoto, Sergey Bravyi, and Isaac~L.
	Chuang.
	\newblock ``Entanglement in the stabilizer formalism''~(2004).
	\newblock
	\href{http://arxiv.org/abs/quant-ph/0406168}{arXiv:quant-ph/0406168}.
	
	\bibitem{Choi1975}
	Man-Duen Choi.
	\newblock ``Completely positive linear maps on complex matrices''.
	\newblock \href{https://dx.doi.org/10.1016/0024-3795(75)90075-0}{Linear Algebra
		and its Applications {\bf 10}, 285--290}~(1975).
	
	\bibitem{Lloyd1997}
	Seth Lloyd.
	\newblock ``Capacity of the noisy quantum channel''.
	\newblock \href{https://dx.doi.org/10.1103/PhysRevA.55.1613}{Physical Review A
		{\bf 55}, 1613--1622}~(1997).
	
	\bibitem{Devetak2005}
	Igor Devetak.
	\newblock ``The private classical capacity and quantum capacity of a quantum
	channel''.
	\newblock \href{https://dx.doi.org/10.1109/TIT.2004.839515}{IEEE Transactions
		on Information Theory {\bf 51}, 44--55}~(2005).
	
	\bibitem{DevetakShor2005}
	Igor Devetak and Peter~W. Shor.
	\newblock ``The capacity of a quantum channel for simultaneous transmission of
	classical and quantum information''.
	\newblock \href{https://dx.doi.org/10.1007/s00220-005-1317-6}{Communications in
		Mathematical Physics {\bf 256}, 287--303}~(2005).
	
	\bibitem{HiguchiSudbery2000}
	Atsushi Higuchi and Anthony Sudbery.
	\newblock ``How entangled can two couples get?''.
	\newblock \href{https://dx.doi.org/10.1016/S0375-9601(00)00480-1}{Physics
		Letters A {\bf 273}, 213--217}~(2000).
	
	\bibitem{HelwigEtAl2012}
	Wolfram Helwig, Wei Cui, Jos{\'e}~Ignacio Latorre, Arnau Riera, and Hoi-Kwong
	Lo.
	\newblock ``Absolute maximal entanglement and quantum secret sharing''.
	\newblock \href{https://dx.doi.org/10.1103/PhysRevA.86.052335}{Physical Review
		A {\bf 86}, 052335}~(2012).
	
	\bibitem{oum2005rankwidth}
	Sang il~Oum.
	\newblock ``Rank-width and vertex-minors''.
	\newblock \href{https://dx.doi.org/10.1016/j.jctb.2005.03.003}{Journal of
		Combinatorial Theory, Series B {\bf 95}, 79--100}~(2005).
	
	\bibitem{DehaeneDeMoor2003}
	Jeroen Dehaene and Bart~De Moor.
	\newblock ``Clifford group, stabilizer states, and linear and quadratic
	operations over {GF}(2)''.
	\newblock \href{https://dx.doi.org/10.1103/PhysRevA.68.042318}{Physical Review
		A {\bf 68}, 042318}~(2003).
	
	\bibitem{BravyiTerhal2009}
	Sergey Bravyi and Barbara~M. Terhal.
	\newblock ``A no-go theorem for a two-dimensional self-correcting quantum
	memory based on stabilizer codes''.
	\newblock \href{https://dx.doi.org/10.1088/1367-2630/11/4/043029}{New Journal
		of Physics {\bf 11}, 043029}~(2009).
	
	\bibitem{aaronson2004improved}
	Scott Aaronson and Daniel Gottesman.
	\newblock ``Improved simulation of stabilizer circuits''.
	\newblock \href{https://dx.doi.org/10.1103/PhysRevA.70.052328}{Phys. Rev. A
		{\bf 70}, 052328}~(2004).
	
	\bibitem{PatelMarkovHayes2008}
	Ketan~N. Patel, Igor~L. Markov, and John~P. Hayes.
	\newblock ``Optimal synthesis of linear reversible circuits''.
	\newblock \href{https://dx.doi.org/10.26421/QIC8.3-4-4}{Quantum Information and
		Computation {\bf 8}, 282--294}~(2008).
	
	\bibitem{briegel2001persistent}
	Hans~J. Briegel and Robert Raussendorf.
	\newblock ``Persistent entanglement in arrays of interacting particles''.
	\newblock \href{https://dx.doi.org/10.1103/PhysRevLett.86.910}{Phys. Rev. Lett.
		{\bf 86}, 910--913}~(2001).
	
	\bibitem{raussendorf2001one}
	Robert Raussendorf and Hans~J. Briegel.
	\newblock ``A one-way quantum computer''.
	\newblock \href{https://dx.doi.org/10.1103/PhysRevLett.86.5188}{Phys. Rev.
		Lett. {\bf 86}, 5188--5191}~(2001).
	
	\bibitem{raussendorf2003measurement}
	Robert Raussendorf, Daniel~E. Browne, and Hans~J. Briegel.
	\newblock ``Measurement-based quantum computation on cluster states''.
	\newblock \href{https://dx.doi.org/10.1103/PhysRevA.68.022312}{Phys. Rev. A
		{\bf 68}, 022312}~(2003).
	
	\bibitem{erdos1959}
	Paul Erd{\H{o}}s and Alfr{\'e}d R{\'e}nyi.
	\newblock ``On random graphs. i.''.
	\newblock \href{https://dx.doi.org/10.5486/PMD.1959.6.3-4.12}{Publicationes
		Mathematicae Debrecen {\bf 6}, 290--297}~(1959).
	
	\bibitem{gilbert1959}
	Edgar~N. Gilbert.
	\newblock ``Random graphs''.
	\newblock \href{https://dx.doi.org/10.1214/aoms/1177706098}{The Annals of
		Mathematical Statistics {\bf 30}, 1141--1144}~(1959).
	
	\bibitem{Roy_Gupta2026Github}
	Pritam Roy and Shashank Gupta.
	\newblock ``{Graph-State Encrypted Cloning Algorithms}''.
	\newblock
	\href{https://github.com/pritam-roy-99/Graph-State-Encrypted-Cloning-Algorithms}{GitHub
		repository}~(2026).
	
	\bibitem{greenberger1989going}
	Daniel~M. Greenberger, Michael~A. Horne, and Anton Zeilinger.
	\newblock ``Going beyond {Bell}'s theorem''.
	\newblock In Menas Kafatos, editor, Bell's Theorem, Quantum Theory and
	Conceptions of the Universe.
	\newblock \href{https://dx.doi.org/10.1007/978-94-017-0849-4_10}{Volume~37 of
		Fundamental Theories of Physics, pages 69--72}.
	\newblock Springer, Dordrecht~(1989).
	
	\bibitem{dicke1954}
	R.~H. Dicke.
	\newblock ``Coherence in spontaneous radiation processes''.
	\newblock \href{https://dx.doi.org/10.1103/PhysRev.93.99}{Phys. Rev. {\bf 93},
		99--110}~(1954).
	
	\bibitem{MorenoParisio2018}
	M.~G.~M. Moreno and Fernando Parisio.
	\newblock ``All bipartitions of arbitrary {Dicke} states''~(2018).
	\newblock  \href{http://arxiv.org/abs/1801.00762}{arXiv:1801.00762}.
	
	\bibitem{MunizziSchnitzer2024}
	William Munizzi and Howard~J. Schnitzer.
	\newblock ``Entropy cones and entanglement evolution for {Dicke} states''.
	\newblock \href{https://dx.doi.org/10.1103/PhysRevA.109.012405}{Phys. Rev. A
		{\bf 109}, 012405}~(2024).
	
	\bibitem{KretschmannKribsSpekkens2008}
	Dennis Kretschmann, David~W. Kribs, and Robert~W. Spekkens.
	\newblock ``Complementarity of private and correctable subsystems in quantum
	cryptography and error correction''.
	\newblock \href{https://dx.doi.org/10.1103/PhysRevA.78.032330}{Physical Review
		A {\bf 78}, 032330}~(2008).
	
	\bibitem{dur2000three}
	W.~D\"ur, G.~Vidal, and J.~I. Cirac.
	\newblock ``Three qubits can be entangled in two inequivalent ways''.
	\newblock \href{https://dx.doi.org/10.1103/PhysRevA.62.062314}{Phys. Rev. A
		{\bf 62}, 062314}~(2000).
	
	\bibitem{ChitambarGour2019}
	Eric Chitambar and Gilad Gour.
	\newblock ``Quantum resource theories''.
	\newblock \href{https://dx.doi.org/10.1103/RevModPhys.91.025001}{Reviews of
		Modern Physics {\bf 91}, 025001}~(2019).
	
	\bibitem{BenyOreshkov2010}
	C{\'e}dric B{\'e}ny and Ognyan Oreshkov.
	\newblock ``General conditions for approximate quantum error correction and
	near-optimal recovery channels''.
	\newblock \href{https://dx.doi.org/10.1103/PhysRevLett.104.120501}{Physical
		Review Letters {\bf 104}, 120501}~(2010).
	
	\bibitem{ArakiLieb1970}
	Huzihiro Araki and Elliott~H. Lieb.
	\newblock ``Entropy inequalities''.
	\newblock \href{https://dx.doi.org/10.1007/BF01646092}{Communications in
		Mathematical Physics {\bf 18}, 160--170}~(1970).
	
	\bibitem{LiebRuskai1973}
	Elliott~H. Lieb and Mary~Beth Ruskai.
	\newblock ``Proof of the strong subadditivity of quantum-mechanical entropy''.
	\newblock \href{https://dx.doi.org/10.1063/1.1666274}{Journal of Mathematical
		Physics {\bf 14}, 1938--1941}~(1973).
	
\end{thebibliography}

\onecolumn
\appendix
\renewcommand{\thesubsection}{\thesection.\arabic{subsection}}
\numberwithin{equation}{section}
\numberwithin{figure}{section}
\numberwithin{table}{section}
\numberwithin{lemma}{section}
\numberwithin{proposition}{section}
\numberwithin{corollary}{section}


\section{General framework and resource conditions}
\label{sec:supp_general_framework}

This section provides the derivations underlying the general framework
of the main text. We first establish the exact recovery--decoupling
criterion and the Pauli action of the sector-wise two-Pauli encoder.
We then prove the exact pure-resource characterization, its maximally
entangled consequences, the single-output boundary results, and the
architecture-independent entropy and spectral constraints.

Throughout,
$\ket{\Omega}_{\widetilde AAS\mathcal N}$ denotes the encoded Choi
state introduced in the main text. Since the encoder leaves the
complete key register $\mathcal N$ untouched,
\begin{equation}
\Omega_{\widetilde A\mathcal N}
=
\frac{I_{\widetilde A}}{2^k}
\otimes
\rho_{\mathcal N},
\qquad
\Omega_{\widetilde A}
=
\frac{I_{\widetilde A}}{2^k}.
\end{equation}

\subsection{Technical preliminaries}
\label{subsec:supp_general_preliminaries}

We begin by establishing the exact information–disturbance relation for the authorized subsystem$S_i\mathcal N$.

\begin{lemma}[Recovery--decoupling equivalence]
\label{lem:supp_recovery_decoupling}
For the pure encoded Choi state, exact recovery of the complete
$k$-qubit input from $S_i\mathcal N$ is equivalent to
\begin{equation}
I(\widetilde A:AS_{\bar i})_\Omega=0,
\label{eq:supp_decoupling_mi}
\end{equation}
or, equivalently,
\begin{equation}
\Omega_{\widetilde AAS_{\bar i}}
=
\frac{I_{\widetilde A}}{2^k}
\otimes
\Omega_{AS_{\bar i}}.
\label{eq:supp_decoupling_product}
\end{equation}
\end{lemma}

\begin{proof}
Suppose that there exists a CPTP map
$\mathcal D_i:S_i\mathcal N\rightarrow A_i'$, with $A_i'\simeq A$,
such that
\[
\mathcal D_i\circ\mathcal Q_i=\id_A.
\]
Applying $\mathcal D_i$ to the authorized subsystem of the encoded
Choi state therefore produces a maximally entangled state between
$\widetilde A$ and $A_i'$. By the data-processing inequality for
quantum mutual information
~\cite{SchumacherNielsen1996},
\[
2k
\le
I(\widetilde A:S_i\mathcal N)_\Omega
\le
2S(\widetilde A)
=
2k.
\]
Hence,
\[
I(\widetilde A:S_i\mathcal N)_\Omega=2k.
\]

Because the encoded Choi state is pure under the tripartition
$\widetilde A:(S_i\mathcal N):(AS_{\bar i})$, the mutual informations
satisfy
\[
I(\widetilde A:S_i\mathcal N)_\Omega
+
I(\widetilde A:AS_{\bar i})_\Omega
=
2S(\widetilde A).
\]
Using $S(\widetilde A)=k$ together with
$I(\widetilde A:S_i\mathcal N)_\Omega=2k$ gives
Eq.~\eqref{eq:supp_decoupling_mi}. Vanishing mutual information is
equivalent to a product state, and therefore
Eq.~\eqref{eq:supp_decoupling_product} follows.

Conversely, suppose that Eq.~\eqref{eq:supp_decoupling_product} holds.
Then $\widetilde A$ is completely decoupled from the complement of
$S_i\mathcal N$. Since
$\Omega_{\widetilde A}
=
\frac{I_{\widetilde A}}{2^k},$ the subsystem $\widetilde A$ is maximally mixed. By uniqueness of
purification, there therefore exists an isometry acting only on
$S_i\mathcal N$ that extracts a $2^k$-dimensional maximally entangled
partner of $\widetilde A$. Tracing out the unused output of this
isometry gives a CPTP decoder $\mathcal D_i$ satisfying
\[
\mathcal D_i\circ\mathcal Q_i=\id_A.
\]
This proves the equivalence.
\end{proof}

Purity also gives
\[
I(\widetilde A\rangle S_i\mathcal N)_\Omega
=
k-I(\widetilde A:AS_{\bar i})_\Omega,
\]
so exact recovery is equivalent to coherent-information saturation at
$k$.

We next determine the Pauli support relevant to the fixed sector-wise
encoder. We use the standard stabilizer and binary-symplectic
representation of Clifford operations
~\cite{Gottesman1997,DehaeneDeMoor2003}.

For one input sector, write
\[
U_{P,Q}^{(m)}
=
e^{-i\pi P_AP_{\mathsf S}^{\otimes m}/4}
e^{-i\pi Q_AQ_{\mathsf S}^{\otimes m}/4},
\qquad
P\neq Q,
\]
and define, up to the conventional Pauli sign,
\[
R=iPQ.
\]

\begin{lemma}[Pauli action of the two-Pauli sector-wise encoder]
\label{lem:supp_encoder_pauli_action}
For a single input sector:
\begin{enumerate}
\item[(i)]
if $m$ is even, no nonidentity signal-only Pauli is mapped to an
operator with trivial signal support;

\item[(ii)]
if $m$ is odd, the unique nonidentity signal direction that can lose
all signal support is $R^{\otimes m}$;

\item[(iii)]
at any fixed signal position, the three nonidentity input Pauli
directions are mapped bijectively onto the three nonidentity
single-qubit Pauli directions.
\end{enumerate}
\end{lemma}

\begin{proof}
It is sufficient to work first in the canonical
$(P,Q)=(X,Z)$ frame. Represent a signal-only Pauli by
\[
(0,a\,|\,0,b),
\qquad
a,b\in\mathbb F_2^m,
\]
where the first coordinate in each half refers to the input qubit.

For a Hermitian Pauli $M$ satisfying $M^2=I$, the Clifford rotation
$V_M=e^{-i\pi M/4}$ acts, up to phase, as
\[
V_MPV_M^\dagger
\doteq
\begin{cases}
P, & [P,M]=0,\\
MP, & \{P,M\}=0.
\end{cases}
\]
Define
\[
p=\mathbf 1_m^{\mathsf T}a,
\qquad
t=\mathbf 1_m^{\mathsf T}b+(m+1)p,
\]
with all arithmetic over $\mathbb F_2$. Successive conjugation by the
two Pauli rotations then gives
\begin{equation}
(0,a\,|\,0,b)
\longmapsto
(0,a\,|\,p,b+p\mathbf1_m)
\longmapsto
(t,a+t\mathbf1_m\,|\,p,b+p\mathbf1_m).
\label{eq:supp_signal_pauli_action}
\end{equation}

From Eq.~\eqref{eq:supp_signal_pauli_action}, trivial signal support
requires
\[
a=t\mathbf1_m,
\qquad
b=p\mathbf1_m.
\]
Taking the parities of these two relations gives
\begin{equation}
p=mt,
\qquad
t=p
\pmod 2.
\label{eq:supp_signal_parity_conditions}
\end{equation}
For even $m$, Eq.~\eqref{eq:supp_signal_parity_conditions} forces
$p=t=0$, and hence $a=b=0$. Thus no nonidentity signal-only Pauli can
lose all signal support. For odd $m$, there is one additional
solution,
$p=t=1,
\qquad
a=b=\mathbf1_m,$ which corresponds to the Pauli direction $Y^{\otimes m}$ in the
canonical frame.

We next consider the action on input Paulis. Represent an input Pauli
by $(u,0\,|\,v,0)$. Applying the same two Clifford transformations
gives
\begin{equation}
(u,0\,|\,v,0)
\longmapsto
(u,0\,|\,v+u,u\mathbf1_m)
\longmapsto
(u+t,t\mathbf1_m\,|\,v+u,u\mathbf1_m),
\qquad
t=v+(m+1)u.
\label{eq:supp_input_pauli_action}
\end{equation}
Substituting the three nonzero binary labels
$(u,v)=(1,0),(1,1),(0,1)$ into
Eq.~\eqref{eq:supp_input_pauli_action} shows that the three
nonidentity input Paulis restrict, at every signal position, to three
distinct nonidentity single-qubit Pauli directions. Since there are
exactly three such directions, the correspondence is bijective.

Finally, every ordered pair of distinct Pauli axes $(P,Q)$ is related
to the canonical pair $(X,Z)$ by a single-qubit Clifford conjugation.
Such a conjugation merely permutes $X,Y,Z$, while preserving whether a
given tensor factor is identity or nonidentity. It therefore preserves
the support statements established above and maps the canonical
exceptional direction $Y$ to
\[
R=iPQ.
\]
Hence, for odd $m$, the unique nonidentity signal direction that can
lose all signal support is $R^{\otimes m}$, completing the proof.
\end{proof}

Because the complete encoder factorizes across input sectors,
Lemma~\ref{lem:supp_encoder_pauli_action} implies that a signal-only
Pauli can lose all signal support only for odd $m$, with every active
sector proportional to $R^{\otimes m}$. Thus the complete exceptional
set is
\begin{align*}
T(c)
=
\prod_{j=1}^{k}
\left(
R_{\mathsf S_j}^{\otimes m}
\right)^{c_j},
\qquad
c\in\mathbb F_2^k. 
\end{align*}

\subsection{Proof of Theorem~1
of the main text}
\label{subsec:supp_proof_general_resource}

We now prove the exact fixed-encoder pure-resource criterion stated as
Theorem~1 in the main text. The
proof combines the recovery--decoupling criterion of
Lemma~\ref{lem:supp_recovery_decoupling} with the Pauli-support
structure established in Lemma~\ref{lem:supp_encoder_pauli_action}.

\begin{proof}
Let $\mathcal P_n$ denote the phase-free $n$-qubit Pauli basis. Expand
the signal marginal as
\[
\rho_S
=
\frac{1}{2^{mk}}
\sum_{\Sigma\in\mathcal P_{mk}}
r_\Sigma\,\Sigma,
\qquad
r_I=1.
\]
Using the Pauli expansion of the maximally entangled state,
\[
\Phi^+_{\widetilde A A}
=
\frac{1}{2^{2k}}
\sum_{\Lambda\in\mathcal P_k}
\Lambda_{\widetilde A}^{\mathsf T}\otimes\Lambda_A,
\]
and tracing out $\mathcal N$, the encoded Choi state becomes
\begin{equation}
\Omega_{\widetilde AAS}
=
\frac{1}{2^{2k+mk}}
\sum_{\Lambda\in\mathcal P_k}
\sum_{\Sigma\in\mathcal P_{mk}}
r_\Sigma\,
\Lambda_{\widetilde A}^{\mathsf T}
\otimes
U_{P,Q}^{(m,k)}
(\Lambda_A\otimes\Sigma_S)
U_{P,Q}^{(m,k)\dagger}.
\label{eq:supp_encoded_choi_pauli}
\end{equation}

We first prove \emph{necessity}. Fix an output $i$. By
Lemma~\ref{lem:supp_recovery_decoupling}, exact recovery from
$S_i\mathcal N$ is equivalent to
$\Omega_{\widetilde AAS_{\bar i}}
=
\frac{I_{\widetilde A}}{2^k}
\otimes
\Omega_{AS_{\bar i}}.$

Hence the complementary state can contain no net contribution
proportional to
$\Lambda_{\widetilde A}^{\mathsf T}$ with $\Lambda\neq I$.

We repeatedly use the Pauli partial-trace property
$\Tr_{S_i}M=0$ whenever $M$ has nonidentity support on at least one
qubit of $S_i$. Conversely, if
$M=M_{AS_{\bar i}}\otimes I_{S_i},$ then
\begin{align*}
   \Tr_{S_i}M
=
2^kM_{AS_{\bar i}}
\neq0. 
\end{align*}

Suppose that $r_\Sigma\neq0$ for some signal Pauli $\Sigma$ outside
the exceptional family identified after
Lemma~\ref{lem:supp_encoder_pauli_action}. Then
\[
E_\Sigma
:=
U_{P,Q}^{(m,k)}
(I_A\otimes\Sigma_S)
U_{P,Q}^{(m,k)\dagger}
\]
has nonidentity signal support. Choose an output $i$ on which this
support is nontrivial. Writing
$S_i=S_{1,i}\cdots S_{k,i}$, let
$M_{j,i}\in\{I,X,Y,Z\}$ denote the restriction of $E_\Sigma$ to
$S_{j,i}$.

For each sector $j$, choose an input Pauli $\lambda_j$ as follows.
If $M_{j,i}=I$, take $\lambda_j=I$. If $M_{j,i}\neq I$, then
Lemma~\ref{lem:supp_encoder_pauli_action}(iii) provides a nonidentity
input Pauli $\lambda_j$ whose encoded restriction on $S_{j,i}$ is
the same Pauli direction $M_{j,i}$. Define
$\Lambda=\bigotimes_{j=1}^k\lambda_j.$ Since $E_\Sigma$ is nontrivial on $S_i$, at least one
$M_{j,i}\neq I$, and hence $\Lambda\neq I$.

Using multiplicativity of unitary conjugation,
\begin{align}
&
U_{P,Q}^{(m,k)}
(\Lambda_A\otimes\Sigma_S)
U_{P,Q}^{(m,k)\dagger}
\nonumber\\
&\qquad=
U_{P,Q}^{(m,k)}
(\Lambda_A\otimes I_S)
U_{P,Q}^{(m,k)\dagger}
\,
U_{P,Q}^{(m,k)}
(I_A\otimes\Sigma_S)
U_{P,Q}^{(m,k)\dagger}.
\label{eq:supp_factor_encoded_paulis}
\end{align}
On every $S_{j,i}$ with $M_{j,i}\neq I$, the two factors have the
same Pauli restriction and therefore multiply to the identity, up to
phase. If $M_{j,i}=I$, we chose $\lambda_j=I$. Consequently,
\begin{equation}
U_{P,Q}^{(m,k)}
(\Lambda_A\otimes\Sigma_S)
U_{P,Q}^{(m,k)\dagger}
=
\pm M_{AS_{\bar i}}\otimes I_{S_i}
\label{eq:supp_identity_on_output}
\end{equation}
for some Pauli operator $M_{AS_{\bar i}}$. Therefore
\[
\Tr_{S_i}
\!\left[
U_{P,Q}^{(m,k)}
(\Lambda_A\otimes\Sigma_S)
U_{P,Q}^{(m,k)\dagger}
\right]
=
\pm2^kM_{AS_{\bar i}}
\neq0.
\]
This produces a nonzero term proportional to
$\Lambda_{\widetilde A}^{\mathsf T}
\otimes
M_{AS_{\bar i}},
\Lambda\neq I,$
in the complementary Choi state.

Distinct surviving terms cannot cancel. Clifford conjugation is
bijective on the phase-free Pauli basis, so distinct pairs
$(\Lambda,\Sigma)$ produce distinct encoded Pauli operators. If two
such operators are both identity on $S_i$, their restrictions to
$AS_{\bar i}$ must also remain distinct. The surviving terms are
therefore distinct Pauli-basis elements and hence Hilbert--Schmidt
orthogonal. This contradicts complementary decoupling. Thus every
nonzero coefficient $r_\Sigma$ must correspond to an exceptional
signal Pauli.

For even $m$, the exceptional family contains only the identity, so
\[
\rho_S
=
\frac{I_S}{2^{mk}}.
\]
For odd $m$, the exceptional operators are exactly
\[
T(c)
=
\prod_{j=1}^{k}
\left(
R_{\mathsf S_j}^{\otimes m}
\right)^{c_j},
\qquad
c\in\mathbb F_2^k,
\]
and hence
\begin{equation}
\rho_S
=
\frac{1}{2^{mk}}
\sum_{c\in\mathbb F_2^k}
\alpha_c\,T(c),
\qquad
\alpha_0=1,
\qquad
\alpha_c\in\mathbb R.
\label{eq:supp_odd_resource_form}
\end{equation}
Here $\alpha_0=1$ follows from normalization,
$\alpha_c\in\mathbb R$ from Hermiticity, and the remaining condition
is precisely $\rho_S\ge0$.

We next prove sufficiency. Suppose first that $m$ is odd and
$\rho_S$ has the form in Eq.~\eqref{eq:supp_odd_resource_form}.
For every exceptional operator,
\begin{equation}
U_{P,Q}^{(m,k)}
(I_A\otimes T(c))
U_{P,Q}^{(m,k)\dagger}
\doteq
R_A(c)\otimes I_S,
\qquad
R_A(c):=
\prod_{j=1}^k R_{A_j}^{c_j}.
\label{eq:supp_exceptional_to_input}
\end{equation}
Thus an exceptional signal component is mapped entirely to the input
system and carries no signal support.

Consider a term in Eq.~\eqref{eq:supp_encoded_choi_pauli} with
$\Lambda\neq I$, and choose a sector $j$ on which $\Lambda$ is
nonidentity. By Lemma~\ref{lem:supp_encoder_pauli_action}(iii), its
encoded image has nonidentity support on every signal position
$S_{j,1},\ldots,S_{j,m}$. Moreover,
\[
U_{P,Q}^{(m,k)}
(\Lambda_A\otimes T(c))
U_{P,Q}^{(m,k)\dagger}
\doteq
U_{P,Q}^{(m,k)}
(\Lambda_A\otimes I_S)
U_{P,Q}^{(m,k)\dagger}
\bigl(R_A(c)\otimes I_S\bigr).
\]
The second factor acts trivially on $S$ and therefore does not alter
the signal support of the first. Hence, for every output $i$,
\[
\Tr_{S_i}
\!\left[
U_{P,Q}^{(m,k)}
(\Lambda_A\otimes T(c))
U_{P,Q}^{(m,k)\dagger}
\right]
=0,
\qquad
\Lambda\neq I.
\]
After tracing any $S_i$, all terms carrying a nonidentity reference
Pauli therefore vanish, leaving
\[
\Omega_{\widetilde AAS_{\bar i}}
=
\frac{I_{\widetilde A}}{2^k}
\otimes
\Omega_{AS_{\bar i}}.
\]
Lemma~\ref{lem:supp_recovery_decoupling} then gives exact recovery
from $S_i\mathcal N$. Since $i$ was arbitrary, recovery holds for
every authorized subsystem. For even $m$, the same argument applies
with the only allowed signal component $\Sigma=I$.

Finally, we prove individual concealment. From
Eq.~\eqref{eq:supp_encoded_choi_pauli},
\[
\Omega_{\widetilde A S_i}
=
\Tr_{AS_{\bar i}}
\Omega_{\widetilde AAS}.
\]
If $\Lambda\neq I$, choose an active sector $j$. Its encoded image is
nonidentity on every $S_{j,\ell}$. Since $m\ge2$, for each fixed $i$
there exists some $\ell\neq i$, so the term has nonidentity support
on the discarded subsystem $S_{\bar i}$ and vanishes under the
partial trace.

If instead $\Lambda=I$ and $c\neq0$, then
Eq.~\eqref{eq:supp_exceptional_to_input} gives
\[
U_{P,Q}^{(m,k)}
(I_A\otimes T(c))
U_{P,Q}^{(m,k)\dagger}
\doteq
R_A(c)\otimes I_S,
\]
with $R_A(c)\neq I_A$. Since every nonidentity Pauli is traceless,
this contribution vanishes when $A$ is traced out. Thus only the term
$\Lambda=I$, $c=0$ survives, yielding
\[
\Omega_{\widetilde A S_i}
=
\frac{I_{\widetilde A}}{2^k}
\otimes
\frac{I_{S_i}}{2^k}.
\]
This is the Choi state of the completely depolarizing channel, so
\[
\mathcal C_i(\rho_A)
=
\frac{I_{S_i}}{2^k}
\]
for every input state $\rho_A$. Hence every individual signal
$S_i$ is perfectly concealed.
\end{proof}

The coherent-information identity above also shows that every
authorized channel is exactly reversible and satisfies
\[
Q^{(1)}(\mathcal Q_i)
=
C_Q(\mathcal Q_i)
=
k,
\]
consistent with the standard quantum-capacity framework
~\cite{Lloyd1997,Devetak2005,DevetakShor2005}.
\subsection{Guaranteed odd-\texorpdfstring{$m$}{m} joint-signal moments}
\label{subsec:supp_joint_ciphertext}

We now derive the operator identity used in the main text to establish
that, for odd $m$, joint access to the complete signal register always
reveals a definite family of input moments. These moments give a
guaranteed component of the joint-signal information, although they
need not characterize the full joint-signal channel.

Consider first a single sector with odd $m$, and define
\[
\Pi_P=P_AP_{\mathsf S}^{\otimes m},
\qquad
\Pi_Q=Q_AQ_{\mathsf S}^{\otimes m},
\qquad
T=R_{\mathsf S}^{\otimes m},
\]
where $R=iPQ$. Since $m$ is odd, $\Pi_P$ and $\Pi_Q$ commute, while
each anticommutes with $I_A\otimes T$. For anticommuting Hermitian
Pauli operators $M$ and $N$, we use
\[
e^{i\pi M/4}Ne^{-i\pi M/4}=iMN.
\]
Successive conjugation by the two factors of the encoder then gives
\[
U_{P,Q}^{(m)\dagger}
\bigl(I_A\otimes T\bigr)
U_{P,Q}^{(m)}
=
(-1)^{(m-1)/2}
R_A\otimes I_{\mathsf S}.
\]

For the $k$-sector encoder, define
\[
T_j=R_{\mathsf S_j}^{\otimes m},
\qquad
T(c)=\prod_{j=1}^{k}T_j^{c_j},
\qquad
R_A(c)=\prod_{j=1}^{k}R_{A_j}^{c_j},
\]
for $c\in\mathbb F_2^k$, and let
\[
\sigma_m=(-1)^{(m-1)/2}.
\]
Because the encoder factorizes across sectors, the single-sector
identity applies independently to every active sector:
\begin{equation}
U_{P,Q}^{(m,k)\dagger}
\bigl(I_A\otimes T(c)\bigr)
U_{P,Q}^{(m,k)}
=
\sigma_m^{|c|}
R_A(c)\otimes I_S,
\label{eq:supp_joint_signal_conjugation}
\end{equation}
where $|c|$ denotes the Hamming weight of $c$.

Now consider the channel from the input to the complete signal
register,
\[
\mathcal C_S(\rho_A)
=
\Tr_{A\mathcal N}
\!\left[
\mathcal V_{\mathcal R}(\rho_A)
\right].
\]
Using Eq.~\eqref{eq:supp_joint_signal_conjugation}, the expectation
value of $T(c)$ is
\begin{equation}
\Tr\!\left[
T(c)\mathcal C_S(\rho_A)
\right]
=
\sigma_m^{|c|}
\Tr\!\left[
R_A(c)\rho_A
\right].
\label{eq:supp_joint_signal_moments}
\end{equation}

Equation~\eqref{eq:supp_joint_signal_moments} shows that joint access
to the complete signal register always reveals the commuting
$R$-basis moments of the input. Equivalently, these moments determine
the diagonal statistics of the input in the product $R$ basis. This
conclusion follows solely from the encoder identity in
Eq.~\eqref{eq:supp_joint_signal_conjugation} and is independent of the
nonidentity coefficients allowed in the resource marginal $\rho_S$.

The moments in Eq.~\eqref{eq:supp_joint_signal_moments} need not
exhaust the information contained in $\mathcal C_S(\rho_A)$.
Nonidentity components of $\rho_S$ can transfer additional, generally
noncommuting, input observables to the complete signal register. The
rank-deficient graph-state example in
Sec.~\ref{subsec:supp_rank_deficient_examples} provides an explicit
case in which the resulting joint-signal statistics determine the
complete input-qubit state, even though every individual signal
remains perfectly concealed. Such informational completeness concerns
state identification from joint measurement statistics and does not,
by itself, imply the existence of an exact CPTP recovery map acting on
the complete signal register.

\subsection{Proofs of Corollaries~1 and~2 of the main text}
\label{subsec:supp_general_corollaries}

We first prove Corollary~1 of the main text.

\begin{proof}
If
\[
\rho_S=\frac{I_S}{2^{mk}},
\]
then all nonidentity signal-Pauli coefficients vanish. The exact
criterion of Theorem~1 is therefore satisfied for every $m\ge2$.
For even $m$, Theorem~1 further shows that this condition is also
necessary.
\end{proof}

We next prove Corollary~2 of the main text.

\begin{proof}
Suppose the resource is maximally entangled across the
signal--noise cut $S:\mathcal N$. By uniqueness of purification,
there exists a unitary acting only on $\mathcal N$ that brings the
resource to the canonical form
\begin{equation}
\ket{\Phi_{2^\nu}}_{S\mathcal N}
=
\bigotimes_{\ell=1}^{\nu}
\ket{\phi^+}_{S_\ell N_\ell}.
\label{eq:supp_maxent_resource_factorization}
\end{equation}

Now let $\nu=m'k'$ with $m'\ge2$. Relabel the signal degrees of
freedom according to
\begin{equation}
\ell=(j-1)m'+i,
\qquad
j=1,\ldots,k',
\qquad
i=1,\ldots,m',
\label{eq:supp_signal_relabeling}
\end{equation}
and define $\mathsf S'_j=(S_{j,1},\ldots,S_{j,m'}),
\qquad
S'_i=S_{1,i}\cdots S_{k',i}.$

Equation~\eqref{eq:supp_signal_relabeling} merely regroups the same
signal degrees of freedom into $k'$ sectors of size $m'$. Since the
resource in Eq.~\eqref{eq:supp_maxent_resource_factorization} remains
maximally entangled across the same signal--noise cut, Corollary~1
implies that it supports the $(m',k')$ architecture under the
corresponding sector-wise encoder and recovery maps.

Thus the physical resource and its signal--noise entanglement remain
unchanged; only the grouping of the signal degrees of freedom and the
associated encoder--decoder pair are modified.

For example, when $\nu=12$, the allowed factorizations include $(m',k')=(2,6),(3,4),(4,3),(6,2),(12,1).$

\end{proof}

\subsection{Proof of Proposition~1 of the main text}
\label{subsec:supp_proof_single_clone_global}

We now prove the unrestricted single-output characterization stated as
Proposition~1 in the main text. Its connection with absolutely
maximally entangled states follows the standard AME framework
~\cite{HelwigEtAl2012}.

\begin{proof}
Set $d=2^k$. Perfect concealment gives
$I(\widetilde A:S)_\Omega=0,$ whereas exact recovery from $S\mathcal N$ gives
\begin{align*}
    I(\widetilde A:S\mathcal N)_\Omega=2k.
\end{align*}

By the chain rule,
\begin{align*}
    I(\widetilde A:S\mathcal N)_\Omega
=
I(\widetilde A:S)_\Omega
+
I(\widetilde A:\mathcal N|S)_\Omega,
\end{align*}
and hence
\begin{equation}
2k
=
I(\widetilde A:\mathcal N|S)_\Omega
\le
2S(\rho_{\mathcal N})
\le
2\log_2 d
=
2k.
\label{eq:supp_single_clone_entropy_bound}
\end{equation}
Thus equality holds throughout, so
$S(\rho_{\mathcal N})=\log_2 d,
\rho_{\mathcal N}=\frac{I_d}{d}.$ Since the resource is pure and
$\dim S=\dim\mathcal N=d$, $\rho_S=\frac{I_d}{d}.$

The Choi input marginal satisfies $\rho_A=I_A/d$. Therefore, before
encoding,
\begin{align*}
    \rho_{AS}
=
\frac{I_A}{d}\otimes\frac{I_S}{d}
=
\frac{I_{AS}}{d^2},
\end{align*}

and unitarity gives
$\Omega_{AS}=\frac{I_{AS}}{d^2}.$

Perfect concealment then implies
\[
\Omega_{\widetilde A S}
=
\frac{I_{\widetilde A}}{d}\otimes\Omega_S
=
\frac{I_{\widetilde A S}}{d^2},
\]
while recovery--decoupling gives
\[
\Omega_{\widetilde A A}
=
\frac{I_{\widetilde A}}{d}\otimes\Omega_A
=
\frac{I_{\widetilde A A}}{d^2}.
\]
Since the encoded Choi state is pure, the complementary two-party
marginals are also maximally mixed. Hence every two-party marginal is
$I/d^2$, and therefore
$\ket{\Omega}_{\widetilde AAS\mathcal N}
$ is an $\operatorname{AME}(4,d)$ state.

Conversely, if the encoded Choi state is $\operatorname{AME}(4,d)$,
then
\[
\Omega_{\widetilde A S}
=
\Omega_{\widetilde A A}
=
\frac{I}{d^2}.
\]
The first condition gives perfect concealment of $S$, while the second
is precisely complementary decoupling and therefore implies exact
recovery from $S\mathcal N$.

For $k=1$, this would require an $\operatorname{AME}(4,2)$ state,
which does not exist
~\cite{HiguchiSudbery2000}.

For $k\ge2$, $d=2^k$ is a prime-power dimension. Identify the
computational basis with $\mathbb F_d$, choose
$\lambda\in\mathbb F_d\setminus\{0,1\}$, and define
\[
U_\lambda
\ket a_A\ket s_S
=
\ket{a+s}_A
\ket{a+\lambda s}_S.
\]
The corresponding linear map has determinant
$\lambda-1\neq0$, so it is bijective over $\mathbb F_d^2$. Acting on
two maximally entangled pairs gives
\begin{equation}
\ket{\Omega_\lambda}
=
\frac1d
\sum_{a,s\in\mathbb F_d}
\ket a_{\widetilde A}
\ket{a+s}_A
\ket{a+\lambda s}_S
\ket s_{\mathcal N}.
\label{eq:supp_ame_finite_field_state}
\end{equation}
Any two of the four labels
$a, a+s, a+\lambda s, s$, determine $(a,s)$ uniquely; in particular,
$a+\lambda s-(a+s)=(\lambda-1)s$, and both
$\lambda$ and $\lambda-1$ are invertible. Hence every two-party
marginal equals $I/d^2$, so
$\ket{\Omega_\lambda}$ is an $\operatorname{AME}(4,d)$ state,
consistent with the finite-field constructions of
Ref.~\cite{HelwigEtAl2012}.

An authorized decoder may be chosen as
\begin{equation}
\ket u_S\ket n_{\mathcal N}
\longmapsto
\ket{u-\lambda n}_{O}
\ket{u+(1-\lambda)n}_{K}.
\label{eq:supp_single_clone_decoder}
\end{equation}
For $u=a+\lambda s$ and $n=s$,
\[
u-\lambda n=a,
\qquad
u+(1-\lambda)n=a+s.
\]
Therefore
\[
\ket{a+\lambda s}_S\ket s_{\mathcal N}
\longmapsto
\ket a_O\ket{a+s}_K.
\]
Applying this to Eq.~\eqref{eq:supp_ame_finite_field_state} and setting
$t=a+s$ gives
\[
\ket{\Omega_\lambda}
\longmapsto
\frac1d
\sum_{a,t\in\mathbb F_d}
\ket a_{\widetilde A}\ket t_A\ket a_O\ket t_K
=
\ket{\Phi_d}_{\widetilde A O}
\otimes
\ket{\Phi_d}_{AK}.
\]
Thus the decoder exactly recovers the complete unknown $k$-qubit
input, including arbitrary correlations with an external reference.
\end{proof}

This identifies the single-output boundary of the present
architecture. The AME condition characterizes the unrestricted task,
whereas the fixed sector-wise encoder imposes a stronger structural
constraint.

\subsection{Proof of Proposition~2 of the main text}
\label{subsec:supp_proof_single_clone_sectorwise}

We now prove the sector-wise single-output obstruction stated as
Proposition~2 in the main text.

\begin{proof}[Proof of Proposition~2]
Let $R=iPQ$. For each input sector $j$,
\begin{equation}
U_{P,Q;j}^{(1)\dagger}
\left(
I_{A_j}\otimes R_{S_{j,1}}
\right)
U_{P,Q;j}^{(1)}
=
R_{A_j}\otimes I_{S_{j,1}}.
\label{eq:supp_single_clone_sectorwise_identity}
\end{equation}
Taking the expectation value of
Eq.~\eqref{eq:supp_single_clone_sectorwise_identity} in the encoded
state gives, for every input state $\rho_A$ and every pure resource,
\begin{equation*}
\Tr\left[
R_{S_{j,1}}\mathcal C_1(\rho_A)
\right]
=
\Tr\left[
R_{A_j}\rho_A
\right].
\end{equation*}
Thus the signal observable $R_{S_{j,1}}$ directly reproduces the
corresponding input moment $\langle R_{A_j}\rangle$, independently of
the resource state. Since the right-hand side can be varied by
changing $\rho_A$, the output $\mathcal C_1(\rho_A)$ cannot be
independent of the input. Hence perfect concealment is impossible for the sector-wise encoder
$U_{P,Q}^{(1,k)}$ for every $k\ge1$, independently of the choice of
pure resource state.
\end{proof}

\subsection{Proof of Proposition~3 and Corollary~3 of the main text}
\label{subsec:supp_proof_all_output_entropy}

We finally remove the sector-wise
assumptions. Let the encoder be an arbitrary unitary on $AS$ acting
trivially on $\mathcal N$, and assume exact recovery from
$S_i\mathcal N$ for every $i=1,\ldots,m$.

We first prove the architecture-independent all-output constraint
stated as Proposition~3 in the main text.

\begin{proof}
Because the encoder acts trivially on $\mathcal N$, the reduced state
of $\widetilde A\mathcal N$ remains unchanged from its initial product
form. Hence
\[
I(\widetilde A:\mathcal N)_\Omega=0.
\]
Exact recovery of the complete $k$-qubit input from
$S_i\mathcal N$ implies
\[
I(\widetilde A:S_i\mathcal N)_\Omega=2k.
\]
Using the chain rule,
\[
I(\widetilde A:S_i\mathcal N)_\Omega
=
I(\widetilde A:\mathcal N)_\Omega
+
I(\widetilde A:S_i|\mathcal N)_\Omega,
\]
we therefore obtain
\[
I(\widetilde A:S_i|\mathcal N)_\Omega=2k.
\]

Since $\dim S_i=2^k$,
\begin{align}
I(\widetilde A:S_i|\mathcal N)_\Omega
&=
S(S_i|\mathcal N)_\Omega
-
S(S_i|\widetilde A\mathcal N)_\Omega
\nonumber\\
&\le
k-(-k)
=
2k,
\label{eq:supp_conditional_MI_bound}
\end{align}
where
\[
S(S_i|\mathcal N)_\Omega\le k,
\qquad
S(S_i|\widetilde A\mathcal N)_\Omega\ge-k,
\]
follow from subadditivity and the Araki--Lieb inequality
~\cite{ArakiLieb1970,LiebRuskai1973}.
Since Eq.~\eqref{eq:supp_conditional_MI_bound} is saturated, both
bounds must be saturated individually:
\[
S(S_i|\mathcal N)_\Omega=k,
\qquad
S(S_i|\widetilde A\mathcal N)_\Omega=-k.
\]
In particular,
\begin{equation}
S(S_i\mathcal N)_\Omega
=
S(\mathcal N)_\Omega+k.
\label{eq:supp_SiN_entropy}
\end{equation}

Exact recovery also implies, by
Lemma~\ref{lem:supp_recovery_decoupling},
\[
\Omega_{\widetilde AAS_{\bar i}}
=
\frac{I_{\widetilde A}}{2^k}
\otimes
\Omega_{AS_{\bar i}}.
\]
Since the global encoded Choi state is pure,
\[
S(S_i\mathcal N)_\Omega
=
S(\widetilde AAS_{\bar i})_\Omega
=
k+S(AS_{\bar i})_\Omega.
\]
Comparing this with Eq.~\eqref{eq:supp_SiN_entropy} gives
\begin{equation}
S(AS_{\bar i})_\Omega
=
S(\mathcal N)_\Omega
\label{eq:supp_ASbar_entropy}
\end{equation}
for every $i$.

Purity of the global state also gives
\[
S(AS)_\Omega
=
S(\widetilde A\mathcal N)_\Omega.
\]
Since $\widetilde A$ and $\mathcal N$ are uncorrelated,
\[
S(AS)_\Omega
=
k+S(\mathcal N)_\Omega.
\]
Together with Eq.~\eqref{eq:supp_ASbar_entropy}, this yields
\[
S(S_i|AS_{\bar i})_\Omega
=
S(AS)_\Omega-S(AS_{\bar i})_\Omega
=
k.
\]

This is the maximal conditional entropy allowed for the
$2^k$-dimensional system $S_i$. Indeed,
\[
D\!\left(
\Omega_{AS}
\,\middle\|\,
\frac{I_{S_i}}{2^k}
\otimes
\Omega_{AS_{\bar i}}
\right)
=
k-S(S_i|AS_{\bar i})_\Omega
=
0.
\]
Hence
\begin{equation}
\Omega_{AS}
=
\frac{I_{S_i}}{2^k}
\otimes
\Omega_{AS_{\bar i}}
\label{eq:supp_each_output_factor}
\end{equation}
for every $i$.

Applying Eq.~\eqref{eq:supp_each_output_factor} successively to the
$m$ disjoint output registers gives
\begin{equation}
\Omega_{AS}
=
\Omega_A
\otimes
\bigotimes_{i=1}^{m}
\frac{I_{S_i}}{2^k}
=
\Omega_A
\otimes
\frac{I_S}{2^{mk}}.
\label{eq:supp_AS_full_factor}
\end{equation}
Therefore
\[
S(AS)_\Omega
=
S(A)_\Omega+mk.
\]
On the other hand,
\[
S(AS)_\Omega
=
k+S(\rho_{\mathcal N}).
\]
Comparing the two expressions gives
\begin{equation}
S(\rho_{\mathcal N})
=
(m-1)k+S(A)_\Omega
\ge
(m-1)k.
\label{eq:supp_all_output_entanglement_bound}
\end{equation}
Since the resource is pure,
$S(\rho_{\mathcal N})=S(\rho_S)$, and
Eq.~\eqref{eq:supp_all_output_entanglement_bound} is precisely the
claimed architecture-independent lower bound on the signal--noise
entanglement.

It remains to establish individual concealment. Since $m\ge2$, for
any fixed $i$ we may choose $j\neq i$. Exact recovery from
$S_j\mathcal N$ gives
\[
I(\widetilde A:AS_{\bar j})_\Omega=0.
\]
Because $S_i\subseteq S_{\bar j}$, data processing implies
\[
I(\widetilde A:S_i)_\Omega=0.
\]
Moreover, Eq.~\eqref{eq:supp_AS_full_factor} gives
\[
\Omega_{S_i}
=
\frac{I_{S_i}}{2^k}.
\]
Hence
\[
\Omega_{\widetilde A S_i}
=
\frac{I_{\widetilde A}}{2^k}
\otimes
\frac{I_{S_i}}{2^k},
\]
which is the Choi state of the completely depolarizing channel.
Therefore exact recovery from every $S_i\mathcal N$ necessarily
implies perfect individual concealment of every $S_i$.
\end{proof}

We now prove the stronger spectral consequence stated as
Corollary~3 in the main text.

\begin{proof}
Set
\[
d=2^k,
\qquad
D=2^{mk},
\qquad
L=\frac{D}{d}=2^{(m-1)k}.
\]
Before encoding, the $AS$ marginal is $\frac{I_A}{d}\otimes\rho_S.$ Since the encoder acts unitarily on $AS$,
$\Omega_{AS}
\simeq
\frac{I_A}{d}\otimes\rho_S,$ where $\simeq$ denotes unitary equivalence. On the other hand,
Eq.~\eqref{eq:supp_AS_full_factor} gives
\[
\Omega_{AS}
=
\Omega_A\otimes\frac{I_S}{D}.
\]
Therefore
\begin{equation}
\frac{I_A}{d}\otimes\rho_S
\simeq
\Omega_A\otimes\frac{I_S}{D}.
\label{eq:supp_spectral_equivalence}
\end{equation}

Let $\mu$ be a distinct eigenvalue of $\Omega_A$ with multiplicity
$n_\mu$. The corresponding eigenvalue on the right-hand side of
Eq.~\eqref{eq:supp_spectral_equivalence} is
$\frac{\mu}{D},$
with multiplicity $Dn_\mu$. If $\lambda$ is the corresponding
eigenvalue of $\rho_S$ with multiplicity $m_\lambda$, then the
left-hand side contains
$\frac{\lambda}{d}$ with multiplicity $d\,m_\lambda$. Equality of the two spectra therefore
requires
$$\frac{\lambda}{d}
=
\frac{\mu}{D},
\qquad
d\,m_\lambda
=
D\,n_\mu.$$

Since $D=dL$, this gives $\lambda=\frac{\mu}{L},
m_\lambda=L\,n_\mu.$ Thus every eigenvalue multiplicity of $\rho_S$ is divisible by $L$. Equivalently, if $\sigma$ is a $k$-qubit density operator with the
same spectrum as $\Omega_A$, then
\begin{equation}
\rho_S
\simeq
\sigma\otimes\frac{I_L}{L},
\qquad
L=2^{(m-1)k}.
\label{eq:supp_resource_spectral_form}
\end{equation}
Hence every distinct eigenvalue of $\rho_S$ has multiplicity divisible
by $2^{(m-1)k}.$

This proves Corollary~3. The condition is necessary for exact
all-output recovery; no sufficiency statement for an arbitrary
encoder is implied.
\end{proof}

\section{Exact characterization of graph-state resources}
\label{sec:supp_graph_resources}

This section provides the derivations underlying Sec.~III of the main
text. We first establish the graph-state cut algebra and characterize
the encoder-compatible cut-kernel directions. We then prove Theorem~2
and derive Corollaries~4 and~5 of the main text.

Throughout, $\nu=mk,$ and all vectors, ranks, and kernels are over $\mathbb F_2$. For a fixed
oriented balanced cut $S:\mathcal N$, write the graph adjacency matrix
as
\[
\Gamma
=
\begin{pmatrix}
A_S & B_{S\mathcal N}\\
B_{S\mathcal N}^{\mathsf T} & A_{\mathcal N}
\end{pmatrix}.
\]
We use the standard stabilizer representation of graph states
~\cite{Gottesman1997,hein2004}.
\subsection{Graph-state cut algebra}
\label{subsec:supp_graph_cut_algebra}

For $x\in\mathbb F_2^\nu$ supported on the signal vertices, define
\[
K(x)
=
\prod_{v\in S}K_v^{x_v},
\qquad
K_v
=
X_v\prod_{u\in N(v)}Z_u.
\]
Up to an irrelevant Pauli phase, its restrictions to the two sides of
the cut are
\[
K(x)\big|_S
\doteq
X^xZ^{A_Sx},
\qquad
K(x)\big|_{\mathcal N}
\doteq
Z^{B_{S\mathcal N}^{\mathsf T}x}.
\]
Hence a signal-generated stabilizer acts trivially on \(\mathcal N\)
precisely when $
x\in\ker(B_{S\mathcal N}^{\mathsf T}),
$ where \(\ker(B_{S\mathcal N}^{\mathsf T})\) denotes the set of binary
signal-side vectors \(x\) satisfying $
B_{S\mathcal N}^{\mathsf T}x=0.
$

\begin{lemma}[Reduced graph-state spectrum and cut kernel]
\label{lem:supp_graph_reduced_spectrum}
Let
\[
r=\rank_{\mathbb F_2}B_{S\mathcal N},
\qquad
\mathcal K=\ker(B_{S\mathcal N}^{\mathsf T}).
\]
For each $x\in\mathcal K$, define the signed signal operator $K_S(x)$
through
\[
K(x)=K_S(x)\otimes I_{\mathcal N}.
\]
Then
\begin{equation}
\rho_S
=
2^{-\nu}\sum_{x\in\mathcal K}K_S(x)
=
2^{-r}\Pi_{\mathcal K},
\label{eq:supp_graph_reduced_state}
\end{equation}
where $\Pi_{\mathcal K}$ is a projector of rank $2^r$. Consequently,
\[
\operatorname{SchmidtRank}_{S:\mathcal N}(\ket G)=2^r,
\qquad
S(\rho_S)=r.
\]
\end{lemma}

\begin{proof}[Proof of Lemma~\ref{lem:supp_graph_reduced_spectrum}]
A graph stabilizer labelled by binary vectors $x$ on $S$ and $z$ on
$\mathcal N$ restricts to the noise subsystem with binary Pauli labels
\[
\left(
z,\,
B_{S\mathcal N}^{\mathsf T}x+A_{\mathcal N}z
\right).
\]
Such a stabilizer survives the trace over $\mathcal N$ exactly when
\[
z=0,
\qquad
B_{S\mathcal N}^{\mathsf T}x=0.
\]
Using the standard stabilizer-state projector expansion
~\cite{Gottesman1997}, we therefore
obtain the first equality in
Eq.~\eqref{eq:supp_graph_reduced_state}.

Let
\[
d_{\mathcal K}
=
\dim\mathcal K
=
\nu-r,
\]
and choose a basis
$x_1,\ldots,x_{d_{\mathcal K}}$ of $\mathcal K$. The corresponding
independent commuting stabilizers define the projector
\[
\Pi_{\mathcal K}
=
\prod_{a=1}^{d_{\mathcal K}}
\frac{I+K_S(x_a)}{2}.
\]
Its rank is $2^{\nu-d_{\mathcal K}}
=
2^r.$ It follows that
\[
\rho_S
=
2^{-r}\Pi_{\mathcal K},
\]
which is the second equality in
Eq.~\eqref{eq:supp_graph_reduced_state}. Thus $\rho_S$ has a flat
nonzero spectrum of rank $2^r$. Since $\ket G$ is pure, the Schmidt
rank across the cut is $2^r$ and $S(\rho_S)=r.$

\end{proof}

Lemma~\ref{lem:supp_graph_reduced_spectrum} reproduces the standard
graph-state cut-rank relation
~\cite{hein2004,FattalEtAl2004}. In particular,
\begin{equation}
\rank_{\mathbb F_2}B_{S\mathcal N}=\nu
\quad\Longleftrightarrow\quad
\rho_S=\frac{I_S}{2^\nu}.
\label{eq:supp_graph_full_cut_rank}
\end{equation}
Combining
Lemma~\ref{lem:supp_graph_reduced_spectrum} with Proposition~3 of the
main text further gives the architecture-independent necessary
condition
\[
\rank_{\mathbb F_2}B_{S\mathcal N}
\ge
(m-1)k.
\]

\subsection{Exceptional kernel structure}
\label{subsec:supp_exceptional_kernel}

We now determine which cut-kernel stabilizers belong to the exceptional
signal algebra selected by Theorem~1. For odd $m$, define the
sector-indicator map
\begin{equation}
F:\mathbb F_2^k\longrightarrow\mathbb F_2^\nu,
\qquad
Fc=
(c_1\mathbf1_m,\ldots,c_k\mathbf1_m).
\label{eq:supp_sector_indicator_map}
\end{equation}

\begin{lemma}[Exceptional kernel structure]
\label{lem:supp_exceptional_structure}
For even $m$, no nonzero cut-kernel vector is exceptional. For odd $m$, every nonzero exceptional vector has the form $x=Fc$,
with $F$ defined in Eq.~\eqref{eq:supp_sector_indicator_map}. Within
the fixed graph-state Pauli frame, the three encoder classes satisfy
\[
\begin{array}{c|c}
\text{encoder class}
&
\text{exceptional condition}
\\
\hline
XZ/ZX
&
(A_S+I_\nu)Fc=0,
\quad
\mathbf1_k^{\mathsf T}c=0,
\\[1mm]
YZ/ZY
&
A_SFc=0,
\\[1mm]
XY/YX
&
\text{no nonzero exceptional vector}.
\end{array}
\]
\end{lemma}

\begin{proof}
Write
\[
x=(x_1,\ldots,x_k),
\qquad
y_j=(A_Sx)_j,
\]
where $x_j,y_j\in\mathbb F_2^m$ denote the restrictions to sector
$j$. The corresponding signal Pauli is
\[
K_S(x)
\doteq
X^xZ^{A_Sx},
\]
whose restriction to sector $j$ is proportional to
$X^{x_j}Z^{y_j}$.

Because the encoder factorizes across sectors, every active sector
must independently lie along the exceptional direction identified in
Lemma~\ref{lem:supp_encoder_pauli_action}. For even $m$, that lemma
immediately excludes every nonzero $x$.

Now let $m$ be odd. The unique nonidentity exceptional Pauli in each
active sector is $R^{\otimes m}$, where $R=iPQ$.

For the $XZ/ZX$ class, $R=\pm Y$. Hence every active sector must
satisfy $x_j=y_j=\mathbf1_m.$

Therefore
\[
x=Fc,
\qquad
A_SFc=Fc,
\]
or equivalently,
\[
(A_S+I_\nu)Fc=0.
\]
Since $A_S$ is symmetric with zero diagonal,
\[
x^{\mathsf T}A_Sx=0
\]
over $\mathbb F_2$. Using $A_Sx=x$, we obtain
\[
0
=
x^{\mathsf T}x
=
m\,\operatorname{wt}(c)
\pmod 2.
\]
Because $m$ is odd, this implies
$\mathbf1_k^{\mathsf T}c=0.$

Thus the even-sector-parity condition follows automatically from the
graph relation once the kernel vector has the complete-sector form
$x=Fc$.

For the $YZ/ZY$ class, $R=\pm X$. Every active sector must therefore
satisfy
\[
x_j=\mathbf1_m,
\qquad
y_j=0,
\]
which gives
\[
x=Fc,
\qquad
A_SFc=0,
\]
with no additional parity constraint.

Finally, for the $XY/YX$ class, $R=\pm Z$. However,
\[
K_S(x)\doteq X^xZ^{A_Sx}
\]
has nontrivial $X$ support whenever $x\neq0$. Hence no nonzero
cut-kernel vector can lie entirely along the required $Z$ direction.
\end{proof}

Accordingly, the exceptional subspaces appearing in Theorem~2 are
\begin{equation}
\mathcal E_{P,Q}^{(m,k)}(A_S)
=
\left\{
\begin{array}{@{}ll@{}}
\{0\},
&
m\ \mathrm{even},
\\[1.5mm]

\left\{
Fc:
\substack{
(A_S+I_\nu)Fc=0,\\
\mathbf1_k^{\mathsf T}c=0
}
\right\},
&
\substack{
m\ \mathrm{odd},\\
(P,Q)\in\{(X,Z),(Z,X)\},
}
\\[2mm]

\{Fc:A_SFc=0\},
&
\substack{
m\ \mathrm{odd},\\
(P,Q)\in\{(Y,Z),(Z,Y)\},
}
\\[1.5mm]

\{0\},
&
\substack{
m\ \mathrm{odd},\\
(P,Q)\in\{(X,Y),(Y,X)\}.
}
\end{array}
\right.
\label{eq:supp_exceptional_subspaces}
\end{equation}
\subsection{Proof of Theorem~2 of the main text}
\label{subsec:supp_proof_graph_exact}

We now prove the exact graph-state criterion stated as Theorem~2 in
the main text.

\begin{proof}
By Lemma~\ref{lem:supp_graph_reduced_spectrum}, the signal marginal of
the graph state is
\[
\rho_S
=
2^{-\nu}
\sum_{x\in\mathcal K}K_S(x),
\qquad
\mathcal K
=
\ker(B_{S\mathcal N}^{\mathsf T}).
\]
Every $x\in\mathcal K$ appears with a nonzero coefficient, and
distinct kernel vectors correspond to distinct signal stabilizers.

By Theorem~1, exact recovery from every $S_i\mathcal N$ under the
fixed two-Pauli encoder is possible if and only if every nonidentity
Pauli component of $\rho_S$ belongs to the exceptional signal algebra.
Lemma~\ref{lem:supp_exceptional_structure}, equivalently
Eq.~\eqref{eq:supp_exceptional_subspaces}, identifies precisely which
graph-state stabilizers have this property. Therefore the exact
condition is
\begin{equation}
\ker(B_{S\mathcal N}^{\mathsf T})
\subseteq
\mathcal E_{P,Q}^{(m,k)}(A_S).
\label{eq:supp_graph_exact_criterion}
\end{equation}

Eq.~\eqref{eq:supp_graph_exact_criterion} is both necessary and
sufficient. Necessity is independent of the recovery map, since
Theorem~1 already allows an arbitrary CPTP decoder on each authorized
subsystem $S_i\mathcal N$. Conversely, whenever
Eq.~\eqref{eq:supp_graph_exact_criterion} holds, Theorem~1 guarantees
exact recovery from every $S_i\mathcal N$ and perfect concealment of
every individual signal $S_i$.

This proves Theorem~2.
\end{proof}

The result shows why cut rank alone does not characterize the
rank-deficient branches. If
\[
\rank_{\mathbb F_2}B_{S\mathcal N}<mk,
\]
the cut kernel is nontrivial. Exact recovery remains possible only
when every missing cut direction lies in the encoder-compatible
exceptional subspace. Thus the cut rank quantifies the amount of
signal--noise entanglement, whereas the complete kernel determines
whether the corresponding entanglement deficit is compatible with
the fixed encoder.

\subsection{Proof of Corollary~4 of the main text}
\label{subsec:supp_proof_graph_classes}

The parity- and encoder-dependent rank consequences follow directly
from Theorem~2, Lemma~\ref{lem:supp_exceptional_structure}, and
rank--nullity.

\begin{proof}
For even $m$, the exceptional subspace is trivial. Theorem~2 therefore
requires
\[
\ker(B_{S\mathcal N}^{\mathsf T})=\{0\},
\]
which is equivalent to
\[
\rank_{\mathbb F_2}B_{S\mathcal N}=mk.
\]

Now let $m$ be odd. For the $XZ/ZX$ class, every exceptional vector
has the form $Fc$ with $\mathbf1_k^{\mathsf T}c=0.$ The allowed coefficient space therefore has dimension at most $k-1$.
Since $F$ is injective,
\begin{equation}
\dim\ker(B_{S\mathcal N}^{\mathsf T})
\le
k-1,
\qquad
\rank_{\mathbb F_2}B_{S\mathcal N}
\ge
mk-k+1.
\label{eq:supp_XZ_rank_bound}
\end{equation}

For the $YZ/ZY$ class, the exceptional subspace is contained in
$\operatorname{im}F$, whose dimension is $k$. Hence
\begin{equation}
\dim\ker(B_{S\mathcal N}^{\mathsf T})
\le
k,
\qquad
\rank_{\mathbb F_2}B_{S\mathcal N}
\ge
(m-1)k.
\label{eq:supp_YZ_rank_bound}
\end{equation}
The rank bound in Eq.~\eqref{eq:supp_YZ_rank_bound} coincides with the
architecture-independent lower bound of Proposition~3. It is not,
however, sufficient by itself: every kernel vector must also satisfy
the exact compatibility condition $A_SFc=0.$

For the $XY/YX$ class, the exceptional subspace is again trivial.
Thus
\[
\rank_{\mathbb F_2}B_{S\mathcal N}=mk
\]
is necessary and sufficient.
\end{proof}

Within the standard graph-state Pauli frame, the $YZ/ZY$ branch is the
only class whose exceptional subspace can attain nullity $k$. If
\[
\dim\mathcal K=k,
\qquad
\mathcal K
\subseteq
\mathcal E_{P,Q}^{(m,k)}(A_S),
\]
then
\begin{equation}
S(\rho_S)
=
\rank_{\mathbb F_2}B_{S\mathcal N}
=
(m-1)k.
\label{eq:supp_graph_entanglement_floor}
\end{equation}
Thus the graph-state resource saturates the architecture-independent
signal--noise entanglement floor. By contrast, the $XZ/ZX$ branch has
nullity at most $k-1$, as shown in
Eq.~\eqref{eq:supp_XZ_rank_bound}.

These distinctions are relative to the chosen Pauli frame. Any ordered
pair of distinct Pauli axes is related to any other by a single-qubit
Clifford conjugation. A simultaneous Clifford rotation of the input and
signal systems therefore maps the encoder together with its compatible
resource representative without changing the signal--noise
entanglement. The classification above thus concerns compatibility
between a fixed graph-state representative and a fixed encoder frame,
rather than an absolute inequivalence among different Pauli-axis
encoders.

\subsection{Proof of Corollary~5 of the main text}
\label{subsec:supp_proof_k1_classification}

For $k=1$, the sector-indicator map generates only the one-dimensional
subspace spanned by $\mathbf1_m$:
\[
\operatorname{im}F
=
\operatorname{span}\{\mathbf1_m\}.
\]

\begin{proof}
For the $XZ/ZX$ class, the only possible nonzero coefficient is
$c=1$, while the even-sector-parity condition requires
\[
\mathbf1_1^{\mathsf T}c=0.
\]
Hence no nonzero exceptional vector exists, and full cut rank is
necessary and sufficient.

For the $XY/YX$ class, the exceptional subspace is already trivial,
so the same full-rank condition follows.

For the $YZ/ZY$ class with even $m$, the general even-$m$ result again
requires full cut rank. Now let $m$ be odd. In this case, the
exceptional subspace is contained in
$\operatorname{span}\{\mathbf1_m\}$ and therefore has dimension at
most one. Besides the full-rank branch, the only possible
rank-deficient valid case is
\begin{equation}
\ker(B_{S\mathcal N}^{\mathsf T})
=
\operatorname{span}\{\mathbf1_m\},
\qquad
A_S\mathbf1_m=0.
\label{eq:supp_k1_rank_deficient_branch}
\end{equation}
The kernel in Eq.~\eqref{eq:supp_k1_rank_deficient_branch} has
dimension one. Rank--nullity therefore gives
\[
\rank_{\mathbb F_2}B_{S\mathcal N}
=
m-1.
\]
This proves Corollary~5.
\end{proof}

The graph-state criterion therefore separates three distinct pieces of
information. The cut rank determines the amount of signal--noise
entanglement, the full cut kernel identifies the correlations
responsible for any entanglement deficit, and the fixed encoder frame
determines whether those correlations are exceptional or obstructive.
Full cut rank completely characterizes the even-$m$ and $XY/YX$
branches, whereas the exceptional odd-$m$ branches require the full
kernel test.

Appendix~C turns this fixed-realization criterion into
\textsc{GSECC} and provides constructive Clifford recovery for every
certified realization.

\section{Certification and constructive decoding}
\label{sec:supp_certification}

This section provides the technical details underlying Sec.~IV of the
main text. We first derive the fixed-cut implementation of
\emph{Graph-State Encrypted-Cloning Certification} (\textsc{GSECC}),
including the row-signature construction of an admissible sector
decomposition. We then prove Proposition~4, Observation~1, and
Observation~2 of the main text, and conclude with the associated
Clifford-synthesis costs.

\subsection{Exact \textsc{GSECC} implementation}
\label{subsec:supp_gsecc}

Let $G=(V,E)$ be a graph on $2\nu$ vertices, with $\nu=mk$. An
$(m,k)$ realization consists of an oriented balanced partition
\[
V=S\sqcup\mathcal N,
\qquad
|S|=|\mathcal N|=\nu,
\]
together with a partition
\[
\Pi_S=\{\mathsf S_1,\ldots,\mathsf S_k\}
\]
of the signal vertices into $k$ sectors of size $m$. We denote the
complete realization by
\begin{equation}
\mathfrak R=(S:\mathcal N,\Pi_S).
\label{eq:supp_gsecc_realization}
\end{equation}

For a fixed encoder $U_{P,Q}^{(m,k)}$, the graph is
\textsc{GSECC}-valid precisely when at least one realization in
Eq.~\eqref{eq:supp_gsecc_realization} satisfies the exact kernel
condition of Eq.~\eqref{eq:supp_graph_exact_criterion}. Certification
therefore has two levels: a fixed-cut test and an existential search
over oriented balanced cuts.

For a prescribed cut $S:\mathcal N$, extract $A_S$ and
$B\equiv B_{S\mathcal N}$ from the adjacency matrix and compute a
kernel-basis matrix
\begin{equation}
W=
\begin{pmatrix}
w_1&\cdots&w_\eta
\end{pmatrix}
\in\mathbb F_2^{\nu\times\eta},
\qquad
\operatorname{col}(W)=\ker B^{\mathsf T},
\qquad
\eta=\dim\ker B^{\mathsf T}.
\label{eq:supp_gsecc_kernel_basis}
\end{equation}
The columns of $W$ can be obtained by binary Gaussian elimination.

For even $m$, and for odd $m$ in the $XY/YX$ encoder class, the
exceptional subspace is trivial. Hence the cut is admissible exactly
when
\[
\eta=0,
\qquad\text{equivalently}\qquad
\rank_{\mathbb F_2}B=\nu.
\]
The sector decomposition is irrelevant in these branches.

For the two exceptional odd-$m$ branches, an admissible sector
decomposition can be reconstructed directly from $W$, without
enumerating partitions of the signal vertices.

\begin{lemma}[Row-signature sector criterion]
\label{lem:supp_row_signature}
For each signal vertex $v\in S$, define its kernel row signature by
\begin{equation}
\sigma(v)=W_{v,:}\in\mathbb F_2^\eta.
\label{eq:supp_row_signature}
\end{equation}
There exists a partition of $S$ into $k$ sectors of size $m$ on which
every vector in $\ker B^{\mathsf T}$ is constant if and only if every
row-signature class defined by Eq.~\eqref{eq:supp_row_signature} has
cardinality divisible by $m$.
\end{lemma}

\begin{proof}
Suppose first that such a sector decomposition exists. If two
vertices $u$ and $v$ lie in the same sector, every
$x\in\ker B^{\mathsf T}$ satisfies $x_u=x_v$. In particular, this
holds for each basis vector $w_a$, so
\[
W_{u,:}=W_{v,:}.
\]
Thus every sector lies entirely within a single row-signature class.
Each signature class is therefore a disjoint union of sectors of size
$m$, and its cardinality must be divisible by $m$.

Conversely, suppose every row-signature class has cardinality
divisible by $m$. Partition each class arbitrarily into blocks of
size $m$. Within every such block, the corresponding rows of $W$ are
identical. Hence every column of $W$, and therefore every vector in
\[
\operatorname{col}(W)=\ker B^{\mathsf T},
\]
is constant on that block. The resulting blocks form the required
sector decomposition.

The criterion is independent of the chosen kernel basis. Indeed,
replacing $W$ by $WM$ for any invertible
$M\in\operatorname{GL}(\eta,\mathbb F_2)$ preserves equality of row
signatures.
\end{proof}

It remains to impose the encoder-specific exceptional-subspace
condition. For the $XZ/ZX$ class, every kernel vector must satisfy
\begin{equation}
(A_S+I_\nu)W=0.
\label{eq:supp_gsecc_XZ_condition}
\end{equation}
For odd $m$, Eq.~\eqref{eq:supp_gsecc_XZ_condition}, together with the
graph constraint, already implies the even-sector-parity condition
derived in Sec.~\ref{subsec:supp_exceptional_kernel}; no independent
parity test is required once a sector-constant decomposition exists.

For the $YZ/ZY$ class, the corresponding condition is
\begin{equation}
A_SW=0.
\label{eq:supp_gsecc_YZ_condition}
\end{equation}
Because Eqs.~\eqref{eq:supp_gsecc_XZ_condition} and
\eqref{eq:supp_gsecc_YZ_condition} are linear, checking them on the
kernel-basis matrix $W$ is equivalent to checking every vector in
$\ker B^{\mathsf T}$.

The exact fixed-cut procedure is therefore as follows. Compute
$A_S$, $B$, and a kernel basis $W$ as in
Eq.~\eqref{eq:supp_gsecc_kernel_basis}. In the full-rank branches,
accept exactly when $\eta=0$. In an exceptional odd-$m$ branch, group
the signal vertices according to the row signatures in
Eq.~\eqref{eq:supp_row_signature}, reject if any signature class has
cardinality not divisible by $m$, and then test
Eq.~\eqref{eq:supp_gsecc_XZ_condition} or
Eq.~\eqref{eq:supp_gsecc_YZ_condition}, as appropriate. If both tests
pass, split each signature class into blocks of size $m$. These blocks
provide an admissible sector assignment satisfying the exact kernel
criterion.

All fixed-cut operations involve binary matrices of size $O(\nu)$,
and ordinary Gaussian elimination gives an $O(\nu^3)$ upper bound.
Row grouping and the remaining matrix tests are lower-order within
this bound.

At the graph level, the oriented signal subsystem can be chosen in
$\binom{2\nu}{\nu}$ ways. Since the row-signature construction removes the separate
enumeration of sector partitions, a direct exhaustive implementation
has the upper bound
\begin{equation}
O\left[
\binom{2\nu}{\nu}\nu^3
\right].
\label{eq:supp_gsecc_exhaustive_cost}
\end{equation}
Eq.~\eqref{eq:supp_gsecc_exhaustive_cost} is an exhaustive-search bound
for the stated certification procedure, not a complexity
classification of \textsc{GSECC}. Graph automorphisms, equivalent
cuts, rank pruning, and early signature rejection can further reduce
the practical search.

Failure of one oriented cut does not invalidate the graph.
\textsc{GSECC} returns \textsc{valid} whenever at least one cut admits
the required sector structure and satisfies the encoder-compatible
kernel condition. A successful certification therefore supplies
precisely the cut and sector data required for constructive recovery.
\subsection{Proof of Proposition~4 of the main text}
\label{subsec:supp_constructive_decoder}

We now prove the constructive Clifford-decoder statement of
Proposition~4.

For a certified realization, define the encoding isometry
\[
V_G\ket{\psi}_A
=
U_{P,Q}^{(m,k)}
\bigl(
\ket{\psi}_A\otimes\ket G_{S\mathcal N}
\bigr).
\]
Since $\ket G$ is a stabilizer state and the sector-wise two-Pauli
encoder is Clifford, the image of $V_G$ is a
$2^k$-dimensional stabilizer code in $AS\mathcal N$
~\cite{Gottesman1997}.

For recovery from the $i$th authorized subsystem, write
\[
\mathcal R_i=S_i\mathcal N,
\qquad
\mathcal E_i=AS_{\bar i},
\]
for the authorized and erased subsystems, respectively.

\begin{proof}
Let $\{K_v\}$ be independent generators of the graph-state
stabilizer. Their encoded images
\[
\widetilde K_v
=
U_{P,Q}^{(m,k)}
K_v
U_{P,Q}^{(m,k)\dagger}
\]
stabilize the code $V_G(\mathcal H_A)$. Representatives of the
encoded input Pauli operators may be chosen as
\[
\widehat X_j
=
U_{P,Q}^{(m,k)}
X_{A_j}
U_{P,Q}^{(m,k)\dagger},
\qquad
\widehat Z_j
=
U_{P,Q}^{(m,k)}
Z_{A_j}
U_{P,Q}^{(m,k)\dagger}.
\]

Because the realization satisfies
Theorem~2 of the main text, the complete $k$-qubit input is
exactly recoverable from $\mathcal R_i$. By the
recovery--decoupling criterion,
\[
I(\widetilde A:\mathcal E_i)_\Omega=0,
\]
so $\mathcal E_i$ is an exactly correctable erasure region for the
stabilizer code defined by $V_G$.

The stabilizer cleaning lemma
~\cite{BravyiTerhal2009}
therefore guarantees that every encoded logical Pauli has an
equivalent representative supported entirely on $\mathcal R_i$.
Hence there exist encoded stabilizer products
$S_{j,X}^{(i)}$ and $S_{j,Z}^{(i)}$ such that
\begin{equation}
\overline X_j^{(i)}
=
\widehat X_j S_{j,X}^{(i)},
\qquad
\overline Z_j^{(i)}
=
\widehat Z_j S_{j,Z}^{(i)}
\label{eq:supp_cleaned_logical_paulis}
\end{equation}
have no support on $\mathcal E_i$. These operators retain the encoded
input action,
\[
\overline X_j^{(i)}V_G
=
V_GX_{A_j},
\qquad
\overline Z_j^{(i)}V_G
=
V_GZ_{A_j},
\]
and therefore satisfy the canonical Pauli symplectic relations.

Any canonical symplectic set can be completed to a full symplectic
basis and mapped to standard single-qubit Pauli generators by a
Clifford transformation
~\cite{DehaeneDeMoor2003,aaronson2004improved}. Hence there exist a Clifford unitary
$C_i$ and a factorization
\[
\mathcal H_{\mathcal R_i}
\simeq
\mathcal H_{A_i'}\otimes\mathcal H_{\mathcal J_i},
\qquad
\dim A_i'=2^k,
\]
such that
\begin{equation}
\begin{aligned}
C_i\overline X_j^{(i)}C_i^\dagger
&=
X_{A_j'}\otimes I_{\mathcal J_i},
\\
C_i\overline Z_j^{(i)}C_i^\dagger
&=
Z_{A_j'}\otimes I_{\mathcal J_i}.
\end{aligned}
\label{eq:supp_clifford_logical_standardization}
\end{equation}
Eq.~\eqref{eq:supp_clifford_logical_standardization} extracts, on the
encoded Choi state, a maximally entangled partner of $\widetilde A$
into $A_i'$. The decoder
\begin{equation}
\mathcal D_i(\rho)
=
\Tr_{\mathcal J_i}
\!\left[
C_i\rho C_i^\dagger
\right]
\label{eq:supp_constructive_decoder}
\end{equation}
therefore satisfies
\[
\mathcal D_i\circ\mathcal Q_i=\id_A.
\]
This proves Proposition~4.
\end{proof}

The construction is algorithmic. Represent the encoded input Paulis
by binary symplectic vectors $\ell_{j,\mu}$, with
$\mu\in\{X,Z\}$, and let the rows of $H_G$ generate the stabilizer of
the encoded subspace. If $\Pi_{\mathcal E_i}$ denotes restriction of
a Pauli vector to the erased subsystem, then cleaning the logical
operators in Eq.~\eqref{eq:supp_cleaned_logical_paulis} reduces to
solving
\begin{equation}
\Pi_{\mathcal E_i}
\left[
\ell_{j,\mu}
+
c_{j,\mu}^{(i)}H_G
\right]
=
0
\qquad
\text{over }\mathbb F_2.
\label{eq:supp_cleaning_linear_system}
\end{equation}
Exact correctability guarantees a solution of
Eq.~\eqref{eq:supp_cleaning_linear_system} for every $j$ and $\mu$.
After these systems are solved, symplectic Gaussian elimination maps
the cleaned representatives to the standard Pauli generators in
Eq.~\eqref{eq:supp_clifford_logical_standardization} and synthesizes
the Clifford $C_i$
~\cite{DehaeneDeMoor2003,aaronson2004improved}.

The construction applies to every certified branch, including the
rank-deficient exceptional odd-$m$ cases. Full cut rank is therefore
not required for exact Clifford recovery. The corresponding
coherent-information and quantum-capacity saturation has already been
stated in Eq.~({\color{blue}2}) of the main text and
is not repeated here.

\subsection{Proof of Observation~1 of the main text}
\label{subsec:supp_explicit_bell_reduction}

We next prove the explicit noise-side Bell reduction stated as
Observation~1. Let
$B\equiv B_{S\mathcal N}$ be invertible. We use the standard
computational-basis representation of graph states
~\cite{hein2004}.

\begin{proof}
The graph state can be written as
\begin{equation}
\ket G
=
2^{-\nu}
\sum_{s,n\in\mathbb F_2^\nu}
(-1)^{
q_S(s)+s^{\mathsf T}Bn+q_{\mathcal N}(n)
}
\ket s_S\ket n_{\mathcal N},
\label{eq:supp_graph_computational_form}
\end{equation}
where
\[
q_S(s)
=
\sum_{a<b}(A_S)_{ab}s_as_b,
\qquad
q_{\mathcal N}(n)
=
\sum_{a<b}(A_{\mathcal N})_{ab}n_an_b.
\]

Applying $CZ_{A_{\mathcal N}}^{(\mathcal N)}$ to
Eq.~\eqref{eq:supp_graph_computational_form} cancels the
noise-internal graph phase and gives
\[
2^{-\nu}
\sum_{s,n}
(-1)^{q_S(s)+s^{\mathsf T}Bn}
\ket s\ket n.
\]

Next apply $H_{\mathcal N}^{\otimes\nu}$. Since
\[
H^{\otimes\nu}\ket n
=
2^{-\nu/2}
\sum_{y\in\mathbb F_2^\nu}
(-1)^{n^{\mathsf T}y}\ket y,
\]
the sum over $n$ satisfies
\[
\sum_n
(-1)^{n^{\mathsf T}(B^{\mathsf T}s+y)}
=
2^\nu\delta_{y,B^{\mathsf T}s}.
\]
The state therefore becomes
\[
2^{-\nu/2}
\sum_s
(-1)^{q_S(s)}
\ket s_S
\ket{B^{\mathsf T}s}_{\mathcal N}.
\]

Because $B$ is invertible, the linear reversible map
$L_{B^{-\mathsf T}}$ sends
$\ket{B^{\mathsf T}s}$ to $\ket s$, yielding
\[
2^{-\nu/2}
\sum_s
(-1)^{q_S(s)}
\ket s_S\ket s_{\mathcal N}.
\]
Finally, $CZ_{A_S}^{(\mathcal N)}$ contributes the phase
$(-1)^{q_S(s)}$ to the correlated noise basis state and therefore
cancels the remaining signal-internal phase. Hence
\begin{equation}
(I_S\otimes U_{\mathcal N})\ket G
=
\ket{\Phi_{2^\nu}}_{S\mathcal N},
\label{eq:supp_bell_reduction}
\end{equation}
where
\begin{equation}
U_{\mathcal N}
=
CZ_{A_S}^{(\mathcal N)}
L_{B^{-\mathsf T}}
H_{\mathcal N}^{\otimes\nu}
CZ_{A_{\mathcal N}}^{(\mathcal N)}.
\label{eq:supp_noise_clifford}
\end{equation}
Eqs.~\eqref{eq:supp_bell_reduction} and
\eqref{eq:supp_noise_clifford} establish the explicit Bell reduction
claimed in Observation~1 of the main text.
\end{proof}

Thus every full-cut-rank graph resource is explicitly related to the
canonical maximally entangled resource by a Clifford acting only on
$\mathcal N$, consistent with the standard bipartite normal form of
stabilizer entanglement
~\cite{FattalEtAl2004}.

Because the Clifford in Eq.~\eqref{eq:supp_noise_clifford} acts only
on $\mathcal N$, it commutes with the sector-wise encoder. The
authorized recovery may therefore be written as
\begin{equation}
\mathcal D_i
=
\mathcal D_i^{\rm Bell}
\circ
\operatorname{Ad}_{I_{S_i}\otimes U_{\mathcal N}},
\label{eq:supp_bell_resource_decoder}
\end{equation}
where $\mathcal D_i^{\rm Bell}$ denotes an exact decoder for the
canonical Bell resource. Eq.~\eqref{eq:supp_bell_resource_decoder}
shows that the noise-side Clifford provides a closed-form
canonicalization of the resource, but is not by itself the complete
authorized recovery map.


\subsection{Proof of Observation~2 of the main text}
\label{subsec:supp_bell_reduction_locality}

We now prove the locality criterion stated as Observation~2. The
statement concerns the specific noise-side Bell canonicalization of
Observation~1, rather than arbitrary recovery maps.

\begin{proof}
Set
\[
M=B^{-\mathsf T}.
\]
Represent a noise Pauli by its binary symplectic vector $(x,z)$. Under
the Clifford
\[
U_{\mathcal N}
=
CZ_{A_S}
L_M
H^{\otimes\nu}
CZ_{A_{\mathcal N}},
\]
the Pauli labels transform as
\begin{equation}
\begin{aligned}
x'
&=
B^{-\mathsf T}A_{\mathcal N}x
+
B^{-\mathsf T}z,
\\
z'
&=
\bigl(
B+A_SB^{-\mathsf T}A_{\mathcal N}
\bigr)x
+
A_SB^{-\mathsf T}z.
\end{aligned}
\label{eq:supp_noise_symplectic_action}
\end{equation}

Suppose first that $U_{\mathcal N}$ is a tensor product of
single-qubit Clifford operations, up to a permutation of the noise
qubits. Such a transformation cannot spread a single-qubit Pauli over
more than one output qubit.

Consider the input $Z_j$, corresponding to $x=0$ and $z=e_j$.
From Eq.~\eqref{eq:supp_noise_symplectic_action},
\[
x'=B^{-\mathsf T}e_j.
\]
Locality therefore requires every column of $B^{-\mathsf T}$ to have
Hamming weight one. Since $B^{-\mathsf T}$ is invertible, it must be a
permutation matrix, and hence so must $B$.

For a permutation matrix, $B^{-\mathsf T}=B$. Eq.~\eqref{eq:supp_noise_symplectic_action}
then gives for the image of $Z_j$
\[
x'=Be_j,
\qquad
z'=A_SBe_j.
\]
The $Z$ component may have support only on the same output qubit as
$Be_j$. Because $A_S$ has zero diagonal, it has no component on that
qubit. Hence
\[
A_SBe_j=0
\]
for every $j$, which implies
$A_S=0.$

With $A_S=0$, Eq.~\eqref{eq:supp_noise_symplectic_action} gives for
the input $X_j$
\[
x'=BA_{\mathcal N}e_j,
\qquad
z'=Be_j.
\]
The same locality argument, together with the zero diagonal of
$A_{\mathcal N}$, yields $A_{\mathcal N}=0.$

Thus locality requires
\begin{equation}
A_S=A_{\mathcal N}=0,
\qquad
B\ \text{a permutation matrix}.
\label{eq:supp_noise_locality_condition}
\end{equation}

Conversely, suppose Eq.~\eqref{eq:supp_noise_locality_condition}
holds. Both controlled-$Z$ layers are then trivial, and
\begin{equation}
U_{\mathcal N}
=
L_{B^{-\mathsf T}}H^{\otimes\nu}.
\label{eq:supp_local_noise_clifford}
\end{equation}
Since $B$ is a permutation matrix,
$L_{B^{-\mathsf T}}$ merely permutes the noise qubits. Thus
Eq.~\eqref{eq:supp_local_noise_clifford} is a tensor product of
single-qubit Hadamard operations up to a wire permutation. This proves
Observation~2.
\end{proof}

Observation~2 therefore characterizes locality only for the explicit
noise-side Bell canonicalization of Observation~1. It does not
classify the locality of arbitrary CPTP or Clifford recovery maps
acting on $S_i\mathcal N$.

\subsection{Clifford synthesis and circuit cost}
\label{subsec:supp_clifford_synthesis}

The general and full-rank constructions lead to complementary
Clifford-synthesis procedures.

For an arbitrary certified realization, including rank-deficient
exceptional cases, the cleaning equations are solved first. The
resulting authorized Pauli representatives form a canonical
symplectic set, from which the decoding Clifford $C_i$ is obtained by
standard symplectic Gaussian elimination
~\cite{DehaeneDeMoor2003,aaronson2004improved}.

For a full-rank cut,
Observation~1 of the main text gives the explicit
noise-side sequence
\[
CZ_{A_{\mathcal N}}
\;\longrightarrow\;
H^{\otimes\nu}
\;\longrightarrow\;
L_{B^{-\mathsf T}}
\;\longrightarrow\;
CZ_{A_S}.
\]
The reversible map $L_{B^{-\mathsf T}}$ can be synthesized with
$O(\nu^2)$ CNOT gates by Gaussian elimination
~\cite{PatelMarkovHayes2008}.
Each controlled-$Z$ layer contains at most
$\nu(\nu-1)/2$ gates, so the complete Bell reduction has an
$O(\nu^2)$ Clifford-gate upper bound under unrestricted connectivity.

The actual cost is graph dependent. If $B$ is a permutation matrix,
the linear-reversible layer reduces to a wire permutation, while
sparsity of $A_S$ and $A_{\mathcal N}$ reduces the controlled-$Z$
count. The fully local limit of
Observation~2 of the main text occurs when
\[
A_S=A_{\mathcal N}=0,
\qquad
B\ \text{is a permutation matrix}.
\]
Explicit certificates and decoders for representative resource
families are given in Appendix~\ref{sec:supp_resource_landscape}.
\section{Resource landscape and representative families}
\label{sec:supp_resource_landscape}

This section provides the explicit constructions and certificates
underlying Sec.~V of the main text. We first establish scalable
full-rank graph families and record the representative $24$-qubit
certificates. We then analyze the rank-deficient example exhibiting
encoder dependence and conclude with the fixed-encoder and
architecture-independent obstructions for Dicke and $W$ resources.

\subsection{Scalable full-rank graph families}
\label{subsec:supp_scalable_families}

The full-rank branch of Theorem~2 immediately yields scalable
resources whenever a graph admits a balanced signal--noise cut with
invertible cut matrix. We denote the rectangular
$2a\times b$ cluster graph by
\[
G_{2a\times b}^{\mathrm{cl}}
\equiv
P_{2a}\square P_b,
\qquad
ab=\nu .
\]

\begin{proposition}[Scalable full-rank graph families]
\label{prop:supp_scalable_graphs}
Let $\nu=mk$. Perfect matchings, paths $P_{2\nu}$, cycles
$C_{2\nu}$, and rectangular cluster graphs
$G_{2a\times b}^{\mathrm{cl}}$ with $ab=\nu$ admit balanced cuts for
which
\[
\rank_{\mathbb F_2}B_{S\mathcal N}=\nu .
\]
Hence each provides a valid resource for every factorization
$\nu=mk$ with $m\ge2$.
\end{proposition}

\begin{proof}
For a perfect matching, choose one endpoint of every edge as signal
and the other as noise. With the corresponding ordering,
\[
A_S=A_{\mathcal N}=0,
\qquad
B_{S\mathcal N}=I_\nu .
\]

For the path $P_{2\nu}$ with vertices
$0,1,\ldots,2\nu-1$, choose
\[
S=(0,2,\ldots,2\nu-2),
\qquad
\mathcal N=(1,3,\ldots,2\nu-1).
\]
The cut matrix is
\[
B_{S\mathcal N}
=
\begin{pmatrix}
1&0&0&\cdots&0\\
1&1&0&\cdots&0\\
0&1&1&\ddots&\vdots\\
\vdots&\ddots&\ddots&\ddots&0\\
0&\cdots&0&1&1
\end{pmatrix},
\]
which is triangular with unit diagonal and is therefore invertible
over $\mathbb F_2$.

For $C_{2\nu}$, one possible balanced cut is
\[
S=\{0,1,3,5,\ldots,2\nu-3\},
\qquad
\mathcal N=\{2,4,\ldots,2\nu-2,2\nu-1\}.
\]
Successive binary row additions reduce the corresponding cut matrix
to the identity, so this cut also has full rank.

Finally, consider $G_{2a\times b}^{\mathrm{cl}}$. Choose alternating
horizontal rows as signal and noise and order them pairwise. The cut
matrix has block-bidiagonal form
\[
B_{S\mathcal N}
=
\begin{pmatrix}
I_b&0&\cdots&0\\
I_b&I_b&\ddots&\vdots\\
0&I_b&\ddots&0\\
\vdots&\ddots&\ddots&I_b
\end{pmatrix},
\]
and is therefore invertible over $\mathbb F_2$.
\end{proof}

The path and rectangular-grid constructions include the standard
linear and cluster graph states
~\cite{hein2004}.
Their full cut rank implies maximal signal--noise entanglement through
the standard graph-state cut-rank relation
~\cite{FattalEtAl2004}.

At the opposite extreme, complete and star graphs have only one
independent binary correlation across every balanced cut. For
$K_{2\nu}$,
\[
B_{S\mathcal N}=J_\nu,
\qquad
\rank_{\mathbb F_2}B_{S\mathcal N}=1,
\]
where $J_\nu$ is the $\nu\times\nu$ all-ones matrix. A star graph
likewise has balanced-cut rank one. Its graph state belongs to the
local-Clifford orbit of the GHZ state
~\cite{hein2004}.
Neither family therefore satisfies the exact graph-state criterion for
$m\ge2$.


\subsection{Verification of Observation~3 of the main text:
rank-deficient recovery, encoder dependence, and joint-signal information}
\label{subsec:supp_rank_deficient_examples}

We now verify Observation~3 of the main text. The same construction
also provides the example anticipated in
Sec.~\ref{subsec:supp_joint_ciphertext}, where the complete signal
register contains more information than the universally accessible
$R$-basis moments.

Take $m=3$, $k=1$, and
\[
A_S=A_{\mathcal N}=0,
\qquad
B=
\begin{pmatrix}
1&0&0\\
0&1&0\\
1&1&0
\end{pmatrix}.
\]
The corresponding cut data are
\begin{equation}
\rank_{\mathbb F_2}B=2,
\qquad
\ker(B^{\mathsf T})
=
\operatorname{span}
\left\{
(1,1,1)^{\mathsf T}
\right\}.
\label{eq:supp_rank_deficient_cut_data}
\end{equation}
Let
\[
w=(1,1,1)^{\mathsf T}.
\]
Since $A_Sw=0$, the unique nontrivial stabilizer supported entirely on
the signal register is
\[
K_S(w)
=
X^wZ^{A_Sw}
=
X^{\otimes3}.
\]
By Lemma~\ref{lem:supp_graph_reduced_spectrum}, the reduced signal
state is therefore
\begin{equation}
\rho_S
=
\frac{1}{8}
\left(
I_S+X^{\otimes3}
\right).
\label{eq:supp_rank_deficient_signal_state}
\end{equation}

For the $YZ/ZY$ encoder class, the exceptional axis is $R=\pm X$.
Hence the nontrivial component in
Eq.~\eqref{eq:supp_rank_deficient_signal_state} is precisely the
complete-sector exceptional direction. Equivalently, the full cut
kernel in Eq.~\eqref{eq:supp_rank_deficient_cut_data} satisfies the
exceptional-subspace condition of Theorem~2. The graph-state resource
therefore supports exact recovery of the input qubit from every
authorized subsystem $S_i\mathcal N$.

This realization also makes explicit the information available under
joint access to all three signals. Consider the $YZ$ representative
and write the input state as
\[
\rho_A
=
\frac{1}{2}
\left(
I_A+r_xX_A+r_yY_A+r_zZ_A
\right).
\]
Direct conjugation by $U_{Y,Z}^{(3,1)}$ gives
\begin{equation}
\mathcal C_S(\rho_A)
=
\frac{1}{8}
\left[
I_S
-r_xX^{\otimes3}
-r_zY^{\otimes3}
+r_yZ^{\otimes3}
\right].
\label{eq:supp_rank_deficient_joint_channel}
\end{equation}
Consequently,
\begin{equation}
\begin{aligned}
\Tr\!\left[
X^{\otimes3}\mathcal C_S(\rho_A)
\right]
&=-r_x,
\\
\Tr\!\left[
Y^{\otimes3}\mathcal C_S(\rho_A)
\right]
&=-r_z,
\\
\Tr\!\left[
Z^{\otimes3}\mathcal C_S(\rho_A)
\right]
&=r_y.
\end{aligned}
\label{eq:supp_rank_deficient_bloch_moments}
\end{equation}
Eq.~\eqref{eq:supp_rank_deficient_bloch_moments} determines all three
Bloch components of the input state. Hence the complete signal channel
in Eq.~\eqref{eq:supp_rank_deficient_joint_channel} is injective on
the input-qubit state space, even though every individual signal is
perfectly concealed. This injectivity concerns state identification
from joint measurement statistics and does not, by itself, imply the
existence of an exact CPTP recovery map acting on the complete signal
register.

The signal--noise entanglement of the resource is
\begin{equation}
S(\rho_S)
=
\rank_{\mathbb F_2}B
=
2
=
(m-1)k.
\label{eq:supp_rank_deficient_entropy_saturation}
\end{equation}
Thus Eq.~\eqref{eq:supp_rank_deficient_entropy_saturation} saturates
the architecture-independent lower bound of Proposition~3. The
missing cut ebit is associated precisely with the exceptional kernel
direction and therefore does not obstruct complementary decoupling.

For the $XZ/ZX$ encoder class, by contrast, the exceptional axis is
$R=\pm Y$. The graph-state criterion would require
\[
(A_S+I_3)w=0.
\]
In the present example,
\begin{equation}
(A_S+I_3)w
=
w
\neq0,
\label{eq:supp_rank_deficient_XZ_failure}
\end{equation}
so the unique nonzero cut-kernel direction is not exceptional.
Moreover, it activates a single sector and therefore violates the
required even-sector-parity condition. Hence the same graph resource,
oriented cut, cut rank, and cut kernel are invalid for the $XZ/ZX$
encoder class.

This comparison is relative to the fixed graph-state Pauli frame.
Ordered pairs of distinct Pauli axes are related by single-qubit
Clifford conjugations. Simultaneous Clifford rotations of the input and
signal systems therefore map between encoder frames and their
compatible resource representatives without changing the
signal--noise entanglement. Observation~3 thus demonstrates the
dependence of a fixed graph-state representative on the chosen encoder
frame, rather than an absolute inequivalence among the Pauli-axis
encoder families.

Finally, the graph can be made connected without changing the
certificate. Adding edges entirely within $\mathcal N$ modifies only
$A_{\mathcal N}$, while leaving $A_S$, $B$, and
$\ker(B^{\mathsf T})$ unchanged. Such modifications therefore preserve
the cut entanglement, the exceptional-kernel condition, and all of the
conclusions above.


\subsection{Representative 24-qubit certificates}
\label{subsec:supp_24q_examples}

We next give the explicit certificates underlying the representative
$24$-qubit examples discussed in the main text. Throughout,
\[
\nu=12,
\qquad
(m,k)=(2,6).
\]
Since $m$ is even, Corollary~4 reduces validity to the existence of a
balanced signal--noise cut satisfying
\begin{equation}
\rank_{\mathbb F_2}B_{S\mathcal N}=12.
\label{eq:supp_24q_full_rank_condition}
\end{equation}

\paragraph{Perfect matching.}
Choose opposite endpoints of each matched edge as signal and noise.
Then
\[
A_S=A_{\mathcal N}=0,
\qquad
B=I_{12},
\]
so Eq.~\eqref{eq:supp_24q_full_rank_condition} is immediately
satisfied. The corresponding graph state is locally Clifford
equivalent to $12$ Bell pairs. For a single matched edge,
\[
CZ\,\ket{+}_S\ket{+}_{\mathcal N}
=
(I_S\otimes H_{\mathcal N})
\ket{\Phi^+}_{S\mathcal N},
\]
and hence, for the complete matching,
\begin{equation}
\ket{G_{\rm match}}
=
\left(
I_S\otimes H_{\mathcal N}^{\otimes12}
\right)
\ket{\Phi_{2^{12}}}_{S\mathcal N},
\label{eq:supp_matching_bell_equivalence}
\end{equation}
where
\[
\ket{\Phi_{2^{12}}}_{S\mathcal N}
=
\bigotimes_{j=1}^{12}
\ket{\Phi^+}_{S_j\mathcal N_j}.
\]
Thus the perfect matching is a graph-state representative of the
canonical Bell-pair resource, up to local Clifford operations.
Consistently, Observation~1 gives
\[
U_{\mathcal N}^{\rm match}
=
H^{\otimes12},
\]
which implements the Bell reduction in
Eq.~\eqref{eq:supp_matching_bell_equivalence}.

\paragraph{Path $P_{24}$.}
Choose
\[
S=(0,2,4,\ldots,22),
\qquad
\mathcal N=(1,3,5,\ldots,23).
\]
The cut matrix is the lower-bidiagonal matrix of
Proposition~\ref{prop:supp_scalable_graphs} and has full rank. Its
inverse transpose is
\[
(M_P)_{ab}
=
\begin{cases}
1,&a\le b,\\
0,&a>b,
\end{cases}
\qquad
M_P=B^{-\mathsf T}.
\]
Observation~1 therefore gives
\begin{equation}
U_{\mathcal N}^{P_{24}}
=
L_{M_P}H^{\otimes12}.
\label{eq:supp_path_bell_reduction}
\end{equation}
A direct implementation of the reversible layer $L_{M_P}$ is the
reverse CNOT ladder
\[
\mathrm{CNOT}_{23\rightarrow21}
\mathrm{CNOT}_{21\rightarrow19}
\cdots
\mathrm{CNOT}_{5\rightarrow3}
\mathrm{CNOT}_{3\rightarrow1}.
\]
Thus the Bell reduction in
Eq.~\eqref{eq:supp_path_bell_reduction} uses $12$ Hadamards and
$11$ CNOTs.

\paragraph{Cycle $C_{24}$.}
Choose
\[
S=(0,3,4,7,8,11,12,15,16,19,20,23)
\]
and
\[
\mathcal N=(1,2,5,6,9,10,13,14,17,18,21,22).
\]
For this cut,
\[
B=I_{12},
\]
and hence Eq.~\eqref{eq:supp_24q_full_rank_condition} holds. The
noise-internal edges are
\[
E_R=
\{(1,2),(5,6),(9,10),(13,14),(17,18),(21,22)\},
\]
while the signal adjacency, transferred to the chosen noise ordering,
gives
\[
E_L=
\{(1,22),(2,5),(6,9),(10,13),(14,17),(18,21)\}.
\]
The noise-side Bell reduction is therefore
\begin{equation}
U_{\mathcal N}^{C_{24}}
=
\left(\prod_{e\in E_L}CZ_e\right)
H^{\otimes12}
\left(\prod_{e\in E_R}CZ_e\right),
\label{eq:supp_cycle_bell_reduction}
\end{equation}
containing $12$ Hadamards and $12$ controlled-$Z$ gates.

\paragraph{Rectangular cluster
$G_{2\times12}^{\mathrm{cl}}$.}
Choose one row as $S$ and the other as $\mathcal N$. Then
\[
B=I_{12},
\qquad
A_S=A_{\mathcal N}=A(P_{12}),
\]
so the cut again satisfies
Eq.~\eqref{eq:supp_24q_full_rank_condition}. Observation~1 gives
\begin{equation}
U_{\mathcal N}^{G_{2\times12}^{\mathrm{cl}}}
=
CZ_{A(P_{12})}
H^{\otimes12}
CZ_{A(P_{12})}.
\label{eq:supp_cluster_bell_reduction}
\end{equation}
The reduction in Eq.~\eqref{eq:supp_cluster_bell_reduction} contains
$12$ Hadamards and $22$ controlled-$Z$ gates.

\paragraph{Sparse Erd\H{o}s--R\'enyi sample.}
For the fixed $G(24,0.35)$ sample generated with seed $7$, a
certifying cut is
\[
S_{\rm ER}
=
\{1,3,4,8,9,11,14,16,17,20,21,22\},
\]
with
\[
\mathcal N_{\rm ER}
=
\{0,2,5,6,7,10,12,13,15,18,19,23\}.
\]
With these orderings,
\begin{equation}
B_{\rm ER}
=
\begin{pmatrix}
1&1&1&0&1&0&1&1&0&0&0&1\\
0&0&0&0&0&0&0&1&0&1&1&1\\
1&0&0&0&1&0&0&1&0&1&1&0\\
0&0&0&1&0&1&0&0&1&0&0&0\\
1&1&0&1&1&0&0&0&1&1&1&0\\
1&0&0&1&0&0&0&0&0&0&0&0\\
0&1&0&1&1&0&0&0&0&0&1&0\\
1&0&1&1&1&1&0&0&1&0&0&0\\
0&0&0&0&1&1&0&0&0&1&1&0\\
0&0&0&0&1&0&1&1&0&1&1&0\\
0&0&0&1&0&0&0&0&0&0&0&0\\
1&0&0&1&0&0&0&0&0&0&1&0
\end{pmatrix},
\qquad
\rank_{\mathbb F_2}B_{\rm ER}=12.
\label{eq:supp_ER_cut_matrix}
\end{equation}
Thus Eq.~\eqref{eq:supp_ER_cut_matrix} provides an explicit
certificate of the full-rank condition
Eq.~\eqref{eq:supp_24q_full_rank_condition}. The corresponding
noise-side canonicalization is
\begin{equation}
U_{\mathcal N}^{\rm ER}
=
CZ_{A_S^{\rm ER}}
L_{B_{\rm ER}^{-\mathsf T}}
H^{\otimes12}
CZ_{A_{\mathcal N}^{\rm ER}}.
\label{eq:supp_ER_bell_reduction}
\end{equation}
The reversible middle layer in
Eq.~\eqref{eq:supp_ER_bell_reduction} can be obtained by the
Gaussian-elimination synthesis described in Appendix~\ref{sec:supp_certification}. Its explicit
inverse matrix is not required for certification.

\paragraph{Excluded examples.}
For the fixed $G(24,0.9)$ sample generated with seed $42$, exhaustive
balanced-cut search gives maximal cut rank $11$. For $K_{24}$ and the
star/GHZ graph, the maximal balanced-cut rank is $1$. Since all these
values are strictly below the requirement in
Eq.~\eqref{eq:supp_24q_full_rank_condition}, none of these examples is
valid for $(m,k)=(2,6)$.

The Erd\H{o}s--R\'enyi results concern only these fixed finite
realizations and do not imply an edge-density threshold. The chosen
sparse sample admits a full-rank balanced cut, whereas the chosen dense
sample does not; no monotonic dependence on edge probability is
inferred.

Among the certified examples, the perfect matching is locally
Clifford equivalent to the canonical Bell-pair resource and, by
Observation~2, is the only one whose explicit noise-side Bell
reduction is local up to permutation. The path, cycle, rectangular
cluster, and sparse random sample require collective Clifford
processing. In every certified case, the complete authorized recovery
can be constructed using Proposition~4.

\subsection{Dicke and W resources}
\label{subsec:supp_dicke_resources}

We finally consider non-graph resources and prove Observation~4 of the
main text. The $n$-qubit Dicke state with $r$ excitations is
~\cite{dicke1954}
\[
\ket{D_n^{(r)}}
=
\binom nr^{-1/2}
\sum_{\substack{z\in\{0,1\}^n\\ |z|=r}}
\ket z.
\]
For the encrypted-cloning resource, set $n=2\nu=2mk.$ Permutation symmetry makes every balanced $\nu|\nu$ bipartition
equivalent. Grouping basis states according to the number $j$ of
excitations in $S$ gives
\begin{equation}
\ket{D_{2\nu}^{(r)}}
=
\sum_{j=j_-}^{j_+}
\sqrt{\lambda_j^{(r)}}\,
\ket{D_\nu^{(j)}}_S
\ket{D_\nu^{(r-j)}}_{\mathcal N},
\label{eq:supp_dicke_schmidt}
\end{equation}
where $j_-=\max\{0,r-\nu\},
j_+=\min\{r,\nu\},$ and
\[
\lambda_j^{(r)}
=
\frac{
\binom{\nu}{j}
\binom{\nu}{r-j}
}{
\binom{2\nu}{r}
}.
\]
The Schmidt vectors for distinct $j$ have different excitation
numbers and are therefore orthogonal. Hence
\begin{equation}
\rank\rho_S
=
j_+-j_-+1
\le
\nu+1,
\qquad
S(\rho_S)
=
H\!\left(\{\lambda_j^{(r)}\}\right)
\le
\log_2(\nu+1).
\label{eq:supp_dicke_rank_entropy}
\end{equation}

We now use these Dicke-state marginals to test the exact resource criterion of Theorem~1, treating the even- and odd-$m$ cases separately.
\begin{proof}[Proof of Observation~4 of the main text]
For even $m$, Theorem~1 requires
\[
\rho_S=\frac{I_S}{2^{mk}}.
\]
By Eq.~\eqref{eq:supp_dicke_rank_entropy}, a Dicke marginal satisfies
\[
\rank\rho_S
\le
mk+1
<
2^{mk},
\]
since $mk\ge2$. It therefore cannot be maximally mixed.

Now let $m\ge3$ be odd. By Theorem~1,
\[
\rho_S
=
\frac{1}{2^{mk}}
\sum_{c\in\mathbb F_2^k}
\alpha_c\,T(c),
\qquad
\alpha_0=1,
\]
with
\[
T(c)
=
\prod_{j=1}^k
\left(
R_{\mathsf S_j}^{\otimes m}
\right)^{c_j}.
\]
Consider two distinct signal qubits $S_{j,a}$ and $S_{j,b}$ in the
same sector. If $c_j=1$, tracing out the remaining $m-2\ge1$ qubits
of that sector removes the term because $\Tr R=0$. If $c_j=0$ but
$c\neq0$, at least one other active sector is traced out completely,
again giving zero. Thus only the identity term survives, and every
valid odd-$m$ resource must satisfy
\begin{equation}
\rho_{S_{j,a}S_{j,b}}
=
\frac{I_4}{4},
\qquad
a\neq b.
\label{eq:supp_required_two_qubit_marginal}
\end{equation}

For a Dicke state with $n=2mk$, permutation symmetry gives
\begin{equation}
\begin{aligned}
\rho_2^{(r)}
={}&
\frac{(n-r)(n-r-1)}{n(n-1)}
\ket{00}\!\bra{00}
+
\frac{r(r-1)}{n(n-1)}
\ket{11}\!\bra{11}
\\
&+
\frac{r(n-r)}{n(n-1)}
(\ket{01}+\ket{10})
(\bra{01}+\bra{10}).
\end{aligned}
\label{eq:supp_dicke_two_qubit_marginal}
\end{equation}
For $0<r<n$, Eq.~\eqref{eq:supp_dicke_two_qubit_marginal} contains a
nonzero
$\ket{01}\!\bra{10}+\ket{10}\!\bra{01}$ contribution and therefore
cannot satisfy Eq.~\eqref{eq:supp_required_two_qubit_marginal}. At
the endpoints $r=0$ and $r=n$, the two-qubit marginal is a pure
product state and again differs from $I_4/4$. Thus no Dicke state
$\ket{D_{2mk}^{(r)}}$ satisfies the exact pure-resource criterion of
Theorem~1 for any $m\ge2$, $k\ge1$, and any sector-wise two-Pauli
encoder. This proves Observation~4.
\end{proof}

Observation~4 is an exact but encoder-specific obstruction.
Independently, Proposition~3 applies to an arbitrary unitary encoder
on $AS$ acting trivially on $\mathcal N$ and requires
\[
S(\rho_S)\ge(m-1)k.
\]
Combining this with Eq.~\eqref{eq:supp_dicke_rank_entropy} gives the
architecture-independent exclusion
\begin{equation}
(m-1)k>\log_2(mk+1)
\quad\Longrightarrow\quad
\text{Dicke resource impossible}.
\label{eq:supp_dicke_entropy_obstruction}
\end{equation}
For fixed excitation number $r$, the sharper necessary condition is
\begin{equation}
(m-1)k
\le
H\!\left(\{\lambda_j^{(r)}\}\right).
\label{eq:supp_dicke_fixed_r_entropy_condition}
\end{equation}
Violation of Eq.~\eqref{eq:supp_dicke_fixed_r_entropy_condition}
excludes the corresponding Dicke resource for any unitary encoder
covered by Proposition~3. Satisfaction of
Eqs.~\eqref{eq:supp_dicke_entropy_obstruction} and
\eqref{eq:supp_dicke_fixed_r_entropy_condition} is only necessary and
does not imply that another encoder succeeds.

For $m=1$, Proposition~1 requires
\[
\rho_S=\frac{I_S}{2^k}.
\]
For $k=1$, the unrestricted task is excluded by the absence of an
$\operatorname{AME}(4,2)$ state. For $k>1$, Eq.~\eqref{eq:supp_dicke_rank_entropy}
gives
\[
\rank\rho_S\le k+1<2^k,
\]
so no Dicke resource has the required maximally mixed
$2^k$-dimensional signal marginal. This is a restriction on the Dicke
family rather than a universal single-output impossibility.

The $W$ state is the single-excitation Dicke state,
\[
\ket{W_{2\nu}}
=
\ket{D_{2\nu}^{(1)}},
\]
with balanced-cut decomposition
\[
\ket{W_{2\nu}}
=
\frac{1}{\sqrt2}
\left(
\ket{W_\nu}_S\ket{0^\nu}_{\mathcal N}
+
\ket{0^\nu}_S\ket{W_\nu}_{\mathcal N}
\right).
\]
Hence
\[
\rank\rho_S=2,
\qquad
S(\rho_S)=1.
\]
Since the $W$ state belongs to the Dicke family, Observation~4 excludes
it for every $m\ge2$ within the fixed two-Pauli architecture.
Independently, Proposition~3 excludes it for an arbitrary unitary
encoder of the form considered there whenever
\begin{equation}
(m-1)k>1.
\label{eq:supp_W_entropy_obstruction}
\end{equation}
The case with $m\ge2$ not excluded by
Eq.~\eqref{eq:supp_W_entropy_obstruction} alone is
\[
(m,k)=(2,1),
\]
which remains excluded for the prescribed two-Pauli encoder by
Observation~4.

The Dicke family is therefore ruled out for two complementary reasons.
First, its signal correlations are incompatible with the restricted
operator structure allowed by the fixed two-Pauli encoder. Second,
independently of that encoder choice, its signal--noise entanglement is
insufficient for exact all-output recovery over a broad range of
parameters.

\end{document}